\documentclass[11pt]{article}

\usepackage[utf8]{inputenc}
\usepackage[T1]{fontenc}
\usepackage[english]{babel}
\usepackage{lmodern}
\usepackage{xcolor}

\usepackage[colorlinks=False]{hyperref}  

\usepackage{comment}

\usepackage{amsmath}  
\usepackage{amsfonts,dsfont}
\usepackage{amssymb}
\usepackage{amsthm}
\usepackage{mathtools}
\usepackage{thmtools}
\usepackage{bm} 
\usepackage{ytableau} 

\numberwithin{equation}{section} 

\newlength{\bredde}
\def\slash#1{\settowidth{\bredde}{$#1$}\ifmmode\,\raisebox{.15ex}{/}
\hspace*{-\bredde} #1\else$\,\raisebox{.15ex}{/}\hspace*{-\bredde} #1$\fi}
\DeclarePairedDelimiter{\floor}{\lfloor}{\rfloor} 

\renewcommand{\epsilon}{\varepsilon} 

\newcommand*{\dif}{\mathop{}\!\mathrm{d}} 

\newcommand{\be}{\begin{equation}}
\newcommand{\ee}{\end{equation}}

\newcommand{\U}{{\rm U}\,}

\usepackage{extpfeil}
\newextarrow{\xrightrightarrows}{{5}{8}{0}{0}}
{\bigRelbar\bigRelbar{\bigtwoarrowsleft\rightarrow\rightarrow}}

\DeclareMathOperator{\sign}{sgn} 
\DeclareMathOperator{\Pf}{Pf} 
\newcommand*{\hconj}[1]{{#1}^{\dagger}} 

\declaretheorem[numberwithin=section]{proposition}
\declaretheorem[numberlike=proposition]{definition/proposition}
\declaretheorem[numberlike=proposition]{theorem}
\declaretheorem[numberlike=proposition]{lemma}
\declaretheorem[numberlike=proposition]{corollary}

\theoremstyle{definition}
\declaretheorem[numberlike=proposition]{example}
\declaretheorem[numberlike=proposition]{definition}

\makeatletter
\newcommand*{\doublerightarrow}[2]{\mathrel{
  \settowidth{\@tempdima}{$\scriptstyle#1$}
  \settowidth{\@tempdimb}{$\scriptstyle#2$}
  \ifdim\@tempdimb>\@tempdima \@tempdima=\@tempdimb\fi
  \mathop{\vcenter{
    \offinterlineskip\ialign{\hbox to\dimexpr\@tempdima+1em{##}\cr
    \rightarrowfill\cr\noalign{\kern.5ex}
    \rightarrowfill\cr}}}\limits^{\!#1}_{\!#2}}}
\makeatother

\begin{document}
\title{
Derivative relations for determinants, Pfaffians and 
characteristic polynomials in random matrix theory
}
\author{{\sc Gernot Akemann$^{1}$, Georg Angermann$^{1}$, Mario Kieburg$^{1,2}$, 
Adrian Padellaro$^{1,3}$}\\~\\
$^1$Faculty of Physics, Bielefeld University, PO-Box 100131, D-33501 Bielefeld, Germany\\
$^2$School of Mathematics and Statistics, University of Melbourne, 813 Swanston Street,\\ Parkville, Melbourne VIC 3010, Australia \\
$^3$International Centre for Theoretical Sciences-TIFR, Shivakote, Hesaraghatta Hobli,\\ Bengaluru North 560089, India
}

\date{}

\maketitle

\begin{abstract}

Explicit expressions are proven for derivatives of the ratio of a
determinant or Pfaffian determinant and a Vandermonde determinant. Such ratios
appear for example in general group integrals of
Harish-Chandra--Itzykson--Zuber type and in expectation values of
products of characteristic polynomials in random matrix theory.
In the latter case we start from known results for general
non-Hermitian and Hermitian ensembles for expectation values without
derivatives, at finite matrix size.
They are given in terms of
the determinant or Pfaffian of the corresponding kernel,
for unitary or orthogonal and symplectic ensembles, respectively.
Several equivalent expressions are proven for general ratios  of
determinants, starting from first order derivatives containing 
the Borel transform of the corresponding matrix or kernel. Higher order
derivatives
are expressed as sums over partitions containing
determinants of derivatives of these,
with coefficients given in terms of combinatorial expressions.
Our most general result is valid for mixed higher order derivatives of
ratios of determinants in several variables.
This generalises previous findings, e.g. for mixed moments in specific
ensembles of random matrices, relevant in applications to the Riemann
$\zeta$-function.  Applications of our results to several examples are
presented, including the complex Ginibre ensemble and the circular unitary ensemble.

\end{abstract}

\noindent
keywords:
derivatives of determinants, characteristic polynomials, non-Hermitian random matrices, Kostka numbers, Borel transform

\begin{center}

{\it Dedicated to Uzy Smilansky on the occasion of his 85th birthday}

\end{center}


\newpage

\section{Introduction}\label{intro}

Determinantal and Pfaffian formulae are ubiquitous in random point
processes in one and two dimensions when they are integrable. This is
partly due to multi-dimensional integration formulas of Andr\'eief \cite{Andreief} and
de Bruijn~\cite{DeBruijn1955}, and their generalisations~\cite{KG09a}, as well as in the
classical group integrals of Harish-Chandra--Itzykson-Zuber~\cite{HC,IZ}, Berezin-Karpelevich~\cite{BK} and
Gelfand-Naimark type~\cite{GN}. 
Our main motivation for this work comes from such expressions
originating from expectation values of products of characteristic
polynomials of random matrices.
Such correlation functions have seen a very rapid development~\cite{AKW22,ABGS21,KW1,KW2,AGKW24,SimmWei,Keating-etal}, driven by
the following applications. Keating and Snaith proposed characteristic
polynomials as a tool in the statistical study of
non-trivial zeros of the Riemann $\zeta$-function~\cite{KS00a} and more
general Dirichlet $L$-functions~\cite{KS00b}. Here, specifically random
matrices from the compact groups ${\rm U}(N)$~\cite{KS00a}, ${\rm USp}(2N)$ and
${\rm O}(N)$~\cite{KS00b} appear, integrated over the respective Haar measure.
    Almost in parallel, Hughes~\cite{Hughes} proposed a combination with
the first derivative of the characteristic polynomial 
to gain further insights.

Secular determinants also appear in the study of chaotic systems, where
the Riemann-zeta function again comes into play. Based on the old
suggestion, that the existence of a Hermitian quantum Hamiltonian may
shed some light on the Riemann hypothesis, it was proposed by Berry
\cite{Berry86} that the Riemann-Siegel formula has a semi-classical
interpretation in quantum chaotic systems, nowadays called
Riemann-Siegel look-alike formula.
It has been prominently advocated by Keating, Smilansky and coworkers
\cite{Keating92,KS99} that such periodic orbit representations would
become key in proving local random matrix statistics, as it was
eventually achieved in \cite{Haakegroup}, see also
\cite{GS2007,Haake2010} for references and a full account.

Expectation values of characteristic polynomials found an
interpretation as partition functions with non-zero quark masses in
Euclidean Quantum Chromodynamics  in 4 (QCD4) \cite{JacEd} and 3 dimensions (QCD3) \cite{JacIsmail}, in an
appropriate low energy limit for the Dirac operator eigenvalues.
In addition to chiral symmetry, again ensembles with unitary, orthogonal
and symplectic symmetry appear, depending on the underlying gauge group
and its representation~\cite{Jac3fold}. Here, typically the average is
taken over Gaussian random matrices with independently distributed
complex, real or quaternion random variables as elements, respectively. However, due
to universality, such restrictions on the choice of distribution become immaterial
in the limit of large matrix dimension $N\to\infty$.
Applications of such massive partition functions include the
derivation of sum rules for low energy Dirac operator eigenvalues in
QCD4 \cite{JacEd} and QCD3 \cite{JacIsmail}, and the generation of spectral correlators from
degenerate masses called replicas, see \cite{DV,ADDV,SV04} and references therein in the context of QCD, as well as \cite{NK01,EK02} in the context of non-Hermitian random matrices. 

Further motivation comes from mathematics and the relation to the
Gaussian free field  and random geometry~\cite{WebbWong}. Here, real
powers of characteristic polynomials are considered, in the
asymptotic large-$N$ limit. A relation to integrable systems
shows up via Painlev\'e functions~\cite{EK02,DeanoSimm}.
To that purpose, the duality between an integer number of
characteristic polynomials and the matrix dimension $N$ plays a crucial
role, cf. \cite{PeterDual} for recent review, which is also known in other contexts due to the supersymmetric
approach to random matrix theory.
It is not easy to give credit to all contributions in this very active
field of research, cf. \cite{Winn,AliAltug,SSD,SerebryakovSimm,Alvarez-Snaith,Cooper-Snaith, SimmWei,Kivimae2025,Keating-etal} and references therein.
We would like to mention, however, a recent branch of random matrix fields where topological properties of the spectra have been analysed~\cite{BHWGG,HKGG1,HKGG2,HKGG3}. In those expressions, one of the authors of the present work came across similar limits of derivatives of ratios of determinants or Pfaffians and  Cauchy determinants. This is due to the consideration of ratios of characteristic polynomials for which such algebraic structures have been already derived before, e.g., see~\cite{BS,KG09a,KG09,AP,Bergere}.

What is the challenge in computing the expectation value of derivatives of
products of characteristic polynomials? The general formulas known from
random matrix theory without derivatives are of the form of a
determinant of the respective kernel of orthogonal polynomials, divided
by a Vandermonde determinant of the arguments of the characteristic
polynomials in the case with unitary symmetry~\cite{BDS,AV03}. For ensembles with
orthogonal or symplectic symmetry we have a Pfaffian determinant
instead, containing the anti-symmetric kernel of skew-orthogonal
polynomials \cite{AB06,KG09,AKP}. It also happens that some mixture of both symmetries might appear, resulting in a Pfaffian~\cite{AN,Mixing,APS}, as it comprises the determinant formulas as well. 
By definition, expected characteristic polynomials as well as their
derivatives are polynomials. However, this is not obvious in terms of
ratios of (Pfaffian) determinants. We will present several strategies to
remove the Vandermonde determinants in the denominator, in order to be
able to obtain polynomial expressions after differentiating these.

An additional motivation for us is the universality of such expressions,
that is the independence of the distribution of random matrix elements
within a given symmetry class. It is expected to hold when taking the
large-$N$ limit. However, most of the expressions derived for
derivatives so far are for specific ensembles, e.g., in the Riemann
$\zeta$-function context  for the Circular Unitary Ensemble (CUE) over
${\rm U}(N)$ (see~\cite{SimmWei,GMN}).
One of our main goals is thus to derive formulae that hold for a given
 symmetry class (unitary, orthogonal or symplectic) at finite-$N$,
for a general distribution of matrix elements,
specified by the weight for the (skew-) orthogonal polynomials. Such
expressions then allow to directly
take the corresponding large-$N$ limits on the respective (skew-)
kernels.

The present work is organised as follows. In 
Subsection~\ref{subsec:setup} we briefly describe the setup for the
application of our results to  expectation values of products of
characteristic polynomials in random matrix theory. Further details are
given in Appendix~\ref{App:results-OUSE}. The reader familiar with this setup 
may skip directly to Section~\ref{sec:DdetPf}, where our main results
are presented. Subsection~\ref{subsec:GenMixed} starts with the
most general result, for mixed derivatives of different degrees at
different arguments.
Theorem~\ref{thm:main.result} provides derivatives of different order at
different points of a Pfaffian over a Vandermonde determinant. It is given by the product
of certain differential operators acting on the Pfaffian of transformed
matrix elements, which is no longer divided by a Vandermonde. Corollary~\ref{cor:special results} yields the corresponding result for a determinant instead of a Pfaffian. These results are general and even apply to expressions which may not be related to averages of characteristic polynomials at all.

In Subsection~\ref{subsec:Borel}, we consider cases that correspond  to
mixed moments of characteristic polynomials of zero-th and first
derivative, all at the same argument.
In Corollary~\ref{cor: factoring derivatives} our main result
applied to the unitary class simplifies to ordinary derivatives acting
on the determinant of a Borel transformed kernel.
These derivatives can be carried out, resulting to a double sum over
determinants of derivatives of the new kernel.
A similar Corollary~\ref{thm: factoring derivatives pfaffian} gives the
corresponding result for the class of Pfaffian determinants.
These corollaries generalise existing results in the literature to
general kernel functions.

Subsection~\ref{subsec:Kostka} is devoted to mixed derivatives
of general, higher order, at one or two points. The derivation follows a
different, representation theoretic/combinatorial route, using Schur functions and
Kostka numbers.
Theorem~\ref{Thm-gen-det} gives the determinantal case, while~\ref{Thm-gen-pfaffian} and~\ref{Thm-gen-symplectic-abs} give the
Pfaffian cases for one and two points, respectively.

In Section~\ref{sec:App-mixed}, we give two examples for the application 
to characteristic polynomials. The infinite dimensional limit of the complex Ginibre ensemble~\cite{Ginibre} is discussed in Subsection~\ref{subsec:Gin} where several statements are proven in Appendix~\ref{apx:GinUE}.  In Subsection~\ref{subsec:CUE},  we discuss the case of the CUE, where the known generalised Bessel-kernel is understood from the
Borel transformation of the initial kernel.
Section~\ref{sec:proofs} collects the proofs of our main theorems and in
Section~\ref{conclusio} we summarise our findings and present some open
problems.

\subsection{Motivation: averages of characteristic polynomials of
random matrices}\label{subsec:setup}

In this subsection, we present the setup for computing expectation values of characteristic polynomials and their derivatives, focusing on non-Hermitian random matrices, which serves as our motivation for the results we have found. Let us again emphasise that the results we present do not depend on the details of the kernel of the determinant or Pfaffian considered. They are not restricted to averages of characteristic polynomials. Thus, they can equally apply to elliptic ensembles, chiral symmetry breaking ensembles, sums and/or products of random matrices. 

Let us turn back to averages of characteristic polynomials and their derivatives. Below we will work with the point processes of random matrix eigenvalues rather than with the matrix representation, where characteristic polynomials are given in terms of determinants. The reason is convenience, in keeping track of normalisation constants. 

We begin with the partition function $Z_N$ or normalising constant, that is given by an $N$-fold integral over the complex eigenvalues $z_1,\ldots,z_N$ of the joint eigenvalue distribution function (jpdf) $\mathcal{P}_N$ in the corresponding ensemble,
\begin{equation}
\label{partZN}
Z_N:=\prod_{j=1}^N\int_\mathbb{C}d^2z_j \mathcal{P}_N(z_1,\ldots,z_N)\ .
\end{equation}
Expectations values of a function $\mathcal{O}:\mathbb{C}^N\to\mathbb{C}$ of the eigenvalues  will be denoted by
\begin{equation}
\label{vev}
\langle \mathcal{O}\rangle:= \frac{1}{Z_N}\prod_{j=1}^N\int_\mathbb{C}d^2z_j \mathcal{P}_N(z_1,\ldots,z_N) \mathcal{O}(z_1,\ldots,z_N)\ .
\end{equation}
The most relevant example for our motivation is the characteristic polynomial $\mathcal{O} = \det(\chi\mathbf{1} - J)$ leading to
\begin{equation}
\label{DNdef}
    D_N(\chi):=\prod_{j=1}^{N} (\chi-z_j)\ , \quad \chi\in\mathbb{C}\ ,
\end{equation}
where $z_j$ are the complex eigenvalues of the random matrix $J$.
We denote its first respectively $p$th order derivative by $D_N(\chi)^\prime$ respectively $\partial_\chi^pD_N(\chi)$. 
A popular object of study are so-called mixed moments, for example
\begin{equation}
\label{AN1def}
    \left\langle |D_N(\chi)|^{2h} |D_N(\chi)^\prime|^{2k-2h} \right\rangle\ , 
\end{equation}
where the powers take nonnegative integer values $h=0,1,\ldots$ and $k=h,h+1,\ldots$ only.
This can be generalised to expectation values of higher derivatives 
\begin{equation}
    \label{1ptGenDerivs}
    \left\langle D_N(\chi)^{m(0)} (D_{N}(\chi)')^{m(1)} \ldots (\partial_{\chi}^{d} D_N(\chi))^{m(d)} \right\rangle, 
\end{equation}
where the multiplicities $m(p)$ determine the power of the $p$th derivative of the characteristic polynomial.
In Section~\ref{subsec:GenMixed} we further generalise to the case of several points $\chi_1,\dots,\chi_L$.

Let us turn to the random matrix ensembles for which we will compute the above quantities. 
We call those ensembles, where the jpdf is given by
\begin{equation}
\label{PNU}
\mathcal{P}^{\rm (U)}_N(\vec z):= \frac{|\Delta_N(\vec z)|^2 \prod_{j=1}^N w(z_j)}{N!\prod_{j=0}^{N-1}h_j}\quad{\rm with}\ \vec z=(z_1,\ldots,z_N),
\end{equation}
 ensembles with unitary symmetry (U).
Typically, the Vandermonde determinant $\Delta_N$ results from the Jacobian of a diagonalising unitary transformation.
However, there are ensembles for which no matrix representation is available, see e.g. \cite{AEP,KK}. 
The Vandermonde determinant as a function of the set of $N$ variables $\vec z=(z_1,\ldots,z_N),$ is defined as
\begin{equation}
\label{Vander}
\Delta_N(\vec z):=\prod_{1\leq k<j\leq N} (z_j-z_k)=\det\left[z_j^{k-1}\right]_{j,k=1}^N\ .
\end{equation}
For $N=0,1$, here and throughout this work the empty product is defined as $1$.
For the weight function $w(z)$ on the complex plane (that may be supported on part of $\mathbb{C}$ only)  we assume that all moments exist, 
\begin{equation}
\label{mom}
\int_\mathbb{C}d^2z_j w(z) z^j \bar{z}^k<\infty\quad \mbox{for}\ j,k=0,1,\ldots
\end{equation}
Here, $\bar{z}$ denotes the complex conjugate of $z$. 
Therefore, the corresponding planar orthogonal polynomials $P_j(z)$ exist and can be constructed via Gram--Schmidt, and their corresponding polynomial kernel is denoted by $\mathfrak{K}_N(z,\bar{z})$, see Appendix \ref{sec:AppA} for more details.
Examples within this symmetry class include, for example,  the complex Ginibre ensemble, its generalisations, the chiral counterpart, or truncated unitary random matrices. For more examples and references see the recent volume~\cite{PeterSungsoo}.

In the ensembles with unitary symmetry, expectation values of products of characteristic polynomials follow from the theorem in~\cite{AV03}
\begin{equation}
\label{productUk}
\left\langle \prod_{j=1}^kD_N(\chi_j) \overline{D_N(\xi_j)}\right\rangle_{\rm U}=C_{N,k}^{\rm (U)}
\frac{\det[\mathfrak{K}_{N+k}(\chi_a,\bar{\xi}_b)]_{a,b=1,\ldots,k}}{\Delta_k(\vec\chi)\Delta_k(\vec{\bar{\chi}})}\ .
\end{equation} 
In Appendix \ref{sec:AppA} the constant $C_{N,k}^{\rm (U)}$ and  the kernel $\mathfrak{K}_{N+k}$ of planar orthogonal polynomials is given, and the more general statement for an unequal number of $k$ characteristic polynomials and $\ell$ conjugated ones with $\ell\leq k$ is reviewed. 
In the Hermitian limit these are no longer distinct, leading to several equivalent identities, depending on how the product is split into $k$ and $\ell$ factors, see \cite{AV03}.

Clearly, to compute the mixed moments in \eqref{1ptGenDerivs} (or more generally in \eqref{LptGenDerivs} below) in the unitary symmetry class, a number of derivatives of the expression \eqref{productUk} has to be taken with respect to a subset of the variables $\chi_j$ and $\bar{\xi}_j$, before setting them equal, $\chi_j=\xi_j=\chi$. 
This can be achieved when applying our results to be presented in the next Section \ref{sec:DdetPf}.

We turn to non-Hermitian 
random matrix ensembles with symplectic symmetry (S). They are defined by having the corresponding jpdf given by \cite{Ginibre,PeterSungsoo}
\begin{equation}
\label{PNS}
\mathcal{P}^{\rm (S)}_N(z_1,\ldots,z_N):=
\prod_{1\leq k<j\leq N} |z_j-z_k|^2|z_j-\bar{z}_k|^2\prod_{j=1}^N w(z_j)|z_j-\bar{z}_j|^2=|\Delta_{2N}(\vec{z},\vec{\bar{z}})|^2 \prod_{j=1}^N w(z_j).
\end{equation} 
Here, the eigenvalues of a matrix $J$ with $N^2$ quaternion entries come in complex conjugated pairs, and the jpdf is proportional to the squared modulus of a Vandermonde determinant of size $2N$ \cite{Kanzieper01}.
For this reason, in the symplectic ensembles a single characteristic polynomial is always of the following product form\footnote{For a random matrix $J$ with quaternionic elements of size $N\times N$, in complex representation of size $2N\times 2N$, 
we have 
$\det[\chi-J]=\prod_{j=1}^{N} (\chi-z_j)(\chi-\bar{z}_j)$ in terms of the $N$ complex conjugate eigenvalue pairs $(z_j,\bar{z}_j)$.}
\begin{equation}
\label{DN-SE}
D_N(\chi) = \prod_{j=1}^{N} (\chi-z_j)(\chi-\bar{z}_j)\ .
\end{equation}
This is in contrast to the single product~\eqref{DNdef} in the unitary ensembles. 
In particular, here the product~\eqref{DN-SE} is real for $\chi\in\mathbb{R}$. 
Imposing the existence of all moments in these ensembles gives rise to planar skew-orthogonal polynomials $q_j$ that can be constructed in analogy to Gram-Schmidt, see~\cite{AKP, AEP}.
The corresponding polynomial skew-kernel which is anti-symmetric is then denoted by $\kappa_N(u,v)=-\kappa_N(v,u)$, and we refer again to Appendix~\ref{sec:AppA} for more details, including the skew-norms $r_j^{\rm (S)}$. The expectation value of products of the form~\eqref{DN-SE} has been computed in~\cite{AB06} and reads for $k$ even
\begin{equation}
\label{productSk1}
\left\langle \prod_{j=1}^k
    D_N(\chi_j)
\right\rangle_{\rm S}=
\left\langle \prod_{j=1}^k
    \left(\prod_{l=1}^N (\chi_j - z_l)(\chi_j - \bar{z}_l)\right)
\right\rangle_{\rm S}=
C_{N,k}^{\rm (S)}
\frac{\Pf[\kappa_{N+k/2}(\chi_a,{\chi}_b)]_{a,b=1,\ldots,k}}{\Delta_k(\vec\chi)}\ .
\end{equation} 
For the formula with an odd number $k$ of products see Appendix  \ref{sec:AppA}, where an extra row and column appears. Notice that this expectation value is holomorphic in all variables $\chi_j$. 
A formula with the same structure holds for ensembles with orthogonal symmetry, see~\cite{AKP}. The only difference is that because of the characteristic polynomial having real coefficients, on the left-hand side the eigenvalues come in complex conjugate pairs or are real.
Once again it is clear how to obtain the derivative of such characteristic polynomials on the left-hand side, and for this operation on the right-hand side our result in Section  \ref{sec:DdetPf} applies.

When choosing $k=2m$ to be even and $\chi_{j+m}=\bar{\chi}_j$ for $j=1,\ldots,m$, we obtain as a corollary of~\eqref{productSk1} an expression for the product of the modulus square of characteristic polynomials, valid for any $m$ even or odd~\cite{A05}. 
In particular, for $m\in\mathbb{N}$ and $\chi_1,\ldots,\chi_m\in \mathbb{C}$ pairwise distinct, it holds for the orthogonal respectively symplectic class ${\rm C=O,S}$:
\begin{equation}
\label{productSk2}
\left\langle \prod_{j=1}^{m}|D_N(\chi_j)|^2\right\rangle_{\rm C}=
\frac{(-1)^{m}\prod_{j=N}^{N-1+m}r_j^{\rm (C)} 
}{\Delta_{2m}(\vec\chi,\vec{\bar{\chi}})}
\Pf
\begin{bmatrix}
\kappa_{N+m}^{\rm (C)} (\chi_a,{\chi}_b)&\kappa_{N+m}^{\rm (C)} (\chi_a,\bar{\chi}_b)\\
\kappa_{N+m}^{\rm (C)} (\bar{\chi}_a,{\chi}_b) & \kappa_{N+m}^{\rm (C)} (\bar{\chi}_a,\bar{\chi}_b)\\
\end{bmatrix}
_{a,b=1,\ldots,m}\ .
\end{equation}
Derivatives of this equation are a special case of our general expression for products of derivatives, \eqref{LptGenDerivs} below, with $L=2m$, $d_l=m_l(1)=1$ and $m_l(0)=0$ for all $l=1,\dots,L$, to which our main result \eqref{main.result} applies. 

We would like to mention that formulas like~\eqref{productSk1} can be also found for the orthogonal class (O) where one considers real matrices. Then, the eigenvalues $z_j$ can be real or come in complex conjugate pairs. Even more generally such Pfaffian structures are not alone restricted to the symplectic or orthogonal class. In lattice Quantum Chromodynamics two of the authors came across ensembles that have a unitary symmetry group but exhibit Pfaffian structures as well, see~\cite{AN,KVZ}. In this case a mixture of orthogonal and skew-orthogonal polynomials appear~\cite{Mixing}. General expressions for averages of products and even ratios of characteristic polynomials of even larger classes have been derived in~\cite{BS,KG09a,KG09}, where very similar expressions as in~\eqref{productUk} and~\eqref{productSk1} hold.

\section{Main results: Derivative relations for determinants and Pfaffians}\label{sec:DdetPf}

Our main results are presented in three different subsections.
Subsection \ref{subsec:GenMixed} starts with the most general result for general higher order derivatives at different points, acting on a Pfaffian determinant of an anti-symmetric matrix over a Vandermonde determinant. It results in a differential operator acting on a Pfaffian of the integral transformed matrix. Setting certain 
matrix elements to zero immediately implies the same result for determinants.
Next, Subsection \ref{subsec:Borel} applies to first derivatives at a single point only, both for determinants respectively Pfaffians. It it is taylored to later applications to characteristic polynomials, for general weight functions,  and contains the Borel transform respectivly integral transform of the corresponding matrix (kernel).
Subsection \ref{subsec:Kostka} provides formulae of combinatorial nature, 
valid for higher order derivatives of both determinants and Pfaffians at a single (pair) of point(s). It is given by the sum over partitions of (Pfaffian) determinants of derivatives of the respective kernels. Here, the coefficients contain so-called Kostka numbers, e.g., see Ref.~\cite{FultonHarris2004}, which are the expansion coefficients of Schur polynomials. This approach generalises existing results for the first derivative~\cite{SimmWei,Wei-unpub}, containing the hook length of the Young diagram of the corresponding partitions.
In the special case of the CUE which is determinantal, Kostka numbers have been shown to appear as well in most recent results parallel to ours \cite{GMN}.

\subsection{General mixed derivatives for determinants and Pfaffians}\label{subsec:GenMixed}

We first want to present a general theorem that is rather versatile and should have several applications to general expectation values of mixed derivatives of characteristic polynomials in determinantal and Pfaffian point processes.
Especially, it could be used to deal with the more general partition function
\begin{equation}
    \label{LptGenDerivs}
    \mathcal{Z}_{d,L}(\vec \chi;{\bf m})=\left\langle
    \prod_{l=1}^L  D_N(\chi_l)^{m_{l,0}} (\partial_{\chi_l}D_{N}(\chi_l))^{m_{l,1}} \cdots (\partial_{\chi_l}^dD_N(\chi_l))^{m_{l,d}}
    \right\rangle,
\end{equation}
which generalises~\eqref{1ptGenDerivs} to multiple points.
Here, we employ the notation $\vec \chi = (\chi_1,\dots,\chi_L)$  and ${\bf m}=\{m_{l,j}\}_{\substack{l=1,\ldots, L\\j=0,\ldots,d}}$ collecting the multiplicities $m_{l,j}$ of the $d_j$th derivative at the point $\chi_l$.
Of the matrix ${\bf m}$, the $l$-th row is summarised as $\vec{m}_l$.
We would like to stress that the following results actually apply to more general settings, not necessarily related to moments of characteristic polynomials.

To set up the approach, we start with $P$ pairwise distinct variables $\vec x=(x_1,\ldots,x_P)\in\mathbb{C}^P$ of which we will take derivatives and $L$ pairwise distinct limiting points $\vec \chi = (\chi_1,\dots,\chi_L) \in \mathbb{C}^L$.
The limit, for which we use the shorthand $\vec x \to \vec \chi$, is encoded in a surjective function $\psi: \{x_1,\dots,x_P\} \rightarrow \{\chi_1,\dots,\chi_L\}$ and really means the limit $\vec x \rightarrow (\psi(x_1),\dots,\psi(x_P))$.
As we now explain, the multiplicity labels $m_{l,j}$ in \eqref{LptGenDerivs} appear as a consequence of labelling the points $x_i$ by how they appear in inverse images of $\psi$.
To see this, consider labelling the elements of the inverse image of $\chi_l$ as follows
\begin{equation}
    \{x_{l,1}, \dots, x_{l,P_l}\} = \psi^{-1}(\chi_l),
    \label{eq: paramPartition1}
\end{equation}
where $P_l = |\psi^{-1}(\chi_l)|$ is the cardinality of the inverse image. The relation between doubly-labelled and single-labelled variables
is
\begin{equation}
    x_{1,1}=x_1,\ldots,x_{1,P_1}=x_{P_1},x_{2,1}=x_{P_1+1},\ldots,\ldots,x_{L,P_L} = x_{P_1+\ldots+P_L} = x_{P}.
\end{equation}
The inverse images partition the set of pair-wise distinct points $\{x_i\}_{i=1}^P$.
Therefore, 
\begin{equation}
    \{x_i\}_{i=1}^P = \bigcup_{l=1}^L \psi^{-1}(\chi_l), \quad \sum_{l=1}^L P_l = P.
\end{equation}
In particular, this gives a correspondence between the labelling $x_i$ and the labelling $x_{l,j}$.
As a consequence, a similar correspondence is induced for differential operators in these variables.
Therefore, for any family of non-negative integers $n_{l,j}$, the differential operator
\begin{equation}
    \prod_{l=1}^L \prod_{j=1}^{P_l} \partial_{x_{l,j}}^{n_{l,j}},
    \label{eq: inv img diff op}
\end{equation}
can be understood as a differential operator in $x_i$.
This operator plays an important role in Theorem \ref{thm:main.result}, where both $x_i$ and $x_{l,j}$ labels are used to state and prove the theorem.
Since this is just a re-labelling the above operator should be understood as
\begin{equation}
    \prod_{i=1}^P \partial_{x_i}^{n_i},
\end{equation}
for some family of integers $n_i$ determined by the integers $n_{l,j}$.
For example, consider the case where $P=3, L=2$ so that $\vec x = (x_1,x_2,x_3)$, $\vec \chi = (\chi_1,\chi_2)$ and $\psi(x_1) = \psi(x_2) = \chi_1$, $\psi(x_3) = \chi_2$.
We can label each inverse image in two ways $\psi^{-1}(\chi_1) = (x_1, x_2) = (x_{1,1}, x_{1,2})$ and $\psi^{-1}(\chi_2) = x_3 = x_{2,1}$.
In this case, the identification gives $n_1 = n_{1,1}$, $n_2 = n_{1,2}$ and $n_3 = n_{2,1}$.
With this notation, the partition function \eqref{LptGenDerivs} can be written as
\begin{equation}
    \label{LptGenDerivsV}
    \mathcal{Z}_{d,L}(\vec \chi;{\bf m}) =
    \lim_{\vec x \to \vec \chi}
     \prod_{l=1}^L \prod_{j=1}^{P_l} \partial_{x_{l,j}}^{n_{l,j}} \left\langle \prod_{i=1}^P D_N(x_i) \right\rangle \coloneqq \lim_{\vec x \to \vec \chi} \left(\prod_{i=1}^P  \partial_{x_i}^{n_i}\right )
      \left\langle \prod_{i=1}^P D_N(x_i) \right\rangle .
\end{equation}
The multiplicities are determined by
\begin{equation}
    m_{l,d} = \sum_{i=1}^{P_l} \delta_{d,n_{l,i}}.
\end{equation}
We emphasise that it could be that some $m_{l,d}$ vanish.

In preparation of the Theorem~\ref{thm:main.result}, we show the following lemma.

\begin{lemma}\label{lemma:der.rel}

Let $d,k,N \in \mathbb{N}$ be integers satisfying $d\geq k\geq1$ and $N>1$.
Introduce points $x,z_1,\ldots,z_N,\zeta_1,\ldots, \zeta_N\in\mathbb{C}$ and auxiliary variables
\begin{equation}
  \vec u =  (u_{1},\ldots, u_{d}).
\end{equation}
We define the function
\begin{equation}\label{Fd.def}
    F_d(\vec u,\chi,\zeta)=\sum_{l=1}^{d}(l-1)!\frac{ u_{l}}{(\zeta-\chi)^l}.
\end{equation}
Then, there exists a unique differential operator $D_{\vec u,k}$, which is a polynomial in $\partial_{ u_1},\ldots,\partial_{ u_k}$, only, satisfying
the identity
\begin{equation}\label{der.rel}
    \partial_{x}^k\prod_{j=1}^N\frac{z_j-x}{\zeta_j-x}=\lim_{\vec u \to 0}D_{\vec u ,k}\exp\left[\sum_{j=1}^N \bigl(F_{d}(\vec u,x,\zeta_j)-F_{d}(\vec u,x,z_j)\bigl)\right]\prod_{j=1}^N\frac{z_j-x}{\zeta_j-x}.
\end{equation}
The operators $D_{\vec u,k}$ admit the scaling relation
\begin{equation}\label{der.scaling}
   D_{(su_1,s^2u_2,\ldots,s^du_d),k}=s^{-k}D_{(u_1,u_2,\ldots,u_d),k}
\end{equation}
and the recurrence relation
\begin{equation}\label{recurrence.diff}
\lim_{\vec u \to 0}D_{\vec u,k}=\lim_{\vec u \to 0}D_{\vec u,k-1}\left[\partial_{ u_1}+\sum_{l=1}^{k-1} u_l\partial_{ u_{l+1}}\right]
\end{equation}
while for $k=1$ it is simply
\begin{equation}\label{D.k1}
D_{\vec u,1}=\partial_{ u_1}.
\end{equation}
\end{lemma} 

\begin{proof}
We prove this statement by induction. It is clear that for $k=1$, we have
\begin{equation}
 \partial_{x}\prod_{j=1}^N\frac{z_j-x}{\zeta_j-x}
 =\sum_{j=1}^N \left(\frac{1}{\zeta_j-x}-\frac{1}{z_j-x}\right)
 \prod_{j=1}^N\frac{z_j-x}{\zeta_j-x}
 =\sum_{j=1}^N 
 \partial_{ u_1}(F_{d}(\vec u,x,\zeta_j)-F_{d}(\vec u,x,z_j))
 \prod_{j=1}^N\frac{z_j-x}{\zeta_j-x},
\end{equation}
which immediately yields the statement and~\eqref{D.k1}.

For $d\geq k>1$, we assume that $D_{\vec u,k-1}$ is indeed unique and only given by a polynomial of $\partial_{ u_1},\ldots,\partial_{ u_{k-1}}$ which means that $F_d$ can be restricted to $F_{k-1}$ only, i.e.,
\begin{equation}
\begin{split}
&\lim_{\vec u \to 0}D_{\vec u ,k-1}\exp\left[\sum_{j=1}^N \bigl(F_{d}(\vec u,x,\zeta_j)-F_{d}(\vec u,x,z_j)\bigl)\right]\prod_{j=1}^N\frac{z_j-x}{\zeta_j-x}\\
=&\lim_{\vec u \to 0}D_{\vec u ,k-1}\exp\left[\sum_{j=1}^N \bigl(F_{k-1}(\vec u,x,\zeta_j)-F_{k-1}(\vec u,x,z_j)\bigl)\right]\prod_{j=1}^N\frac{z_j-x}{\zeta_j-x}.
\end{split}
\end{equation}
Furthermore, we note the commutation relations $\partial_xD_{\vec u,k-1}=D_{\vec u,k-1}\partial_x$ and $ u_kD_{\vec u,k-1}=D_{\vec u,k-1} u_k$
 as well as the simple identity
\begin{equation}
\partial_x\bigl(F_{k-1}(\vec u,x,\zeta_j)-F_{k-1}(\vec u,x,z_j)\bigl)=\sum_{l=1}^{k-1} u_l\partial_{ u_{l+1}}\bigl(F_{k}(\vec u,x,\zeta_j)-F_{k}(\vec u,x,z_j)\bigl).
\end{equation}
Collecting everything, we have
\begin{equation}
\begin{split}
    \partial_{x}^k\prod_{j=1}^N\frac{z_j-x}{\zeta_j-x}=&\lim_{\vec{u} \to 0}D_{\vec u ,k-1}\partial_x\exp\left[\sum_{j=1}^N \bigl(F_{k-1}(\vec u,x,\zeta_j)-F_{k-1}(\vec u,x,z_j)\bigl)\right]\prod_{j=1}^N\frac{z_j-x}{\zeta_j-x}\\
    =&\lim_{\vec u \to 0}D_{\vec u ,k-1}\left[\partial_{ u_1}+\sum_{l=1}^{k-1} u_l\partial_{ u_{l+1}}\right]\exp\left[\sum_{j=1}^N \bigl(F_{k}(\vec u,x,\zeta_j)-F_{k}(\vec u,x,z_j)\bigl)\right]\prod_{j=1}^N\frac{z_j-x}{\zeta_j-x}
\end{split}
\end{equation}
yielding the recurrence relation~\eqref{recurrence.diff}.
Commuting $D_{\vec u ,k-1}$ with $\partial_{ u_1}+\sum_{l=1}^{k-1} u_l\partial_{ u_{l+1}}$ under the limit $\vec u\to0$ yields the claim that $D_{\vec u ,k}$ is indeed a polynomial of the partial derivatives  $\partial_{ u_1},\ldots,\partial_{ u_k}$ alone.
The scaling relation~\eqref{der.scaling} is a direct consequence of the recurrence relation~\eqref{recurrence.diff} where one rescales $u_j\to s^j u_j$.
\end{proof}

To illustrate what the differential operators $D_{\vec{u} ,k}$ explicitly look like,  we give
\begin{equation}
\begin{split}
k=2:&\quad D_{\vec u,2}=\partial_{ u_1}^2+\partial_{ u_2},\\
k=3:&\quad D_{\vec u,3}=\partial_{ u_1}^3+3\partial_{ u_2}\partial_{ u_1}+\partial_{ u_3},\\
k=4:&\quad D_{\vec u,4}=\partial_{ u_1}^4+6\partial_{ u_2}\partial_{ u_1}^2+3\partial_{ u_2}^2+4\partial_{ u_3}\partial_{ u_1}+\partial_{ u_4}.
\end{split}
\end{equation}
As we prove in Appendix \ref{apx: bell poly}, the derivative operators 
$D_{\vec{u} ,k}$ 
are given by $k$th complete Bell polynomials \cite{Bell34}\footnote{Thanks to Jeanne Scott for noticing this connection.}
\begin{equation}
    D_{\vec u, k} = B_k(\partial_{u_1},\dots,\partial_{u_k}) \, .
\end{equation}

As a final ingredient we define an integral transformation of a function $f$ that shall be analytic in some neighbourhoods of each of the points $\chi_1,\ldots,\chi_L$. This transform is
\begin{equation}\label{transform}
    [\mathcal{K}_{\vec{\chi},\alpha} f](\boldsymbol{u})=\oint_{\mathcal{C}}\frac{\dif\zeta\, f(\zeta)}{2\pi i }\exp\left[\sum_{l=1}^LF_{d}(\vec u_l,\chi_l,\zeta)\right]\frac{\zeta^{P-\alpha}}{\prod_{l=1}^L(\zeta-\chi_l)^{P_l}},
\end{equation}
where $\vec u_l=( u_{l ,1}, \dots, u_{l, d})$, $\boldsymbol{ u}=\{ u_{l ,j}\}_{\substack{l=1,\ldots,L\\j=1,\ldots,d}}$ and $\mathcal{C}$ is a contour encircling counter-clockwise each $\chi_l$ close enough, so that it lies in the domain where $f$ is analytic. 
For example, with $L=2, d=1$ the transformation is explicitly given by
\begin{equation}
    \label{eq: L=2 Ktransform}
    [\mathcal{K}_{\vec{\chi},\alpha} f]( u_{1,1}, u_{2,1})=\left(\oint_{|\zeta - \chi_1| = \varepsilon} +\oint_{|\zeta - \chi_2| = \varepsilon} \right)\frac{\dif\zeta\, f(\zeta)}{2\pi i }\exp\left[
    \frac{ u_{1,1}}{(\zeta - \chi_1)} + \frac{ u_{2,1}}{(\zeta - \chi_2)}
    \right]\frac{\zeta^{P-\alpha}}{(\zeta-\chi_1)^{P_1}(\zeta-\chi_2)^{P_2}}.
\end{equation}
Moreover, we define $\mathcal{K}_{\vec{\chi},\alpha}\otimes\mathcal{K}_{\vec{\chi},\gamma} A$ for a function $A$ of two entries. It means that this transformation is applied in each entry, separately.

\begin{theorem}[Higher order derivatives of Pfaffians]\label{thm:main.result}\

We choose the above-mentioned partition of the complex variables $\vec x\in\mathbb{C}^P$ and their limits $\vec \chi\in\mathbb{C}^L$. Let $P,Q\in\mathbb{N}_0$ such that $P+Q \in 2\mathbb{N}$ is positive and even, $C=-C^T\in{\rm Asym}_{\mathbb{C}}(Q)$ a fixed antisymmetric complex $Q \times Q$ matrix, and consider the functions $A:\mathbb{C}^2\to\mathbb{C}$ and $B_1,\ldots,B_Q:\mathbb{C}\to\mathbb{C}$ which should be analytic at the points $\chi_l$ and $A(x_1,x_2)=-A(x_2,x_1)$ is skew-symmetric. Then, it is
\begin{equation}\label{main.result}
\begin{split}
&\lim_{\vec x\to \vec\chi}\left(\prod_{l=1}^L\prod_{j=1}^{P_l}\partial_{x_{l,j}}^{n_{l,j}}\right) \frac{1}{\Delta_P(\vec x)}{\rm Pf}\left[\begin{array}{cc} A(x_a,x_c) & B_d(x_a) \\ -B_b(x_c) & C_{bd} \end{array}\right]_{\substack{a,c=1,\ldots,P\\b,d=1,\ldots,Q}}\\
    =&\lim_{\boldsymbol{u}\to 0}\left(\prod_{l=1}^L\prod_{k=1}^{d}D_{\vec u_l,k}^{m_{l,k}}\right){\rm Pf}\left[\begin{array}{cc} [\mathcal{K}_{\vec{\chi},\alpha}\otimes\mathcal{K}_{\vec{\chi},\gamma}A](\boldsymbol{u},\boldsymbol{u}) & [\mathcal{K}_{\vec{\chi},\alpha}B_d](\boldsymbol{u}) \\ -[\mathcal{K}_{\vec{\chi},\gamma}B_b](\boldsymbol{u}) & C_{bd} \end{array}\right]_{\substack{\alpha,\gamma=1,\ldots,P\\b,d=1,\ldots,Q}}.
\end{split}
\end{equation}
\end{theorem}

This theorem is proven in Sec.~\ref{sec:proof.main}. It has several immediate consequences. For instance, when we deal only with a determinant instead of a Pfaffian we have the following two special cases. They do not warrant any proof as they are purely restrictions or multiple applications of the general theorem.

\begin{corollary}\label{cor:special results}\

\begin{enumerate}
\item	In the case $P=Q$, $A=0$ and $C=0$, Theorem~\eqref{thm:main.result} reduces to
\begin{equation}\label{special.result.1}
\begin{split}
\hspace*{-0.5cm} \lim_{\vec x\to \vec \chi}\left(\prod_{l=1}^L\prod_{j=1}^{P_l}\partial_{x_{l,j}}^{n_{l,j}}\right) \frac{1}{\Delta_P(\vec x)}\det\left[ B_d(x_a) \right]_{\substack{c,d=1,\ldots,P}}=\lim_{\boldsymbol{u}\to 0}\left(\prod_{l=1}^L\prod_{k=1}^{d}D_{\vec u_l,k}^{m_{l,k}}\right)\det\left[ [\mathcal{K}_{\vec{\chi},\alpha}B_d](\boldsymbol{u}) \right]_{\substack{\alpha,d=1,\ldots,P}}.
\end{split}
\end{equation}

\item	Considering the setting $B_d(x_a)=B(x_a,y_d)$ in part (1) and dividing by $\Delta_P(\vec y)$ with parameters $\widetilde P_l$, $\widetilde{n}_{l,j}$, $\widetilde{m}_{l,k}$ so that we can take the limit $\vec{y}\to\vec{\xi}$ yields
\begin{equation}\label{special.result.2}
\begin{split}
&\lim_{\vec x\to \vec \chi, \vec{y}\to \vec{\xi}}\left(\prod_{l=1}^L\prod_{j=1}^{P_l}\partial_{x_{l,j}}^{n_{l,j}}\right) \left(\prod_{l=1}^L\prod_{j=1}^{\widetilde P_l}\partial_{y_{l,j}}^{\widetilde n_{l,j}}\right) \frac{1}{\Delta_P(\vec x)\Delta_P(\vec{y})}\det\left[ B(x_a,y_b) \right]_{\substack{a,b=1,\ldots,P}}\\
=&\lim_{\boldsymbol{u},\boldsymbol{v}\to 0}\left(\prod_{l=1}^L\prod_{k=1}^{d}D_{\vec u_l,k}^{m_{l,k}}\right)\left(\prod_{l=1}^L\prod_{k=1}^{d}D_{\vec{v}_l,k}^{\widetilde m_{l,k}}\right)\det\left[ [\mathcal{K}_{\vec{\chi},\alpha}\otimes\mathcal{K}_{\vec{\xi},\gamma}B](\boldsymbol{u},\boldsymbol{v}) \right]_{\substack{\alpha,\gamma=1,\ldots,P}}.
\end{split}
\end{equation}
\end{enumerate}
\end{corollary}

We concede that it is not immediate where the advantage is in this theorem and its corollaries in its full generality. Surely, we have carried out the limit $\vec x\to\vec\chi$ and have taken care of the Vandermonde determinant(-s) in the denominator. Yet, we are still left with products of differential operators which are rather involved in full generality. Therefore, we want to show in the next subsection what those results would look like in some simple but, nonetheless, interesting non-trivial cases.

\subsection{First derivatives, Borel transforms and combinatorial formulas}\label{subsec:Borel}

To highlight that Theorem~\ref{thm:main.result} has its virtues, we would like to discuss the case of Corollary~\ref{cor:special results} for $L=d=1$ implying $P_1=P=k$. This means there is only one limiting point $\chi$, i.e., $x_1,\ldots,x_P\to\chi$, and at most first derivatives in $x_j$ may appear, implying $n_{1,1}=\ldots=n_{1,m_{1,0}}=0$ and $n_{1,m_{1,0}+1}=\ldots=n_{1,P}=1$. This is the setting of Corollary~\ref{cor:special results}.1. For Corollary~\ref{cor:special results}.2 we have an additional limiting point $\widetilde x_1,\ldots,\widetilde x_P\to\widetilde\chi$ and take the choice  $\widetilde n_{1,1}=\ldots=\widetilde n_{1,\widetilde m_{1,0}}=0$ and $\widetilde n_{1,\widetilde m_{1,0}+1}=\ldots=\widetilde n_{1,P}=1$. To simplify the situations we assume $m_{1,0}=\widetilde m_{1,0}=h$.

In terms of the example of averages of characteristic polynomials, the two settings would correspond to computing the partition functions
\begin{equation}\label{av.no.modulus}
    \mathcal{Z}_{1,1}^{(\U)}(\chi;h,k-h)=\left\langle D_N(\chi)^{h} [D_N(\chi)^\prime]^{k-h} \right\rangle_{\U}
\end{equation}
and
\begin{equation}\label{av.modulus}
    \mathcal{Z}_{1,2}^{(\U)}(\chi,\overline\chi;h,k-h)=\left\langle |D_N(\chi)|^{2h} |D_N(\chi)^\prime|^{2k-2h} \right\rangle_{\U},
\end{equation}
both in the unitary class where we set $\widetilde\chi=\overline{\chi}$.

As we will see in Corollary \ref{cor: factoring derivatives}, the integral transform~\eqref{transform} in Corollary~\ref{cor:special results} can be recast  in terms of more known transform, namely the Borel transformation also known as Borel summation~\cite{Borel}. Borel has introduced it to improve the convergence of  a series as a possible regularisation for a divergent series by dividing the Taylor series with an additional factorial. In our situation it will naturally arise from applying Theorem~\ref{thm:main.result}.
We give two equivalent representations of the Borel transform. The two definitions are related via the residue theorem.

\begin{definition}[Borel transform]\

Let $f: \mathbb{C} \rightarrow \mathbb{C}$, being holomorphic in the open disc $D(\chi,R)$ centered at $\chi$ with radius $R>0$ and  $\partial_\chi^jf(\chi)$  denotes the $j$-th derivative at $\chi$.
The first order Borel transform of $f$ at $\chi$  is, then, defined as (see Ref.~\cite{Borel})
\begin{eqnarray}
   [\mathcal{B}_{\chi}f]( u) =[\mathcal{B}_{\chi,1}f]( u) 
    &:=&\oint_{|z|=R/2} \frac{\dif z}{2\pi i z}\, f\left(\chi+z\right)\exp\left[ u z^{-1}\right] 
   \label{eq: borel transf2}\\ 
    &=& \sum_{j=0}^\infty \frac{\partial_\chi^jf(\chi)}{j!^2}  u^j \label{eq: borel transf}
\end{eqnarray}
and its higher order version is
\begin{eqnarray}
   [\mathcal{B}_{\chi,d}f](\vec u) 
    &:=&\oint_{|z|=R/2} \frac{\dif z}{2\pi i z}\, f(\chi+z)\exp\left[\sum_{j=1}^d (j-1)!\,u_j z^{-j}\right],
   \label{eq: borel transf2.higher}
\end{eqnarray}
where $\vec u=( u_1,\ldots, u_d)\in\mathbb{C}^d$ with $d\geq1$.

\end{definition}

\begin{example}
The higher order version for $d=2$ has a series representation as well. To derive this we perform a Hubbard--Stratonovich transformation~\cite{Hubbard1959} of the term
\begin{equation}
\exp[u_2z^{-2}]=\int_{-\infty}^\infty \frac{\dif t}{\sqrt{\pi}}\exp\left[-t^2+2\sqrt{u_2}z^{-1}t\right].
\end{equation}
The integrals over $t$ and $z$ can be interchanged as they evidently converge absolutely, so that we can carry out the integral over $z$ first which is the ordinary Borel transform of $f$,
\begin{equation}
[\mathcal{B}_{\chi,2}f](\vec u) =\int_{-\infty}^\infty \frac{\dif t}{\sqrt{\pi}} e^{-t^2}  [\mathcal{B}_{\chi}f]( u_1+2\sqrt{u_2}t)=\int_{-\infty}^\infty \frac{\dif t}{\sqrt{\pi}} e^{-t^2}   \sum_{j=0}^\infty \frac{\partial_\chi^jf(\chi)}{j!^2}  ( u_1+2\sqrt{u_2}t)^j.
\end{equation}
Also the series can be interchanged with the integral due to absolute convergence of the  Taylor series and 
\begin{equation}
\begin{split}
\frac{1}{j!(R/2)^j}\int_{-\infty}^\infty \frac{\dif t}{\sqrt{\pi}} e^{-t^2}| u_1+2\sqrt{u_2}t|^j&\leq \frac{2^j}{j!}\int_{-\infty}^\infty \frac{\dif t}{\sqrt{\pi}} e^{-t^2}(| u_1|^j+2^j|u_2|^{j/2}|t|^j)\\
&\leq \frac{(2|u_1|)^j+2^{2j}|u_2|^{j/2}\Gamma[(j+1)/2]}{j!(R/2)^j}
\end{split}
\end{equation}
which has a uniform bound in $j>0$ by a constant. Thus, the integral over $t$ is
\begin{equation}
\int_{-\infty}^\infty \frac{\dif t}{\sqrt{\pi}} e^{-t^2}( u_1+2\sqrt{u_2}t)^j
=(2u_2)^{j}\exp\left[-\frac{u_1^2}{4u_2}\right]\partial_{u_1}^j\exp\left[\frac{u_1^2}{4u_2}\right]=(-i\sqrt{u_2})^{j} H_j\left(i\frac{u_1}{2\sqrt{u_2}}\right),
\end{equation}
where $H_j$ are the Hermite polynomials orthogonal with respect to the weight $e^{-x^2}$. The square root in $u_2$ is not an issue because the combination with the Hermite polynomial and the factor $u_2^{j/2}$ guarantees that the whole expression is a polynomial in $u_2$ and, thence, entire.  Therefore, the second order Borel transform is in terms of a series equal to
\begin{equation}\label{eq: borel transf.2d}
[\mathcal{B}_{\chi,2}f](\vec u) =\sum_{j=0}^\infty \frac{\partial_\chi^jf(\chi)}{j!^2}  (-i\sqrt{u_2})^{j} H_j\left(i\frac{u_1}{2\sqrt{u_2}}\right).
\end{equation}
It can be readily checked in this formula that $[\mathcal{B}_{\chi,2}f](u_1,u_2=0)=[\mathcal{B}_{\chi}f](u_1)$ and that $[\mathcal{B}_{\chi,2}f](u_1,u_2)$ is an entire function in both variables $u_1$ and $u_2$.

Finding some representation of the higher order Borel transform beyond $d=2$ which is similarly explicit and compact as~\eqref{eq: borel transf} or~\eqref{eq: borel transf.2d} is still an open problem. The reason is that increasingly more variables enter the expansion. It is actually quite remarkable for $d=1$ and $d=2$, that the sum of the two powers in $u_1$ and $u_2$ has an explicit form in terms of a known family of functions, namely the Hermite polynomials.
\end{example}

When combining this definition with Corollary~\ref{cor:special results}, we arrive at the following two results. For the second formulations which are of combinatorial nature, we
need the multinomial
\begin{equation}
    \binom{m}{r_1, r_2, \dots, r_k} :=\frac{m!}{r_1!\,r_2!\cdots r_k!}\ ,
\end{equation}
with $\sum_{l=1}^k r_l=m$.

\begin{corollary}
\label{cor: factoring derivatives}\

\begin{enumerate}
\item	Under the conditions of Corollary~\ref{cor:special results}.1 and the special case discussed at the beginning of Sec.~\ref{subsec:Borel} it is
\begin{align}
    \lim_{\vec x\to \chi}\left(\prod_{j=h+1}^{k}\partial_{x_j}\right) \frac{{\rm det}\left[ B_b(x_a)\right]_{\substack{a,b=1,\ldots,k}}}{\Delta_k(\vec x)}
    &=\lim_{u\to 0}\partial_u^{k-h} {\rm det}\left[\partial_u^{a-1} [\mathcal{B}_{\chi}B_b](u)  \right]_{\substack{a,b=1,\ldots,k}}
    \label{eq: factoring derivatives}\\
    &= \sum_{\substack{0 \leq r_a \leq h-k \\ \sum_a r_a = h-k}} \binom{h-k}{r_1, r_2, \dots, r_k} \frac{\det\left[ \partial_\chi^{r_a+a-1}B_b(\chi)\right]_{{a,b=1,\ldots,k}} }{\prod_{c=1}^k(r_c+c-1)!} .
    \label{eq: factoring derivatives2}
\end{align}

\item	Under the conditions of Corollary~\ref{cor:special results}.1 and the special case discussed at the beginning of Sec.~\ref{subsec:Borel} it is
\begin{align}
&\lim_{\vec x\to \chi, \vec{y}\to \vec{\xi}}\left(\prod_{j=h+1}^{k}\partial_{x_{j}}\right) \left(\prod_{j=h+1}^{k}\partial_{y_{j}}\right) \frac{1}{\Delta_k(\vec x)\Delta_k(\vec{y})}\det\left[ B(x_a,y_b) \right]_{\substack{a,b=1,\ldots,k}}\nonumber\\
=&\lim_{u,v\to 0}\partial_u^{k-h}\partial_{v}^{k-h}\det\left[\partial_u^{a-1}\partial_v^{b-1} [\mathcal{B}_{{\chi}}\otimes\mathcal{B}_{{\xi}}B](u,v) \right]_{\substack{a,b=1,\ldots,k}}\label{firstBorelU}\\
=&\sum_{\substack{0 \leq r_a \leq k-h \\ \sum_a r_a = k-h}} 
\sum_{\substack{0 \leq s_b \leq k-h \\ \sum_b s_b = k-h}} \binom{k-h}{r_1, r_2, \dots, r_k}\binom{k-h}{s_1, s_2, \dots, s_k}\frac{\det\left[ \partial_\chi^{r_a+a-1}\partial_\xi^{s_b+b-1}B(\chi,\xi)\right]_{{a,b=1}}^{k} }{\prod_{c=1}^k(r_c+c-1)!(s_c+c-1)!}  .
   \label{firstBorelU2}
\end{align}
\end{enumerate}
\end{corollary}


\begin{proof}
We consider a single limiting point $\chi_1=\chi$ in~\eqref{special.result.1}, that is $L=1$, with $h$ first derivatives ($d=1$) and none in the other $k-h$ variables, i.e., $P=k$.
In that case, the result reads
\begin{equation}\label{main.result.L=1}
    \lim_{\vec x\to \chi}\left(\prod_{j=h+1}^{k}\partial_{x_j}\right) \frac{1}{\Delta_k(\vec x)}{\rm det}\left[ B_b(x_a)\right]_{{a,b=1,\ldots, k}}\\
=\lim_{u\to 0}\partial_u^{k-h} {\rm det}\left[ [\mathcal{K}_{\chi,\alpha}B_b](u) \right]_{{\alpha,b=1,\ldots,k}},
\end{equation}
and the transformation in \eqref{main.result.L=1} equals
\begin{equation}
    [\mathcal{K}_{\chi,\alpha} f](u)=\oint_{|\zeta-\chi|=R/2}\frac{\dif \zeta\, f(\zeta)}{2\pi i } \exp\left[\frac{u}{\zeta-\chi}\right] \frac{\zeta^{k-\alpha}}{(\zeta-\chi)^k},
\end{equation}
where we recall that $R$ is the radius of the open neighbourhood centered at $\chi$ where the function $f$ is holomorphic.
To connect \eqref{main.result.L=1} to the Borel transform, we note that the span of $\{1, \zeta, \zeta^2,\dots,\zeta^{k-1}\}$ and of the monic polynomials $\{1,(\zeta-\chi),\dots,(\zeta-\chi)^{k-1}\}$ are the same  by a binomial transformation.
In particular, we can recombine the entries in the determinant on the right hand side of~\eqref{main.result.L=1} as follows
\begin{align}
    \sum_{\beta=\alpha}^{k} (-\chi)^{\beta-\alpha}\binom{k-\alpha}{\beta-\alpha}  [\mathcal{K}_{\chi,\beta} f](u)  
    &= \oint_{|\zeta-\chi|=R/2}\frac{\dif \zeta\, f(\zeta)}{2\pi i } \exp\left[\frac{u}{\zeta-\chi}\right] \frac{(\zeta-\chi)^{k-\alpha}}{(\zeta-\chi)^k}=\partial_u^{\alpha-1} [\mathcal{B}_{\chi} f](u)\ ,
\end{align}
yielding the claim~\eqref{eq: factoring derivatives} and when applying this twice also~\eqref{firstBorelU}.
In the first equality, we have carried out the binomial sum after shifting $\beta\to \beta + \alpha$, and in the second one we have shifted $\zeta=\chi+z$.

The combinatorial formula follows when employing the series representation~\eqref{eq: borel transf} of the Borel transform and then pulling the sums out of the determinant by using the multi-linearity thereof, i.e., it is
\begin{equation}
\begin{split}
\lim_{u\to 0}\partial_u^{k-h} {\rm det}\left[\partial_u^{a-1} [\mathcal{B}_{\chi}B_b](u)  \right]_{\substack{a,b=1,\ldots,k}}=&\lim_{u\to 0}\partial_u^{k-h} {\rm det}\left[ \sum_{j=a-1}^\infty\frac{\partial_\chi^jB_b(\chi)}{j!(j-a+1)!}u^{j-a+1}\right]_{\substack{a,b=1,\ldots,k}}\\
&\hspace*{-4cm}=\lim_{u\to 0}\partial_u^{k-h} {\rm det}\left[ \sum_{r=0}^\infty\frac{\partial_\chi^{r+a-1}B_b(\chi)}{(r+a-1)!r!}u^{r}\right]_{\substack{a,b=1,\ldots,k}}\\
&\hspace*{-4cm}=\lim_{u\to 0}\partial_u^{k-h} \sum_{r_1,\ldots,r_k=0}^\infty\left(\prod_{a=1}^k\frac{u^{r_a}}{r_a!(r_a+a-1)!}\right) {\rm det}\left[ \partial_\chi^{r+a-1}B_b(\chi)\right]_{\substack{a,b=1,\ldots,k}}.
\end{split}
\end{equation}
The derivative in $u$ with the limit $u\to0$ yields the factorial $(h-k)!$ in the numerator of the multinomial and restricts the multi-sum to $\sum_{a=1}^kr_a=k-h$. This gives~\eqref{eq: factoring derivatives2} and~\eqref{firstBorelU2} when applying this calculation twice.
\end{proof}

In this proof it becomes clear that one can actually readily generalise the result in terms of the Borel transform to higher derivatives as long as we keep only a single limiting point ($L=1$). The formula becomes only marginally more complicated as it will involve the derivative operators $D_{\vec u,j}$ given by the recurrence~\eqref{recurrence.diff}. It is the counterpart of expectation values of moments of a characteristic polynomial and its derivatives at a single point as it has been discussed in~\cite{KS00b,SimmWei}.

\begin{corollary}
\label{cor: factoring derivatives.higher}\

We consider the setting of Corollary~\ref{cor: factoring derivatives}.1  with the only difference in the order of the derivatives which are given by the partition $\vec h=(h_0,\ldots, h_d)\in\mathbb{N}_0^{d+1}$ and $\sum_{j=0}^dh_j=k$ and $(n_1,\ldots,n_k)=
(0,\ldots,0,1,\ldots,1,\ldots,d,\ldots,d)$, with multiplicities $h_j$. Then, we obtain the more general result
\begin{align}
    \lim_{\vec x\to \chi}\left(\prod_{j=1}^{k}\partial_{x_j}^{n_j}\right) \frac{{\rm det}\left[ B_b(x_a)\right]_{\substack{a,b=1,\ldots,k}}}{\Delta_k(\vec x)}
    &=\lim_{\vec u\to 0}\left(\prod_{j=1}^dD_{\vec u,j}^{h_j}\right) {\rm det}\left[\partial_{u_1}^{a-1}  [\mathcal{B}_{\chi,d}B_b](\vec u) \right]_{\substack{a,b=1,\ldots,k}}.
    \label{eq: factoring derivatives3}
\end{align}
\end{corollary} 

\begin{example}
In the simplest case beyond $d=1$, namely $d=2$, this case can be explicitly written as follows
\begin{equation}
\begin{split}
&\lim_{\vec x\to \chi}\left(\prod_{j=h_0+1}^{h_0+h_1}\partial_{x_j}\right)\left(\prod_{j=h_0+h_1+1}^{k}\partial_{x_j}^2\right) \frac{{\rm det}\left[ B_b(x_a)\right]_{\substack{a,b=1,\ldots,k}}}{\Delta_k(\vec x)}\\
    =&\lim_{u_1,u_2\to 0}\partial_{u_1}^{h_1}(\partial_{u_1}^2+\partial_{u_2})^{k-h_0-h_1} {\rm det}\left[\partial_{u_1}^{a-1}  [\mathcal{B}_{\chi,2}B_b](u_1,u_2) \right]_{\substack{a,b=1,\ldots,k}}.
\end{split}
\end{equation}
\end{example}

When dealing with the symplectic or orthogonal class ($\rm C=S,O$) of expectation values of the form
\begin{equation}\label{part.S.O}
    \mathcal{Z}_{1,1}^{({\rm C})}\left(\chi,\overline\chi;\left[\begin{array}{cc} h, & k-h\\ h, & k-h\end{array}\right]\right)=\left\langle |D_N(\chi)|^{2h} |D_N(\chi)'|^{2k-2h}\right\rangle_{\rm C}
\end{equation}
we actually need Theorem~\ref{thm:main.result} with $P=2k, Q=0, L=2$, $P_1=P_2=k$, $d=1$ and $m_{1,0}=m_{2,0}=h$. Therefore, there is no matrix block with $B$ and $C$. For this purpose, we consider the map $\psi$ given by
\begin{equation}
    \psi(x_i) = \begin{cases}
        \chi, &\text{ if $i\in\{1,\dots,k\}$}\\
        \xi, &\text{ if $i\in \{k+1,\dots,2k\}$}
    \end{cases}
\end{equation}
and inverse images
\begin{equation}
    \{x_{1,1},\dots,x_{1,k}\} = \psi^{-1}(\chi), \quad 
    \{x_{2,1},\dots,x_{2,k}\} = \psi^{-1}(\xi).
\end{equation}

\begin{corollary}\label{cor:symplectic.case}
\label{thm: factoring derivatives pfaffian}
    Under the condition of Theorem~\ref{thm:main.result} and the above restriction it is
\begin{equation}
    \begin{aligned}
    &\lim_{\vec x \rightarrow \vec\chi}
   \left( \prod_{a=h+1}^k \partial_{x_{1,a}}\partial_{x_{2,a}}\right)
    \frac{1}{\Delta_{2k}(\vec x)}
    {\rm Pf}
        \begin{bmatrix}
            A(x_{1,a},x_{1,b})&A(x_{1,a},x_{2,b})\\
            A(x_{2,a} ,x_{1,b}) & A(x_{2,a},x_{2,b})\\
        \end{bmatrix}_{a,b=1,\ldots,k}   \\
        =&\lim_{u,v \rightarrow 0}
        (\partial_u
        \partial_v)^{k-h}
        {\rm Pf} 
        \left[
            [\mathcal{K}_{\vec{\chi},\alpha} \otimes \mathcal{K}_{\vec{\chi},\beta} A](u,v)
            \right]_{\alpha,\beta=1,\ldots,2k}
    \end{aligned}
\end{equation}
	with the integral transform
\begin{equation}\label{transform2}
    [\mathcal{K}_{\vec{\chi},\alpha} f](u,v)=\oint_{\mathcal{C}}\frac{\dif \zeta\, f(\zeta)}{2\pi i }\exp\left[\frac{u}{\zeta-\chi}+\frac{v}{\zeta-\xi}\right]\frac{\zeta^{2k-\alpha}}{(\zeta-\chi)^{k}(\zeta-\xi)^{k}},
\end{equation}
	where $\mathcal{C}$ encircles $\chi$ and $\xi$ close enough so that the contour lies in neighbourhoods about these point where  $f$ is holomorphic thereon.
\end{corollary}

The proof is trivial as the corollary is only a restriction of the general Theorem~\ref{thm:main.result}.
To compute the integral transform explicitly, it seems to be that an explicit $f$ is needed. The general integral is surely a challenge and poses an open problem.

\subsection{Higher order derivatives and Kostka numbers}
\label{subsec:Kostka}

Already in the results of the previous section, particularly in~\eqref{eq: factoring derivatives2} and~\eqref{firstBorelU2}, it could be seen that the problem has also a combinatorial interpretation. This is the reason why we present also a second approach.
This is based on the theory of symmetric functions. Thus,  we first introduce our notation for integer partitions, Young diagrams, tableaux and some associated quantities (compare~\cite[\S 4.1]{FultonHarris2004}).

An {\it integer partition} is a (finite or infinite) sequence
\begin{equation}
    \vec\lambda = (\lambda_1, \lambda_2, \dots, \lambda_j, \dots)
\end{equation}
of non-negative integers in semi-decreasing order
\begin{equation}\label{ordering}
    \lambda_1 \geq \lambda_2 \geq \dots \geq \lambda_j\geq \dots
\end{equation}
with finitely many non-zero terms.
We do not distinguish between partitions differing by zeros, e.g. $(2,1) \sim (2,1,0) \sim (2,1,0,0,\dots)$.
The non-zero $\lambda_j$ are called {\it parts} of $\vec\lambda$ and the number of parts is the {\it length} of $\vec\lambda$ denoted by $l(\vec\lambda)=\#\{\lambda'\in\{\lambda_1,\lambda_2,\ldots\}|\lambda'>0\}$.
Furthermore, we define
\begin{equation}
    |\vec\lambda| = \sum_{j=1}^\infty \lambda_j,
\end{equation}
and write $\vec\lambda \vdash m$ to denote a partition $\vec\lambda$ such that $|\vec\lambda| = m$. In particular, a sum over $\vec\lambda \vdash m$ with $l(\vec\lambda)\leq k$ implies a sum over the $k$ variables $\lambda_1,\ldots,\lambda_k=0,\ldots,m$ with $ |\vec\lambda| =m$ and ordering~\eqref{ordering}.
For a non-negative integer $k$ and a partition $\vec\lambda$ with $l(\vec\lambda) \leq k$, we define the strictly ascending \textit{shifted sequence} $\widehat{\lambda}$ by
\begin{equation}\label{lambda.hat.rel}
    \widehat{\lambda}_j = j - 1+\lambda_{k - j+ 1} , \text{ with }1 \leq j\leq k.
\end{equation}
For brevity, we do not indicate the integer $k$ in the notation.
We use the following shorthand for the product of factorials of a tuple of integers $\vec{r} = (r_1,\dots,r_k)$ of length $k$
\begin{equation}
    \vec{r}! = \prod_{j=1}^{k} r_j!
\end{equation}

Associated with any partition $\vec\lambda \vdash m$, there is a corresponding Young diagram $Y_{\vec\lambda}$, which consists of $|\vec\lambda|$ square boxes arranged according to $\vec\lambda$.
Specifically, the $j$th row of $Y_{\vec\lambda}$ has $\lambda_j$ boxes, for example
\begin{equation}
    \vec\lambda = (4,2,2,1) \longleftrightarrow Y_{\vec\lambda} = \ytableaushort{{} {} {} {}, {} {}, {} {}, {}}
\end{equation}
We will use $\vec\lambda$ and $Y_{\vec\lambda}$ interchangeably in what follows.

A Young tableau of shape $\vec\lambda \vdash m$ and weight $\vec\alpha \vdash m$ is a filling of the boxes of $Y_{\vec\lambda}$ with integers $\{1,\dots, l(\vec\alpha)\}$.
In particular, exactly $\alpha_j$ copies of the integer $j$ is used in the filling.
For example, 
\begin{equation}
    \ytableaushort{2 1 1, 3}  \label{eq: not SSYT}
\end{equation}
is a tableau with shape $\vec{\lambda} = (3,1)$ and weight $\vec{\alpha} = (2,1,1)$.
The set of Young tableaux of shape $\vec\lambda$ and weight $\vec\alpha$ is designated $\mathbf{YT}_{\vec\lambda,\vec\alpha}$.
For $|\vec\lambda| \neq |\vec\alpha|$, $\mathbf{YT}_{\vec\lambda,\vec\alpha}$ is empty as neither empty boxes in $Y_{\vec\lambda}$ nor leftover integers are allowed.
A Young tableau is called semi-standard if the filling is weakly increasing along rows and strictly increasing along columns.
We call $\mathbf{SSYT}_{\vec\lambda,\vec\alpha}$ the subset of tableaux in $\mathbf{YT}_{\vec\lambda,\vec\alpha}$ which are semi-standard.
For example,
\begin{equation}
    \ytableaushort{1 1 2, 3}, \quad \quad
    \ytableaushort{1 1 3, 2}, \label{eq: SSYTs}
\end{equation}
are semi-standard tableaux of shape $\vec\lambda = (3,1)$ and weight $\vec\alpha=(2,1,1)$, but the tableau in equation \eqref{eq: not SSYT} is not semi-standard.
The number of semi-standard Young tableaux of shape $\vec\lambda$ and weight $\vec\alpha$ is known as the Kostka number
\begin{equation}
    \label{eq:KostkaDef}
    K_{\vec\lambda,\vec\alpha} = |\mathbf{SSYT}_{\vec\lambda,\vec\alpha}|.
\end{equation}
Since the tableaux in~\eqref{eq: SSYTs} are all the semi-standard tableaux of shape $(3,1)$ and weight $(2,1,1)$, we have $K_{(3,1), (2,1,1)} = 2$.
Note that for $|\vec\alpha| \neq |\vec\lambda|$ the Kostka number $K_{\vec\lambda, \vec\alpha} = 0$. 

In the special case of $\vec\alpha = (1,\ldots,1)$ with $l(\vec\alpha)=m$, the Kostka number $K_{\vec\lambda,(1,\ldots,1)}$ is equal to the dimension of the irreducible representation of the symmetric group $S_m$ corresponding to the partition $\vec\lambda$.
This is one of the few special cases where the Kostka numbers have a known closed-form~\cite[\S 4.1]{FultonHarris2004},
\begin{equation}\label{eq:SymIrrepDim}
    K_{\vec\lambda,(1,\ldots,1)} = \frac{m!}{\prod_{b \in Y_{\vec\lambda}} h_{\vec\lambda}(b)}=\frac{m! \Delta_m(\widehat{\lambda})}{\widehat{\lambda}!},
\end{equation}
where the product is over all boxes in $Y_{\vec\lambda}$.
The hook length $h_{\vec\lambda}(b)$ is equal to the number of boxes in the hook shape going through the box $b$. 
For example, the dark grey box in the Young diagram below has a hook length equal to $4$
\begin{equation}
    \begin{ytableau}
        *(white) & *(white) & *(white) & *(white) \\
        *(white) & *(gray) & *(lightgray) & *(lightgray) \\
        *(white) & *(lightgray) & *(white)
    \end{ytableau} \, .
\end{equation}

We are now ready to state our first result for higher order derivatives of the ratio of a determinant and Vandermonde determinants. This theorem is the combinatorial counterpart of  Corollary~\ref{cor:special results}.2. It generalises, however, this result to higher derivatives. Hence, it covers the case $L=1$ but arbitrary $d$.

\begin{theorem}[Higher order derivatives of determinants -- combinatorial formula]\label{Thm-gen-det}\

    Given a positive integer $k$, two partitions $\vec\alpha$ and $\vec\beta$ of length $l(\vec\alpha),l(\vec\beta)\leq k$,
    and a function $B \colon \mathbb{C} \times \mathbb{C} \to \mathbb{C}$ analytic at the point $(x,y)=(\chi,\xi)$. It is
    \begin{equation}\label{del-gen-det}
    \lim_{(\vec{x},\vec{y}) \to (\chi,\xi)} \left( \prod_{a=1}^k \partial_{x_a}^{\alpha_a} \partial_{y_a}^{\beta_a} \right) \frac{\det[B(x_a,y_b)]_{a,b=1}^{k}}{\Delta_k(\vec{x}) \Delta_k(\vec{y})} = 
    \vec\alpha! \vec\beta! \sum_{\substack{\vec\mu \vdash |\vec\alpha|\\ l(\vec\mu) \leq k}} \sum_{\substack{\vec\lambda \vdash |\vec\beta| \\ l(\vec\lambda) \leq k}} \frac{K_{\vec\mu,\vec\alpha} K_{\vec\lambda,\vec\beta}}{\widehat{\mu}!\ \widehat{\lambda}!} \det\left[\partial_\chi^{\widehat{\mu}_a}\partial_{\xi}^{\widehat{\lambda}_b}B(\chi,\xi)\right]_{a,b=1}^{k} .
    \end{equation}
\end{theorem}

The proof is provided in Subsection~\ref{proof:Thm-gen-unitary}. 
For a related very recent result at $\chi=\xi$ in the CUE we refer to \cite{GMN}.
One can apply this theorem to higher order derivatives of characteristic polynomials in the unitary class. It actually has a simple Pfaffian determinant counterpart.
With the same notation as introduced above, the following holds 
 which is also proven in Subsection~\ref{proof:Thm-gen-unitary}.

\begin{theorem}[Higher order derivatives of Pfaffians with $L=1$ -- combinatorial formula]\label{Thm-gen-pfaffian}\

    Consider an integer $k\in\mathbb{N}$, a partition $\vec \alpha$ with length $l(\vec\alpha) \leq k$,
    and a skew-symmetric function $A \colon \mathbb{C} \times \mathbb{C} \to \mathbb{C}$ analytic at $(x,y)=(\chi,\chi)$. Then, we have
    \begin{equation}\label{del-gen-pfaffian}
        \lim_{\vec x \to \chi} \left( \prod_{j=1}^{2k} \partial_{x_j}^{\alpha_j} \right) \frac{\Pf[A(x_a,x_b)]_{a,b=1,\ldots,2k}}{\Delta_{2k}(\vec x)} = 
    \vec\alpha!  \sum_{\substack{\vec\mu \vdash |\vec\alpha| \\ l(\vec\mu) \leq 2k}} \frac{K_{\vec\mu,\vec\alpha}}{\widehat{\mu}!} \Pf\left[\partial_u^{\widehat{\mu}_a}\partial_v^{\widehat{\mu}_b}A(u,v)|_{u=v=\chi}\right]_{a,b=1,\ldots,2k} .
    \end{equation}
\end{theorem}

It is sometimes useful, especially when considering expectation values of the form~\eqref{part.S.O} or similar, to compute for the Pfaffian the case of $L=2$.
As we have seen in Corollary~\ref{thm: factoring derivatives pfaffian} this becomes quickly involved, even when one restricts to derivatives of at most first order ($d=1$).
Yet, for arbitrary orders of derivatives in the case of two limiting points we arrive at the following result that is derived in Subsection \ref{proof:Thm-gen-symplectic-abs}.

\begin{theorem}[Higher order derivatives of Pfaffians with $L=2$ - combinatorial formula]\label{Thm-gen-symplectic-abs}\

Given a positive integer $k$, a partition $\vec \alpha$, and a skew-symmetric function $A \colon \mathbb{C} \times \mathbb{C} \to \mathbb{C}$
holomorphic at the three points $(\chi,\chi)$, $(\chi,\xi)$ and  $(\xi,\xi)$, the following holds
\begin{equation}
\begin{aligned}\label{del-gen-symplectic-abs}
    \lim_{\vec x \to (\chi,\xi)} &\Big( \prod_{j=1}^k (\partial_{x_{1,j}} \partial_{x_{2,j}})^{\alpha_j} \Big) \frac{1}{\Delta_{2k}(\vec x)}
        \Pf\begin{bmatrix}
            A(x_{1,a},x_{1,b}) & A(x_{1,a},x_{2,b}) \\
            A(x_{2,a},x_{1,b}) & A(x_{2,a},x_{2,a})
        \end{bmatrix}_{a,b=1,\ldots,k} \\
    & = \frac{(\vec\alpha!)^2 (-1)^{|\vec\alpha|}}{(\xi - \chi)^{2|\vec\alpha| + k^2}}
    \sum_{0 \leq q,q' \leq |\vec\alpha|} (-1)^{q'} (\xi - \chi)^{q + q'}
    \sum_{\substack{\vec\nu \vdash q \\ l(\vec\nu) \leq k}} \frac{1}{\widehat{\nu}!} \sum_{\substack{\vec\eta \vdash q' \\ l(\vec\eta) \leq k}} \frac{1}{\widehat{\eta}!}
    \\ 
    & \hspace{-1cm}\times
    \Pf
    \begin{bmatrix}
        \partial_u^{\widehat{\nu}_a}\partial_v^{\widehat{\nu}_b}A(u,v)|_{u=v=\chi}& \partial_\chi^{\widehat{\nu}_a}\partial_\xi^{\widehat{\eta}_b}A(\chi,\xi) \\
        \partial_\xi^{\widehat{\eta}_a}\partial_\chi^{\widehat{\nu}_b}A(\xi,\chi) & \partial_u^{\widehat{\eta}_a}\partial_v^{\widehat{\eta}_b}A(u,v)|_{u=v=\xi}
    \end{bmatrix}_{a,b=1,\ldots,k}\,
    \sum_{\substack{\vec\lambda,\vec\mu \vdash |\vec\alpha| \\ l(\vec\lambda),l(\vec\mu) \leq k}} K_{\vec\lambda,\vec\alpha} K_{\vec\mu,\vec\alpha} \widetilde{\mathcal{A}}\left(\substack{{\widehat\lambda}, {\widehat\nu} \\ {\widehat\mu}, {\widehat\eta}}\right),
\end{aligned}
\end{equation}
    where the coefficients $\widetilde{\mathcal{A}}$ are antisymmetric separately in each of the four tuples of variables. These coefficients, as functions of the upper tuples of size $m$ each and lower tuples of size $n$ each, admit the expansion
\begin{equation}
    \label{eq:def-mathcal-A-tilde}
    \widetilde{\mathcal{A}}\left( \substack{{\widehat\lambda}, {\widehat\nu} \\ {\widehat\mu}, {\widehat\eta}} \right) = \sum_{\sigma \in S_m , \tau \in S_n} \sign(\sigma) \sign(\tau) \mathcal{A}\left( \substack{\widehat{\lambda}_1 - \widehat{\nu}_{\sigma(1)}, \dots, \widehat{\lambda}_m - \widehat{\nu}_{\sigma(m)} \\ \widehat{\mu}_1 - \widehat{\eta}_{\tau(1)}, \dots, \widehat{\mu}_n - \widehat{\eta}_{\tau(n)}} \right)
\end{equation}
where
\begin{equation}
    \label{eq:def-mathcal-A}
    \mathcal{A}\Big( \substack{x_1, \dots, x_m \\ y_1, \dots, y_n} \Big) = \sum_{\substack{r^{(l)}_1+\cdots+r^{(l)}_m=  x_l  \\ l=1,\ldots,m}}
        \sum_{\substack{s^{(j)}_1+\cdots+s^{(j)}_n=  y_j \\ j=1,\ldots,n}} \prod_{j=1}^n \prod_{l=1}^m \binom{r_j^{(l)} + s_l^{(j)}}{r_j^{(l)}}
\end{equation}
is a multi-sum over non-negative integers and $\mathcal{A}$ is symmetric separately in both of the tuples.
\end{theorem}
In all cases we have checked $\tilde{\mathcal{A}}\Big( \substack{\widehat{\lambda}, \widehat{\nu} \\ \widehat{\mu}, \widehat{\eta}} \Big)$ is a non-negative integer.
A proof of this fact is the topic of upcoming work~\cite{GA-AP-unpub}.

As we have mentioned before, Eq.~\eqref{firstBorelU2} reflects the relation to Kostka numbers, in particular~\eqref{del-gen-det}. This can be seen by~\cite[Theorem 2]{Lederer} for the case of the partition $(1,\ldots,1)$ which reads in our setting as follows.

\begin{lemma}
    \label{FirstOrderKostkaLemma}
    Let $(u_{a,p})_{\substack{a=1,\dots,k \\ p \in \mathbb{N}_0}}$ be a $k$-tuple of sequences of complex numbers.
    Then, the following relation holds true,
    \begin{equation}
        \label{eq:FirstOrderKostkaLemma}
        \sum_{\substack{0 \leq r_a \leq k \\ \sum_a r_a = k}} \binom{k}{r_1, r_2, \dots, r_k} \frac{\det\left[ u_{b,r_a+a-1}\right]_{{a,b=1,\dots,k }} }{\prod_{c=1}^k(r_c+c-1)!}
            = \sum_{\vec\lambda \vdash k} \frac{\det\left[ u_{b,\widehat{\lambda}_a}\right]_{{a,b=1,\dots,k }} }{\widehat{\lambda}!} K_{\vec\lambda,(1,\dots,1)}.
    \end{equation}
\end{lemma}

\begin{proof}
We denote the left hand side by
\begin{equation}
    F(\vec{u}) := \sum_{\substack{0 \leq r_a \leq k \\ \sum_a r_a = k}} \binom{k}{r_1, r_2, \dots, r_k} \frac{\det\left[ u_{b,r_a+a-1}\right]_{{a,b=1,\dots,k }} }{\prod_{c=1}^k(r_c+c-1)!}.
\end{equation}
Here, a term with $r_a+a = r_{a'}+a'$ for $a \neq a'$ vanishes due to the determinant. Hence for the nonvanishing terms,
each tuple $(r_a+a-1)_{1\leq a \leq k}$ is a permutation of a shifted partition $\widehat{\lambda}$ of $k$, so we substitute $r_a = \widehat{\lambda}_{\sigma(a)} - a + 1$ and get

\begin{equation}
  F(\vec{u})
   = \sum_{\vec\lambda \vdash  k} \sum_{\sigma \in S_k} \frac{k!}{\prod_{c=1}^k (\widehat\lambda_{\sigma(c)}-c+1)!} \frac{\det\left[ u_{b,\widehat\lambda_{\sigma(a)}}\right]_{{a,b=1,\dots,k }} }{\widehat{\lambda}!}.
\end{equation}
Reordering the columns of the determinant expression yields the sign of $\sigma$. In the final step, we pull factors out of the sum over the symmetric group and swap $c$ and $\sigma(c)$,
\begin{equation}
  F(\vec{u})
    = \sum_{\vec\lambda \vdash  k} \frac{\det\left[ u_{b,\widehat\lambda_a}\right]_{{a,b=1,\dots,k }} }{\widehat{\lambda}!} \sum_{\sigma \in S_k} \sign(\sigma) \frac{k!}{\prod_{c=1}^k (\lambda_{k-c+1}+c-\sigma(c))!},
\end{equation}
such that the remaining sum over the symmetric group is the Kostka number $K_{\vec\lambda,(1,\ldots,1)}$ as shown in~\cite[Theorem 2]{Lederer}.
\end{proof}

\section{Application to mixed characteristic polynomials}\label{sec:App-mixed}

In this section, we consider averages of moments of characteristic polynomials and their derivatives. We will discuss some standard examples and begin two general results.

\subsection{Example: the complex Ginibre ensemble}
\label{subsec:Gin}

The complex Ginibre ensemble (Gin) consists of matrices with independently and identically distributed complex normal distributed entries. It jpdf of complex eigenvalues $\vec z\in\mathbb{C}^N$ is \cite{Ginibre}
\begin{equation}
\mathcal{P}_N^{\rm (Gin)}(\vec z)=\frac{1}{(2\pi)^N\prod_{j=0}^Nj!}|\Delta_N(\vec z)|^2\exp\left[-\sum_{j=1}^N|z_j|^2\right],
\end{equation}
where we identify the weight function $w(z)=e^{-|z|^2}$. The monomials $z^j$ are the orthogonal polynomials of this ensembl, with  norms $h_j=\int_{\mathbb{C}}|z|^{2j}e^{-|z|^2}\mathrm{d}^2z=\pi j!$
The average of products of characteristic polynomials of this ensemble exhibits a determinantal structure~\eqref{productUk} with
the polynomial kernel
\begin{equation}
    \label{eq:KGinUEFinite}
    \mathfrak{K}_N^{\rm (Gin)}(x, \bar{y}) = \frac{1}{\pi}\sum_{j=0}^{N-1} \frac{(x \bar{y})^j}{j!}.
\end{equation}
In the large-$N$ limit in the bulk of the spectrum, the kernel converges to the exponential function
\begin{equation}
    \label{eq:KGinUELim}
  \lim_{N\to\infty}\mathfrak{K}_N^{\rm (Gin)}(x, \bar{y})=  \mathfrak{K}_{\infty}^{\rm (Gin)}(x, \bar{y}) = \frac{1}{\pi}e^{x \bar{y}}
\end{equation}
which is valid for any fixed pair $x,y\in\mathbb{C}$.

\begin{example}[Moments without derivatives]

Let us first consider the limit of the moments of characteristic polynomials without any derivatives. For that we need the mixed derivatives of the kernel,
\begin{equation}
    \label{eq:KGinUELimD}
    \partial_u^a\partial_v^b\mathfrak{K}_{\infty}^{\rm (Gin)}(u, v)=  \partial_u^a\frac{u^b}{\pi}e^{uv}=\frac{e^{uv}}{\pi}(\partial_u+v)^au^b=\frac{a!}{\pi}u^{b-a}L_a^{b-a}(-u v)e^{uv}=\frac{b!}{\pi}v^{a-b}L_b^{a-b}(-u v)e^{uv},
\end{equation}
which are essentially generalised Laguerre polynomials times an exponential function and a monomial. When using Corollary~\ref{cor: factoring derivatives}.2 for $k=h$ or equivalently Theorem~\ref{Thm-gen-det} for $|\vec \alpha|=|\vec \beta|=0$ we arrive at
\begin{equation}
   \lim_{N\to\infty} \frac{\left\langle |D_{N}(\chi)|^{2k}\right\rangle_{\rm  Gin}}{\prod_{j=N}^{N+k-1} \pi j!}
    = \frac{1}{\prod_{j=0}^{k-1} j!^2}\det\left[ \frac{e^{|\chi|^2}}{\pi}(\partial_u+\overline\chi)^{a-1}u^{b-1}|_{u=\chi}\right]_{a,b=1,\ldots,k}.
\end{equation}
We pull the factors $e^{|\chi|^2}/\pi$ out the determinant and recombine the polynomial $(\partial_u+\overline\chi)^{a-1}$ to $\partial_u^{a-1}$ by linear combination of the rows. This gives us an upper triangular matrices with entries $\partial_u^{a-1}u^{b-1}$ which indeed vanishes when $a>b$ and is equal to $(a-1)!$ on the diagonal. Thus, we arrive at
\begin{equation}
    \label{eq:GinUENoD-det-1}
   \lim_{N\to\infty} \frac{\left\langle |D_{N}(\chi)|^{2k}\right\rangle_{\rm Gin}}{\prod_{j=N}^{N+k-1} \pi j!}
    = \frac{e^{k|\chi|^2}}{\pi^k\prod_{j=0}^{k-1} j!}=\frac{e^{k|\chi|^2}}{\pi^k G(k+1)}
\end{equation}
with the Barnes $G$-function. This result was derived previously (apart from some prefactors) for expectations of more general powers of the modulus of the characteristic polynomial by Webb and Wong~\cite{WebbWong}. Indeed, the analytic continuation in $k$ is expected to hold 
in our result as well.
For a discussion of \eqref{eq:KGinUELimD} and its relation to Painlev\'e in the limit at the edge of the spectrum we refer to \cite[Thm. 1.3]{DeanoSimm}, where also two limiting points in the bulk are given.
\end{example}

To extend this example to include derivatives we use Theorem~\ref{Thm-gen-det}.
The detailed computations are given in Appendix~\ref{apx:GinUE}.
For general derivatives the final answer is expressible in terms of factorial Schur polynomials~\cite{BiedenharnLouck1989}, which have determinantal expressions\footnote{Adopted from \cite[(4.3)]{MacDonald1992} with the order of the columns of the numerator determinant inverted, as MacDonald's definition of the Vandermonde also has inverse order compared to \eqref{Vander}.} 
\begin{equation}
    \label{eq:DefineFactorialSchur}
    t_{\vec{\nu}}(x_1,\dots,x_k) = \frac{\det[x_a(x_a-1) \dots (x_a-\widehat{\nu}_{b}+1)]_{a,b=1,\dots,n}}{\Delta_k(x)} =\frac{1}{\Delta_k(x)}\det\left[\frac{x_a!}{(x_a-\widehat{\nu}_b)!}\right]_{a,b=1,\ldots,k}.
\end{equation}
As shown in detail in Appendix \ref{apx:GinUE}, the mixed moment can be written as
\begin{equation}
    \label{eq:thmGinUE}
    \lim_{N\to\infty}\frac{\left\langle \prod_{j=1}^k |\partial_{\chi}^{\alpha_j}D_N(\chi)|^2 \right\rangle_{\rm Gin}}{\prod_{j=N}^{N+k-1} \pi j!}
    = \vec{\alpha}!^2 \Big(\frac{e^{|\chi|^2}}{\pi}\Big)^k \sum_{m=0}^{|\vec{\alpha}|} |\chi|^{2m} \sum_{\substack{\vec{\nu} \vdash (|\vec{\alpha}|-m) \\ l(\vec{\nu}) \leq k}} \frac{1}{\widehat{\nu}!}
    \left( \sum_{\substack{\vec{\lambda} \vdash |\vec{\alpha}| \\ l(\vec{\lambda}) \leq k}}
    K_{\vec{\lambda},\vec{\alpha}} \frac{\Delta_k(\widehat{\lambda})}{\widehat{\lambda}!} t_{\vec{\nu}}(\widehat{\lambda}) \right)^2\ .
\end{equation}
Let us see what it is in specific situations.

\begin{example}[Moments with first derivatives]\label{exmaple.1}

In this case of moments only involving at most first derivatives, the relevant Kostka numbers are of the form $K_{\vec{\lambda},\vec{\alpha}}$ with $\vec{\alpha} = (1,\dots,1,0,\dots,0)$ with $h$ zeroes and $k-h$ ones.
These Kostka numbers are given by~\eqref{eq:SymIrrepDim} as well were we set $m=k-h$ and $K_{\vec{\lambda}, \vec{\alpha}}=K_{\vec{\lambda}, (1,\ldots,1)}=(k-h)!\Delta_k(\widehat{\lambda})/\widehat{\lambda}!$.
For the following equation \eqref{eq:corGinUEZerothFirstExplicit} we use as shifting constant $k-h$ instead of $k$ in \eqref{lambda.hat.rel}, so
$\widehat{\lambda}_j = \lambda_{k-h-j+1}+j-1$, and consistently we change the length restriction of the partitions.
As we show in Appendix~\ref{subsec:GinUEFirstDeriv}, the product of $h$ characteristic polynomials and $k-h$ of its first derivatives has expected modulus squared, with $k-h \geq 1$ given by
\begin{equation}
\begin{aligned}
    \label{eq:corGinUEZerothFirstExplicit}
    &\lim_{N\to\infty}\frac{\left\langle |D_N(\chi)|^{2h} |D_N(\chi)'|^{2(k-h)}\right\rangle_{\rm Gin}}{\prod_{j=N}^{N+k-1} \pi j!}\\
    &\hspace{2em}=
        \begin{aligned}[t]
        &\frac{e^{k|\chi|^2}}{\pi^k (k-h)!^2 G(h+1)} \sum_{m=0}^{k-h-2} |\chi|^{2m} \sum_{\substack{\vec{\nu} \vdash (k-h-m) \\ l(\vec\nu) \leq k-h}} \frac{1}{\prod_{j=1}^{k-h} (\nu_j+k-j)!}
         \left[ \sum_{\substack{\vec{\lambda} \vdash k-h \\ l(\vec\lambda) \leq k-h}} \left(K_{\vec\lambda,(1,\ldots,1)}\right)^2 t_{\vec\nu}(\widehat{\lambda}) \right]^2 \\
        &+ \frac{e^{k|\chi|^2}}{\pi^k} \bigg(\frac{(k-h)^2 |\chi|^{2(k-h-1)}}{k G(k+1)}
        +\frac{|\chi|^{2(k-h)}}{G(k+1)} \bigg)\ .
        \end{aligned}
\end{aligned}
\end{equation}
For $k-h\leq1$ the sum over $m$ is empty and for $k=h$ we recover \eqref{eq:GinUENoD-det-1} by also dropping the $|\chi|^{2(k-h-1)}$ term.
At constant $k-h \in \mathbb{N}_0$, equation \eqref{eq:corGinUEZerothFirstExplicit} can be analytically continued in either $k$ or $h$.
The analytic continuation independently in both variables is still an open question.

\end{example}

\begin{example}[Moments with one higher derivative]
For the case of $\vec{\alpha}=(n,0,\dots,0)$ in~\eqref{eq:thmGinUE} the result is surprisingly simple. The relevant Kostka numbers are
\begin{equation}
    K_{\vec{\lambda}, \vec{\alpha}} = \begin{cases}
        1, &\text{ if $\vec{\lambda}=(n)$,} \\
        0, &\text{ otherwise.}
    \end{cases}
\end{equation}
Therefore, the sum over $\vec{\lambda}$ collapses to a single term given by the one-part partition.
Secondly, $t_{\vec{\nu}}(\widehat{\lambda}) \neq 0$ forces $\vec{\nu}$ to be a one-part partition as well.
Combining this with the general result~\eqref{eq:thmGinUE}, this gives
\begin{equation}\label{high.der.Gin}
    \lim_{N\to\infty}\frac{ \left\langle |\partial_{\chi}^{n}D_N(\chi)|^2|D_N(\chi)|^{2(k-1)}  \right\rangle_{\rm Gin}}{\prod_{j=N}^{N+k-1} \pi j!}
    = \frac{n!^2}{(n+k-1)! G(k)} \frac{e^{k|\chi|^2}}{\pi^k} L_{n+k-1,n}(-|\chi|^2)
\end{equation}
with Barnes $G$-function and where 
\begin{equation}\label{Lag.def}
    L_{a,b}(x) := \sum_{m=0}^b \binom{a}{m}\frac{(-x)^m}{m!} 
\end{equation}
denotes the truncated Laguerre polynomial of index $a$ with powers higher than $b$ dropped. This expression in \eqref{high.der.Gin}  is evidently analytic in $k$ and can be continued therein. We refer to Appendix~\ref{apx: ex one deriv GinUE} for further details.

\end{example}

\begin{example}[Moments with two higher derivatives]
For two derivatives, in particular the setting $\vec{\alpha}=(n_1,n_2,0,\dots,0)$ with $n_1 \geq n_2$, the Kostka numbers simplify to
\begin{equation}
    K_{\vec{\lambda}, \vec{\alpha}} =
    \begin{cases}
        1, &\text{ if $l(\vec{\lambda}) \leq 2$ and $\lambda_2 \leq n_2$, } \\
        0, &\text{ otherwise.}
    \end{cases}
\end{equation}
After some simplifications, which are given in Appendix~\ref{apx: ex two derivs GinUE}, we find the result
\begin{equation}\label{2nd.high.der.Gin}
\begin{aligned}
    & \lim_{N\to\infty}\frac{\left\langle |\partial_{\chi}^{n_1}D_N(\chi)|^2 |\partial_{\chi}^{n_2}D_N(\chi)|^2 |D_N(\chi)|^{2(k-2)} \right\rangle_{\rm Gin}}{\prod_{j=N}^{N+k-1} \pi j!} \\
    &= \frac{n_1!^2 n_2!^2}{G(k-1)} \frac{e^{k|\chi|^2}}{\pi^k} \sum_{m=0}^{n_1+n_2} \frac{|\chi|^{2m}}{m!^2}
    \sum_{r=0}^{\floor{\frac{1}{2}(n_1+n_2-m)}} \frac{\bigg[\sum_{s=0}^{n_2} [\binom{m}{s-r} - \binom{m}{n_1+n_2-s-r+1}] \bigg]^2}{(r+k-2)!(n_1+n_2-m-r+k-1)!} \ .
\end{aligned}
\end{equation}
Here, we have used the floor function $\lfloor x\rfloor$ which gives the largest integer such that $\lfloor x\rfloor\leq x$.
Analytic continuation in $k$ is possible in this result as well.

\end{example}

\subsection{Example: the circular unitary ensemble}
\label{subsec:CUE}

The Circular Unitary ensemble (CUE) consists of matrices distributed with respect to Haar measure of the unitary group ${\rm U}(N)$. Its jpdf of complex eigenvalues $e^{i\theta_j}$, with $\theta_j\in [0,2\pi)$, on the unit circle is given by 
\cite{Mehta2004}
\begin{equation}
\label{CUEjpdf}
\mathcal{P}_N^{\rm (CUE)}(\vec \theta)=\frac{1}{(2\pi)^N N!}\prod_{1\leq k<j\leq N}\left| e^{i\theta_j}-e^{i\theta_k}\right|^2.
\end{equation}
The orthogonal polynomials are again the monomial, with unit norm, leading to the truncated geometric series as kernel, see \eqref{CUEkernel} below. 

We consider the family of expectation values
\begin{equation}\label{eq:CUEoriginal}
   \mathcal{Z}_{d}^{({\rm CUE})}\left(\chi;\vec h\right):= \mathcal{Z}_{d,L=2}^{({\rm CUE})}\left(\left[\begin{array}{c}\chi\\ \overline\chi\end{array}\right];\left[\begin{array}{c} \vec h\\\vec h\end{array}\right]\right)=\left\langle
    | D_N(\chi)|^{2h_{0}} |\partial_\chi D_{N}(\chi)|^{2h_1} \cdots |\partial_\chi^d D_N(\chi)|^{2h_d}
    \right\rangle_{\rm CUE}
\end{equation}
with $|\vec h|=\sum_{j=0}^dh_j=k$. Especially, we aim for explicit formulas for the cases $d=0,1,2$. These expectation values are related to
\begin{equation}\label{eq:CUEoriginal.ref}
    \widetilde{\mathcal{Z}}_{d}^{({\rm CUE})}\left(\chi;\vec h\right):=\left\langle| D_N(\chi)|^{2h_{0}} |\chi\partial_\chi D_{N}(\chi)|^{2h_1} \cdots |(\chi\partial_\chi)^d D_N(\chi)|^{2h_d}
    \right\rangle_{\rm CUE},
\end{equation}
as they have been studied in~\cite{KS00b,SimmWei} while in~\cite{AKW22,ABGS21,KW1,KW2,AGKW24,Keating-etal} they studied partition functions of the form
\begin{equation}\label{eq:CUEoriginal.ref.b}
    \widehat{\mathcal{Z}}_{d}^{({\rm CUE})}\left(\chi;\vec h\right):=\left\langle| D_N(\chi)|^{2h_{0}} |\chi\partial_\chi \chi^{-N/2}D_{N}(\chi)|^{2h_1} \cdots |(\chi\partial_\chi)^d \chi^{-N/2}D_N(\chi)|^{2h_d}
    \right\rangle_{\rm CUE},
\end{equation}
though they are not exactly the same. Certainly, it is for the simplest cases $\mathcal{Z}_{0}^{({\rm CUE})}\left(\chi;\vec h\right)=\widetilde{\mathcal{Z}}_{0}^{({\rm CUE})}\left(\chi;\vec h\right)=\widehat{\mathcal{Z}}_{0}^{({\rm CUE})}\left(\chi;\vec h\right)$ and $\mathcal{Z}_{1}^{({\rm CUE})}\left(\chi;\vec h\right)=|\chi|^{2h_1}\widetilde{\mathcal{Z}}_{1}^{({\rm CUE})}\left(\chi;\vec h\right)$.  Yet, starting with $d=2$ we need a linear combination given by
\begin{equation}
\begin{split}
    &\widetilde{\mathcal{Z}}_{2}^{({\rm CUE})}\left(\chi;h_0,h_1,h_2\right):=\left\langle| D_N(\chi)|^{2h_{0}} |\chi\partial_\chi D_{N}(\chi)|^{2h_1}|(\chi^2\partial_\chi^2+\chi\partial_\chi) D_{N}(\chi)|^{2h_2}
    \right\rangle_{\rm CUE}\\
    =&\sum_{l_1,l_2=0}^{h_2}\binom{h_2}{l_1}\binom{h_2}{l_2}\chi^{h_1+2h_2-l_1}\bar\chi^{h_1+2h_2-l_2}\mathcal{Z}_{2,1}^{({\rm CUE})}\left(\left[\begin{array}{c}\chi\\ \overline\chi\end{array}\right];\left[\begin{array}{c} h_0,h_1+l_1,h_2-l_1 \\h_0,h_1+l_2,h_2-l_2 \end{array}\right]\right).
\end{split}
\end{equation}
Considering the setting of~\cite{KS00b,SimmWei}, where the authors consider the large $N$ limit at $|\chi|=1$, or  without loss of generality $\chi=1$, due to the isotropic spectrum, it happens that each derivative scales like $N$, see the ensuing discussion. Thus, it is
\begin{equation}\label{asymp.equiv}
\lim_{N\to\infty}\lim_{|\chi|\to1}\frac{\mathcal{Z}_{d}^{({\rm CUE})}\left(\chi;\vec h\right)}{\widetilde{\mathcal{Z}}_{d}^{({\rm CUE})}\left(\chi;\vec h\right)}=1,
\end{equation}
meaning it makes them asymptotically equivalent when the powers $\vec h$ remain fixed. 

The relation with $\widehat{\mathcal{Z}}_{d}^{({\rm CUE})}$ as studied in~\cite{Keating-etal} is evidently more complicated and we need to consider also mixed powers of derivatives of $D_N(\chi)$ and $D_N(\chi^*)$.

The above averages can be traced back to those without derivatives of the form
\begin{equation}\label{expect.CUE.1}
 \mathcal{Z}_{0,2k}^{({\rm CUE})}\left(\left[\begin{array}{c}\vec x\\ \vec y\end{array}\right],\left[\begin{array}{c}k\\ k\end{array}\right]\right)=\left\langle
   \prod_{j=1}^k\det(x_j\mathbf{1}_N-U)\det(y_j\mathbf{1}_N-\hconj{U})
    \right\rangle_{\rm CUE}\\
\end{equation}
with $\hconj{U}$ the Hermitian conjugate of $U\in\U(N)$.
Because the unitary matrices satisfy $\hconj{U}=U^{-1}$, those expectation values are equivalent to
\begin{equation}\label{expect.CUE.2}
   \mathcal{Z}_{0}^{({\rm CUE})}\left(\left[\begin{array}{c}\vec x\\ \vec y\end{array}\right]\right)=(-1)^{Nk}\left(\prod_{j=1}^ky_j^N\right)\left\langle
  \det U^{-k} \prod_{j=1}^k\det(x_j\mathbf{1}_N-U)\det(y'_j\mathbf{1}_N-U)
    \right\rangle_{\rm CUE}
\end{equation}
where $y'_j=y_j^{-1}$. This certainly only works when all $y_j\neq0$.

Like for the Ginibre ensemble, the monomials are the orthogonal polynomials of the corresponding ensemble. Combining this together with~\eqref{expect.CUE.1} and~\eqref{expect.CUE.2}, we know that there are two useful representations for this ensemble. The first one is in terms of a $k\times k$ determinant, but at the expense of dividing it by two Vandermonde determinants, one in $\vec u$ and one in $\vec v$, i.e.,
\begin{equation}\label{kernel.CUE.1}
 \mathcal{Z}_{0}^{({\rm CUE})}\left(\left[\begin{array}{c}\vec x\\ \vec y\end{array}\right]\right)=\frac{1}{\Delta_k(\vec x)\Delta_k(\vec y)}\det\left[\frac{1-(x_ay_b)^{N+k}}{1-x_ay_b}\right]_{a,b=1,\ldots,k}.
\end{equation}
For this purpose we recall that the corresponding kernel takes the simple form \cite[ch. 11]{Mehta2004}
\begin{equation}
\label{CUEkernel}
\mathfrak{K}_N^{({\rm CUE})}(x_a,y_b)=\sum_{j=0}^{N-1}(x_ay_b)^j=\frac{1-(x_ay_b)^{N}}{1-x_ay_b}.
\end{equation}

\begin{example}[Finite $N$-results of up to two derivatives]
Coming back to the original expectation value with the derivatives~\eqref{eq:CUEoriginal}, we can express it in terms of~\eqref{expect.CUE.1} as follows
\begin{equation}\label{eq:CUEoriginal.mod}
   \mathcal{Z}_{d}^{({\rm CUE})}\left(\chi,\vec h\right)=\lim_{\vec x\to\chi,\vec y\to\bar\chi}\left(\prod_{j=1}^k\partial_{x_j}^{n_j}\partial_{y_j}^{n_j}\right)\mathcal{Z}_{0,2k}^{({\rm CUE})}\left(\left[\begin{array}{c}\vec x\\ \vec y\end{array}\right],\left[\begin{array}{c}k\\ k\end{array}\right]\right)
\end{equation}
with $(n_1,\ldots,n_k)=(0,\ldots,0,1,\ldots,1,\ldots,d,\ldots,d)$, with multiplicity $h_j$ of $j$.
For general orders of derivatives, this expression gets quickly involved due to the generalised Borel transform~\eqref{eq: borel transf2.higher}. Yet, everything is relatively under control when only up to second derivatives appear ($d=2$). Then, we can employ Corollary~\ref{cor: factoring derivatives.higher} twice and obtain
\begin{equation}\label{CUE.result.general.d2}
\begin{split}
\mathcal{Z}_{1}^{({\rm CUE})}\left(\chi; k-h_1-h_2,h_1,h_2\right)=&\lim_{\vec u\to0}[\partial_{u_1}\partial_{\bar u_1}]^{h_1}[(\partial_{u_1}^2+\partial_{u_2})(\partial_{\bar u_1}^2+\partial_{\bar u_2})]^{h_2}\\
&\times\det\left[\partial_{u_1}^{a-1}\partial_{\bar u_1}^{b-1}[\mathcal{B}_\chi\otimes \mathcal{B}_{\bar\chi}\mathfrak{K}_{N+k}^{\rm (CUE)}](\vec u,\vec{\bar u})\right]_{a,b=1,\ldots k}.
\end{split}
\end{equation}
The double second order Borel transform and its derivatives can be explicitly carried out employing~\eqref{eq: borel transf.2d} in both entries leading to
\begin{equation}\label{CUE.result.gen-kernel.d2}
\begin{split}
\partial_{u_1}^{a-1}\partial_{\bar u_1}^{b-1}\left[\mathcal{B}_{\chi,2}\otimes \mathcal{B}_{\bar\chi,2}\mathfrak{K}_{N+k}^{\rm (CUE)}\right](\vec u,\vec{\bar u})=&\partial_{u_1}^{a-1}\partial_{\bar u_1}^{b-1}\sum_{l=0}^{N+k-1}\left|\sum_{j=0}^l\frac{l!}{j!^2(l-j)!}\chi^{l-j}(-i\sqrt{u_2})^{j} H_j\left(i\frac{u_1}{2\sqrt{u_2}}\right)\right|^2\\
=&\sum_{l=\max\{a-1,b-1\}}^{N+k-1}\left(\sum_{j=a-1}^l\frac{l!\chi^{l-j}(-i\sqrt{u_2})^{j-a+1} H_{j-a+1}\left(i\frac{u_1}{2\sqrt{u_2}}\right)}{j!(l-j)!(j-a+1)!}\right)\\
&\times\left(\sum_{j=b-1}^l\frac{l!\bar\chi^{l-j}(-i\sqrt{\bar u_2})^{j-b+1} H_{j-b+1}\left(i\frac{\bar u_1}{2\sqrt{\bar u_2}}\right)}{j!(l-j)!(j-b+1)!}\right).
\end{split}
\end{equation}

The case $h_2=0$, meaning only up to first derivatives, simplifies even more. For this case one may use
\begin{equation}
\lim_{u_2\to0}(-i\sqrt{u_2})^{j} H_{j}\left(i\frac{u_1}{2\sqrt{u_2}}\right)=u_1^{j}\ ,
\end{equation}
to find
\begin{equation}\label{CUE.result.general}
\mathcal{Z}_{1}^{({\rm CUE})}\left(\chi;k-h_1,h_1\right)=\lim_{u\to0}[\partial_{u}\partial_{\bar u}]^{h_1}\det\left[\partial_u^{a-1}\partial_{\bar u}^{b-1}[\mathcal{B}_\chi\otimes \mathcal{B}_{\bar\chi}\mathfrak{K}_{N+k}](u,\bar u)\right]_{a,b=1,\ldots k},
\end{equation}
 where
\begin{equation}\label{CUE.result.gen-kernel}
\begin{split}
\partial_u^{a-1}\partial_{\bar u}^{b-1}[\mathcal{B}_\chi\otimes \mathcal{B}_{\bar\chi}\mathfrak{K}_{N+k}](u,\bar u)=\sum_{l=\max\{a-1,b-1\}}^{N+k-1}\chi^{l-a+1}\bar\chi^{l-b+1}L_{l-a+1}^{(a-1)}\left(-\frac{u}{\chi}\right)L_{l-b+1}^{(b-1)}\left(-\frac{\bar u}{\bar\chi}\right)
\end{split}
\end{equation}
with $L_p^{(q)}$ the generalised Laguerre polynomials.
\end{example}

\begin{example}[Large $N$-results inside the unit disc]
When assuming that $|\chi|<1$, the limit $N\to\infty$ is easily tractable for the expectation value~\eqref{eq:CUEoriginal.mod}. When starting from~\eqref{kernel.CUE.1} and using the fact that then also $|x_ay_b|<1$, we find the limit
\begin{equation}
\lim_{N\to\infty}\mathcal{Z}_{0}^{({\rm CUE})}\left(\left[\begin{array}{c}\vec x\\ \vec y\end{array}\right];\left[\begin{array}{c}k\\ k\end{array}\right]\right)=\frac{1}{\Delta_k(\vec x)\Delta_k(\vec y)}\det\left[\frac{1}{1-x_ay_b}\right]_{a,b=1,\ldots,k}=\frac{1}{\prod_{a,b=1}^k(1-x_ay_b)},
\end{equation}
where we have applied the explicit evaluation of the Cauchy determinant~\cite{KG09a}. This result was already found in~\cite[Eq. (2.21)]{SimmWei}. It immediately yields for the expectation value without any derivative the result
\begin{equation}
\lim_{\vec x\to\chi,\vec y\to\bar\chi}\lim_{N\to\infty}\mathcal{Z}_{0}^{({\rm CUE})}\left(\left[\begin{array}{c}\vec x\\ \vec y\end{array}\right],\left[\begin{array}{c}k\\ k\end{array}\right]\right)=\frac{1}{(1-|\chi|^2)^{k^2}},
\end{equation}
which is analytically continuable in $k$ and equivalent to the fact that the random variable $|D_N(\chi)|^2$ is log-normal distributed when $|\chi|<1$.

The case of up to first derivatives can be derived from this result, too. We note, that, the limit $N\to\infty$ is uniform in $x_a$ and $y_b$ when $|x_a|,|y_b|<1$ and the function is holomorphic about $\chi$ and $\bar\chi$. Then, Weierstrass' theorems about such limits allow us to interchange them with $\vec x\to\chi$ and $\vec y\to\bar\chi$ as well as with the derivatives. Therefore, it is
\begin{equation}
\begin{split}
\lim_{N\to\infty}\mathcal{Z}_{1}^{({\rm CUE})}\left(\chi;k-h_1,h_1\right)&=\lim_{\vec x\to\chi,\vec y\to\bar\chi}\left(\prod_{j=h_0+1}^{h_0+h_1}\partial_{x_j}\partial_{y_j}\right)\frac{1}{\prod_{a,b=1}^k(1-x_ay_b)}\\
&=\lim_{\vec y\to\bar\chi}\left(\prod_{j=h_0+1}^{h_0+h_1}\partial_{y_j}\right)\frac{1}{\prod_{b=1}^k(1-\chi y_b)^k}\left[\sum_{b=1}^k\frac{y_b}{1-\chi y_b}\right]^{h_1},
\end{split}
\end{equation}
where we have carried out the derivatives in $x_j$ already. After expressing the term with the sum
\begin{equation}
\left[\sum_{b=1}^k\frac{y_b}{1-\chi y_b}\right]^{h_1}=\partial_u^{h_1}\exp\left[u\sum_{b=1}^k\frac{y_b}{1-\chi y_b}\right]\biggl|_{u=0}
\end{equation}
as a derivative over an auxiliary variable $u$, we can also carry out the derivatives in $y_j$ leading to
\begin{equation}
\begin{split}
\lim_{N\to\infty}\mathcal{Z}_{1}^{({\rm CUE})}\left(\chi;k-h_1,h_1\right)=\lim_{u\to0}\partial_u^{h_1}\frac{[k\chi(1-|\chi|^2) +u]^{h_1}}{(1-|\chi|^2)^{k^2+2h_1}}
=\frac{h_1!L_{h_1}^{(0)}(-k^2|\chi|^2)}{(1-|\chi|^2)^{k^2+2h_1}}.
\end{split}
\end{equation}
In the second step we have rescaled and shifted $u\to (1-|\chi|^2)(u-k^2|\chi|^2)/(k\bar\chi)$ to identify the Rodrigues formula for the Laguerre polynomials. This result allows for an analytic continuation in both $k$ and $h_1$ with the help of Carlson's theorem, when identifying $h_1!=\Gamma(h_1+1)$ with the Gamma function, and $L_{h_1}^{(0)}(z)={ _1F_1}(-h_1;1;z)$ with the confluent hypergeometric function. 
It agrees with \cite[Thm. 1.1]{SimmWei}, where also the example for two points is given.

Unfortunately, the above computations seem to carry only as far as that. For higher derivatives we exploit the integral representation of the higher order Borel transform~\eqref{eq: borel transf2.higher} in combination with Corollary~\ref{cor: factoring derivatives.higher} only applied to $\vec x$ at the moment. Then, it is
\begin{equation}\label{limx}
   \mathcal{Z}_{d}^{({\rm CUE})}\left(\chi,\vec n\right)=\lim_{\vec y\to\bar\chi,\vec u\to0}\left(\prod_{j=1}^k\partial_{y_j}^{n_j}\right)\left(\prod_{j=1}^dD_{\vec u,j}^{h_j}\right)\frac{\det\left[\partial_{u_1}^{a-1}[\mathcal{B}_{\chi,2}\otimes {\rm id} \;\, \mathfrak{K}_{N+k}^{\rm (CUE)}](\vec u,y_b)\right]_{a,b=1,\ldots,k}}{\Delta_k(\vec y)},
\end{equation}
with the kernel
\begin{equation}
\begin{split}
\partial_{u_1}^{a-1}[\mathcal{B}_{\chi,2}\otimes {\rm id} \;\, \mathfrak{K}_{N+k}^{\rm (CUE)}](\vec u,y_b)
=&\oint_{|\chi+z|=R}\frac{\dif z}{2\pi i z^{a}}\frac{1-(\chi+z)^{N+k}y_b^{N+k}}{1-(\chi+z)y_b}\exp\left[\sum_{j=1}^d(j-1)! u_jz^{-j}\right],
\end{split}
\end{equation}
where we choose the radius $R>0$ such that $|(\chi+z)y_b|<1$ and $|\chi|<R$ so that the origin lies in the interior of the counter-clockwise contour. This is possible when $|y_b|$ is close enough to $|\chi|<1$. With such a choice the integrand has a uniform limit in $z$ because of $|\chi|<1$, and in any compact subset of $\vec u\in\mathbb{C}^d$, so that
\begin{equation}
\begin{split}
\mathfrak{BK}_{a,b}(\chi; \vec u, y_b):=&\lim_{N\to\infty}\partial_{u_1}^{a-1}[\mathcal{B}_{\chi,2}\otimes {\rm id}\mathfrak{K}_{N+k}^{\rm (CUE)}](\vec u,y_b)=\oint_{|\chi+z|=R}\frac{\dif z}{2\pi i z^{a}}\frac{\exp\left[\sum_{j=1}^d(j-1)! u_jz^{-j}\right]}{1-(\chi+z)y_b}.
\end{split}
\end{equation}
The integrand has only the isolated essential singularity at the origin $z=0$ encircled by the contour and the simple pole at $z=(1-\chi y_b)/y_b$ in the exterior of the contour on the Riemann sphere. Thus, we can evaluate the integral at the simple pole via the residue theorem which leads to
\begin{equation}
\begin{split}
\mathfrak{BK}_{a,b}(\chi; \vec u, y_b)=\frac{y_b^{a-1}}{(1-\chi y_b)^a}\exp\left[\sum_{j=1}^d(j-1)! u_j\frac{y_b^{j}}{(1-\chi y_b)^j}\right].
\end{split}
\end{equation}
The exponential functions and a factor $1/(1-\chi y_b)$ can be pulled out of the determinant in~\eqref{limx} and what remains is the ratio of Vandermonde determinants
\begin{equation}
\frac{\det[(y_b/[1-\chi y_b])^{a-1}]_{a,b=1,\ldots k}}{\Delta_k(\vec y)}=\frac{1}{\prod_{a=1}^k(1-\chi y_a)^{k-1}}.
\end{equation}
Combining this in~\eqref{limx} with the other terms we arrive at
\begin{equation}
\begin{split}
   \mathcal{Z}_{d}^{({\rm CUE})}\left(\chi,\vec n\right)=&\lim_{\vec y\to\bar\chi,\vec u\to0}\left(\prod_{j=1}^k\partial_{y_j}^{n_j}\right)\left(\prod_{j=1}^dD_{\vec u,j}^{h_j}\right)\prod_{a=1}^k\frac{\exp\left[\sum_{j=1}^d(j-1)! u_jy_a^{j}/(1-\chi y_a)^j\right]}{(1-\chi y_a)^{k}}\\
   =&\lim_{\vec u\to0}\left(\prod_{j=1}^dD_{\vec u,j}^{h_j}\right)\left(\prod_{j=0}^d\left[\lim_{y\to\bar\chi}\partial_y^j\frac{\exp\left[\sum_{j=1}^d(j-1)! u_jy^{j}/(1-\chi y)^j\right]}{(1-\chi y)^{k}}\right]^{h_j}\right)
\end{split}
\end{equation}
because the products in $y_j$ factorise. 

We can still massage the expression which can be easier done when writing the derivative in $y$ in terms of a contour integral,
\begin{equation}
\lim_{y\to\bar\chi}\partial_y^j\frac{\exp\left[\sum_{j=1}^d\frac{(j-1)! u_jy^{j}}{(1-\chi y)^j}\right]}{(1-\chi y)^{k}}=j!\oint\frac{\dif z}{2\pi i (z-\bar\chi)^{j+1}(1-\chi z)^k}\exp\left[\sum_{l=1}^d(l-1)! \frac{u_lz^{l}}{(1-\chi z)^l}\right].
\end{equation}
The counter-clockwise contour encircles $\bar\chi$ close enough so that the singularity at $z=1/\chi$ lies in its exterior. Next, we perform the M\"obius transformation
\begin{equation}
z=\bar\chi\frac{1+z'}{1+|\chi|^2z'}\qquad\Leftrightarrow \qquad z'=-\frac{\bar\chi-z}{|\chi|^2z-\bar\chi}\ ,
\end{equation}
which maps the two singularities $\bar\chi$ and $1/\chi$ to $0$ and $\infty$, respectively. This leads to
\begin{equation}
\lim_{y\to\bar\chi}\partial_y^j\frac{\exp\left[\sum_{j=1}^d\frac{(j-1)! u_jy^{j}}{(1-\chi y)^j}\right]}{(1-\chi y)^{k}}=j!\oint\frac{\dif z'}{2\pi i {z'}^{j+1}}\frac{(1+|\chi|^2z')^{k+j-1}}{\bar\chi^{j}(1-|\chi|^2)^{k+j}}\exp\left[\sum_{l=1}^d(l-1)! \frac{u_l\bar\chi^l(1+z')^{l}}{(1-|\chi|^2)^l}\right].
\end{equation}
Finally, we exploit the scaling property~\eqref{der.scaling}, when rescaling $u_j\to(1-|\chi|^2)^j u_j/\bar\chi^j$, to find the result
\begin{equation}\label{CUE.result.1}
\begin{split}
   \mathcal{Z}_{d}^{({\rm CUE})}\left(\chi,\vec n\right)=&\frac{1}{(1-|\chi|^2)^{k^2+2\sum_{j=1}^djh_j}}\lim_{\vec u\to0}\left(\prod_{j=1}^dD_{\vec u,j}^{h_j}\right)\left(\prod_{j=0}^d\mathfrak{P}_j^{h_j}(|\chi|^2;\vec u)\right),
\end{split}
\end{equation}
with the functions
\begin{equation}
\mathfrak{P}_j(|\chi|^2;\vec u)=j!\oint_{|z'|=1}\frac{\dif z'}{2\pi i {z'}^{j+1}}(1+|\chi|^2z')^{k+j-1}\exp\left[\sum_{l=1}^d(l-1)! u_l(1+z')^{l}\right],
\end{equation}
which are polynomials of degree $j$ in $|\chi|^2$. The cases $d=0$ and $d=1$ can be readily regained and we believe that also for $d>1$ our result~\eqref{CUE.result.1} allows for analytic continuation, once rewriting the derivative in $\vec u$ as integrals.
\end{example}

The representation~\eqref{expect.CUE.2} is especially useful when the limiting point $\chi$ lies on the unit circle. We note that $\chi$ can always be rotated to $\chi=1$ in such a case due to the isotropy of the random matrix ensemble, meaning one has only to consider one limiting point. This also carries over when going from $\vec y$ to $\vec {y'}$ in~\eqref{expect.CUE.2}. Then, the expectation value can be expressed in terms of a $2k\times2k$ determinant and only a single Vandermonde determinant, in particular it is
\begin{equation}\label{kernel.CUE.2}
\begin{split}
 \mathcal{Z}_{0}^{({\rm CUE})}\left(\left[\begin{array}{c}\vec x\\ \vec y\end{array}\right]\right)=&\frac{(-1)^{Nk}}{N!}\left(\prod_{j=1}^ky_j^N\right)\left(\prod_{j=1}^N\int_{|z_j|=1}\frac{\dif z_j}{2\pi z_j^{k+1}}\right)
 \Delta_N(\vec z^*)\Delta_N(\vec z) \prod_{l=1}^N\prod_{j=1}^k(x_j-z_l)(y'_j-z_l)\\
    =&\left(\prod_{j=1}^ky_j^N\right)\frac{1}{\Delta_{2k}(\vec x,\vec{y'} )}\det\left[\begin{array}{c|c} x_a^{b-1}&x_a^{N+k+b-1}\\ \hline {y'}_a^{b-1}&{y'}_a^{N+k+b-1} \end{array}\right]_{a,b=1,\ldots,k}
\end{split}
\end{equation}
where we have combined the second Vandermonde determinant and subsequent factors in the first line into a larger Vandermonde determinant, $\Delta_{N+2k}(\vec x,\vec y,\vec z)/\Delta_{2k}(\vec x,\vec y)$. Also, 
we have applied the generalised Andr\'eief identity~\cite[Appendix C]{KG09a} and the orthogonality of the monomials. We recall the relation $y_j'=y_j^{-1}$.

\begin{example}[Large $N$-result on the unit circle]
To cover both types of partition functions~\eqref{eq:CUEoriginal.ref} and~\eqref{eq:CUEoriginal.ref.b} we consider the following 
one-parameter family of partition functions at $\chi=1$,
\begin{equation}
\begin{split}
&\widetilde{\mathcal{Z}}_{d,c}^{({\rm CUE})}\left(\chi=1;\vec h\right)= \lim_{\vec x,\vec {y'}\to1}\left(\prod_{j=1}^k(x_j\partial_{x_j})^{n_j}(-y'_j\partial_{y'_j})^{n_j}\right)\left(\prod_{j=1}^k\frac{x_j^{Nc}}{{y'_j}^{Nc}}\right)\mathcal{Z}_{0}^{({\rm CUE})}\left(\left[\begin{array}{c}\vec x\\ \vec y\end{array}\right]\right)\\
&\quad=\lim_{\vec x,\vec {y'}\to1}\left(\prod_{j=1}^k(x_j\partial_{x_j}+Nc)^{n_j}(-y'_j\partial_{y'_j}+N+Nc)^{n_j}\right)\frac{\det\left[\begin{array}{c|c} x_a^{b-1}&x_a^{N+k+b-1}\\ \hline {y'}_a^{b-1}&{y'}_a^{N+k+b-1} \end{array}\right]_{a,b=1,\ldots,k}}{\Delta_{2k}(\vec x,\vec{y'} )}.
\end{split}
\end{equation}
The case $c=0$ corresponds to~\eqref{eq:CUEoriginal.ref} and the choice $c=-1/2$ yields the partition function~\eqref{eq:CUEoriginal.ref.b}.

The proper scaling is $x_j=1+w_j/N$ and $y_j=1+v_j/N$ which gives
\begin{equation}
\begin{split}
\widetilde{\mathcal{Z}}_{d,c}^{({\rm CUE})}\left(\chi=1;\vec h\right)=&\lim_{\vec w,\vec {v}\to0}\left(\prod_{j=1}^k(N(c+\partial_{w_j})+w_j]\partial_{w_j})^{n_j}(N(1+c-\partial_{v_j})-v_j\partial_{v_j})^{n_j}\right)\\
&\times\frac{N^{k(2k-1)}}{\Delta_{2k}(\vec w,\vec{v} )}\det\left[\begin{array}{c|c}  (1+w_a/N)^{b-1}&(1+w_a/N)^{N+k+b-1}\\ \hline (1+v_a/N)^{b-1}&(1+v_a/N)^{N+k+b-1} \end{array}\right]_{a,b=1,\ldots,k}\\
=&\lim_{\vec w,\vec {v}\to0}\left(\prod_{j=1}^k(N(c+\partial_{w_j})+w_j]\partial_{w_j})^{n_j}(N(1+c-\partial_{v_j})-v_j\partial_{v_j})^{n_j}\right)\\
&\times\frac{N^{k^2}}{\Delta_{2k}(\vec w,\vec{v} )}\det\left[\begin{array}{c|c}  w_a^{b-1}&(1+w_a/N)^{N+k}w_a^{b-1}\\ \hline v_a^{b-1}&(1+v_a/N)^{N+k}v_a^{b-1} \end{array}\right]_{a,b=1,\ldots,k}.
\end{split}
\end{equation}
In the second line we have recombined the columns.
When rescaling the whole expression with $N^{-k^2-2\sum_{j=1}^djh_j}$, we can perform the limit by Weierstrass' Theorem, due to the uniform convergence of entire functions on any compact set of the complex plane. Therefore, it is
\begin{equation}
\begin{split}
\lim_{N\to\infty}\frac{\widetilde{\mathcal{Z}}_{d,c}^{({\rm CUE})}\left(\chi=1;\vec h\right)}{N^{k^2+2\sum_{j=1}^djh_j}}=&\lim_{\vec w,\vec {v}\to0}\left(\prod_{j=1}^k\sum_{l_1,l_2=0}^{n_j}\binom{n_j}{l_1}\binom{n_j}{l_2}c^{n_j-l_1}(1+c)^{n_j-l_2}\partial_{w_j}^{l_1}(-\partial_{v_j})^{l_2}\right)\\
&\times\frac{\det\left[\begin{array}{c|c}  w_a^{b-1}&e^{w_a}w_a^{b-1}\\ \hline v_a^{b-1}&e^{v_a}v_a^{b-1} \end{array}\right]_{a,b=1,\ldots,k}}{\Delta_{2k}(\vec v,\vec w )}.
\end{split}
\end{equation}
When carrying out the limit, it is a good thing that the result of Corollary~\ref{cor: factoring derivatives.higher} can be understood as a multi-linear map, where we replace each $\partial_{x_j}^{n_j}$ by $D_{\vec u,n_j}$. In the present situation, we replace $\partial_{w_j}^{l_1}(-\partial_{v_j})^{l_2}$ by $(-1)^{l_2}D_{\vec u,l_1}D_{\vec u,l_2}$, leading to
\begin{equation}\label{CUE.circle.1}
\begin{split}
\lim_{N\to\infty}\frac{\widetilde{\mathcal{Z}}_{d,c}^{({\rm CUE})}\left(\chi=1;\vec h\right)}{N^{k^2+2\sum_{j=1}^djh_j}}=&\lim_{\vec u\to0}\left(\prod_{n=1}^d\left[\sum_{l_1,l_2=0}^{n}\binom{n}{l_1}\binom{n}{l_2}c^{n-l_1}(1+c)^{n-l_2}(-1)^{l_2}D_{\vec u,l_1}D_{\vec u,l_2}\right]^{h_j}\right)\\
&\times\det\left[\partial_{u_1}^{a-1}\mathfrak{F}_{b-1}^{(1)}(\vec u),\ \partial_{u_1}^{a-1}\mathfrak{F}_{b-1}^{(2)}(\vec u)\right]_{\substack{a=1,\ldots,2k\\b=1,\ldots,k}}\ ,
\end{split}
\end{equation}
with the two sets of functions
\begin{equation}
\begin{split}
\partial_{u_1}^{a-1}\mathfrak{F}_{b-1}^{(1)}(\vec u)=&\oint_{|z|=1}\frac{\dif z}{2\pi i z^{a-b+1}}\exp\left[\sum_{j=1}^d(j-1)!u_jz^{-j}\right],\\
\partial_{u_1}^{a-1}\mathfrak{F}_{b-1}^{(2)}(\vec u)=&\oint_{|z|=1}\frac{\dif z}{2\pi i z^{a-b+1}}\exp\left[z+\sum_{j=1}^d(j-1)!u_jz^{-j}\right].
\end{split}
\end{equation}
It happens that $\partial_{u_1}^{a-1}\mathfrak{F}_{b-1}^{(1)}(\vec u)$ vanishes when $a>b$, as then the isolated essential singularity at the origin is the only singularity on the Riemann sphere. For $a=b$ it is $\partial_{u_1}^{a-1}\mathfrak{F}_{a-1}^{(1)}(\vec u)=1$ as the integrand can be evaluated at its residue at infinity. Both insights imply that the first $k$-columns of the determinant in~\eqref{CUE.circle.1} are a strictly upper triangular, with unity on the main diagonal ,so that we can readily Laplace expand the determinant in those columns, leading to the simplification
\begin{equation}\label{CUE.circle.2}
\begin{split}
\lim_{N\to\infty}\frac{\widetilde{\mathcal{Z}}_{d,c}^{({\rm CUE})}\left(\chi=1;\vec h\right)}{N^{k^2+2\sum_{j=1}^djh_j}}=&\lim_{\vec u\to0}\left(\prod_{n=1}^d\left[\sum_{l_1,l_2=0}^{n}\binom{n}{l_1}\binom{n}{l_2}c^{n-l_1}(1+c)^{n-l_2}(-1)^{l_2}D_{\vec u,l_1}D_{\vec u,l_2}\right]^{h_j}\right)\\
&\times\det\left[\partial_{u_1}^{k+a-1}\mathfrak{F}_{b-1}^{(2)}(\vec u)\right]_{a,b=1,\ldots,k}.
\end{split}
\end{equation}

In the case $d=1$, we find
\begin{equation}\label{CUE.circle.d1}
\begin{split}
\lim_{N\to\infty}\frac{\widetilde{\mathcal{Z}}_{1,c}^{({\rm CUE})}\left(\chi=1;k-h_1,h_1\right)}{N^{k^2+2h_1}}=&\lim_{u\to0}\left(c(c+1)+\partial_u-\partial_u^2\right)^{h_1}\det\left[u^{(b-a-k)/2}I_{k+a-b}(2\sqrt{u})\right]_{a,b=1,\ldots,k},
\end{split}
\end{equation}
where $I_j$ is the modified Bessel function of the first kind. This result equivalent to a result derived in~\cite[Theorem 1]{CRS} for $h_1=k$ and $c=0$ while it reproduces~\cite[Corollary~1.5 in combination with Theorem~1.2]{Keating-etal} for $c=-1/2$.

Thence, let us go over to the case $d=2$ where it is
\begin{equation}\label{CUE.circle.d2}
\begin{split}
&\lim_{N\to\infty}\frac{\widetilde{\mathcal{Z}}_{2,c}^{({\rm CUE})}\left(\chi=1;\vec h\right)}{N^{k^2+2h_1+4h_2}}
=\lim_{\vec u\to0}\left(c(c+1)+\partial_{u_1}-\partial_{u_1}^2\right)^{h_1}\left(\left[(c+\partial_{u_1})^2+\partial_{u_2}\right]\left[(c+1-\partial_{u_1})^2+\partial_{u_2}\right]\right)^{h_2}\\
&\quad\quad\quad\quad\quad\quad\quad\quad\quad\quad\quad\quad\times\det\left[\partial_{u_1}^{k+a-1}\mathfrak{F}_{b-1}^{(2)}(\vec u)\right]_{a,b=1,\ldots,k}.
\end{split}
\end{equation}
 Then, we need to compute
\begin{equation}
\begin{split}
\partial_{u_1}^{a-1}\mathfrak{F}_{b-1}^{(2)}(\vec u)=\oint_{|z|=1}\frac{\dif z}{2\pi i z^{a-b+1}}\exp\left[z+u_1z^{-1}+u_2z^{-2}\right].
\end{split}
\end{equation}
When employing a Hubbard--Stratonovich transformation~\cite{Hubbard1959} to linearise
\begin{equation}
e^{u_2z^{-2}}=\int_{-\infty}^\infty\frac{\dif t}{\sqrt{\pi}}e^{-t^2+2\sqrt{u_2}tz^{-1}},
\end{equation}
it becomes clear that this function is equal to
\begin{equation}
\begin{split}
\partial_{u_1}^{k+a-1}\mathfrak{F}_{b-1}^{(2)}(\vec u)=\int_{-\infty}^\infty\frac{\dif t}{\sqrt{\pi}}e^{-t^2}(u_1+2\sqrt{u_2}t)^{(b-a-k)/2}I_{k+a-b}\left(2\sqrt{u_1+2\sqrt{u_2}t}\right)
\end{split}
\end{equation}
which is a convolution of a Bessel function with a Gaussian.

\end{example}

\section{Proofs}\label{sec:proofs}

\subsection{Proof of Theorem~\ref{thm:main.result}}\label{sec:proof.main}

Let $V=\{x_a^{\gamma-1}\}_{a,\gamma=1,\ldots,P}$ be the Vandermonde matrix where $a$ is the row index and $\gamma$ the column index. Its inverse $W=V^{-1}$ will be at the heart of the proof. Especially the matrix entries of its inverse are important to get the Vandermonde determinant absorbed into the Pfaffian. Those matrix entries have the contour integral representation
\begin{equation}
W_{\alpha b}=\oint_{|z|=R}\frac{\dif z}{2\pi i z^\alpha} \prod_{l\neq b}\frac{z-x_l}{x_b-x_l}\ ,
\end{equation}
which can be derived via the cofactor matrix. The contour runs counter-clockwise and the radius $R>0$ can be chosen arbitrarily. We will, however, choose $R$ large enough such that the contour $|z|=R$ is much larger than the maximum $\max_{j=1,\ldots,P}\{x_j\}$, out of convenience for the ensuing discussion.

In the first step, we make use of the inverse $W$ and multiply it into the Pfaffian, along the identity
\begin{equation}
\det W\ {\rm Pf}\left[\begin{array}{cc} A & B\\-B^T &C\end{array}\right]= {\rm Pf}\left[\begin{array}{cc} WAW^T & WB\\-B^TW^T &C\end{array}\right].
\end{equation}
In this way we, absorb the reciprocal of the Vandermonde determinant, i.e.,
\begin{equation}
\frac{1}{\Delta_P(\vec x)}{\rm Pf}\left[\begin{array}{cc} A(x_a,x_c) & B_d(x_a) \\ -B_b(x_c) & C_{bd} \end{array}\right]_{\substack{a,c=1,\ldots,P\\b,d=1,\ldots,Q}}={\rm Pf}\left[\begin{array}{cc} \mathfrak{A}_{\alpha\gamma} & \mathfrak{B}_{\alpha d} \\ -\mathfrak{B}_{\gamma b} & C_{bd} \end{array}\right]_{\substack{\alpha,\gamma=1,\ldots,P\\b,d=1,\ldots,Q}}
\end{equation}
with
\begin{equation}
\begin{split}
\mathfrak{A}_{\alpha\gamma}=&\oint_{|z_1|=R}\frac{\dif z_1}{2\pi i z_1^\alpha}\oint_{|z_2|=R}\frac{\dif z_2}{2\pi i z_2^\gamma} \sum_{a,c=1}^P A(x_a,x_c) \left(\prod_{l\neq a}\frac{z_1-x_l}{x_a-x_l}\right)\left(\prod_{l\neq c}\frac{z_2-x_l}{x_c-x_l}\right)\\
=&\oint_{|z_1|=R}\frac{\dif z_1}{2\pi i z_1^\alpha}\oint_{|z_2|=R}\frac{\dif z_2}{2\pi i z_2^\gamma} \oint_{|\zeta_1|=\mathcal{C}}\frac{\dif \zeta_1}{2\pi i (z_1-\zeta_1)}\oint_{|\zeta_2|=\mathcal{C}}\frac{\dif \zeta_2}{2\pi i (z_2-\zeta_2)} \\
&\times A(\zeta_1,\zeta_2) \prod_{l=1}^P\frac{(z_1-x_l)(z_2-x_l)}{(\zeta_1-x_l)(\zeta_2-x_l)}
\end{split}
\end{equation}
and
\begin{equation}
\begin{split}
\mathfrak{B}_{\alpha d}=&\oint_{|z|=R}\frac{\dif z}{2\pi i z^\alpha}\sum_{a=1}^P B_d(x_a) \left(\prod_{l\neq a}\frac{z-x_l}{x_a-x_l}\right)=\oint_{|z|=R}\frac{\dif z}{2\pi i z^\alpha}\oint_{|\zeta|=\mathcal{C}}\frac{\dif \zeta}{2\pi i (z-\zeta)}B_d(\zeta) \prod_{l=1}^P\frac{z-x_l}{\zeta-x_l}.
\end{split}
\end{equation}
The second identity follows from the residue theorem where the contour $\mathcal{C}$ is chosen in such a way, that every $x_j$ is run around counter-clockwise only once and does not encircle or cross the contour $|z|=R$. Interestingly, while we started with pairwise distinct variables $\vec x=(x_1,\ldots,x_P)$, the new representation allows us to analytically continue to degenerate parameters. Actually, this is the reason why this representation is advantageous.

The contour integrals can be pulled out of the Pfaffian by reversing the generalised de Bruijn identity~\cite{DeBruijn1955}, particularly using the multilinearity of the Pfaffian, which  leads to
\begin{equation}
\begin{split}
\frac{1}{\Delta_P(\vec x)}{\rm Pf}\left[\begin{array}{cc} A(x_a,x_c) & B_d(x_a) \\ -B_b(x_c) & C_{bd} \end{array}\right]_{\substack{a,c=1,\ldots,P\\b,d=1,\ldots,Q}}
=&\left(\prod_{l=1}^P \oint_{|z_j|=R}\frac{\dif z_j}{2\pi i z_j^{j}}\oint_{|\zeta_j|=\mathcal{C}}\frac{\dif \zeta_j}{2\pi i (z_j-\zeta_j)}\right)\\
&\times\left(\prod_{l,k=1}^P\frac{z_k-x_l}{\zeta_k-x_l}\right){\rm Pf}\left[\begin{array}{cc} A(\zeta_a,\zeta_c) & B_d(\zeta_a) \\ -B_b(\zeta_c) & C_{bd} \end{array}\right]_{\substack{a,c=1,\ldots,P\\b,d=1,\ldots,Q}}.
\end{split}
\end{equation}
This expression is ideal to perform the derivatives and the limit $\vec x\to\vec \chi$. For this task, we change to the notation~\eqref{eq: paramPartition1} of the partitions and employ Lemma~\ref{lemma:der.rel} to get
\begin{equation}
\begin{split}
&\lim_{\vec x\to \vec\chi}\left(\prod_{l=1}^L\prod_{j=1}^{P_l}\partial_{x_{l,j}}^{n_{l,j}}\right)\prod_{j,o=1}^N\frac{z_j-x_o}{\zeta_j-x_o}=\lim_{\vec x\to \vec\chi}\left(\prod_{l=1}^L\prod_{j=1}^{P_l}\partial_{x_{l,j}}^{n_{l,j}}\right)\prod_{l=1}^L\prod_{o=1}^{P_l}\prod_{j=1}^P\frac{z_j-x_{l,o}}{\zeta_j-x_{l,o}}\\
=&\lim_{\mathbf{ u}\to 0}\left(\prod_{l=1}^L\prod_{k=1}^{d}D_{\vec u_l,k}^{m_{l,k}}\right)\exp\left[\sum_{l=1}^L\sum_{j=1}^P \bigl(F_d(\vec u_l,\chi_l,\zeta_j)-F_d(\vec u_l,\chi_l,z_j)\bigl)\right]\prod_{l=1}^L\prod_{j=1}^{P}\frac{(z_j-\chi_l)^{P_l}}{(\zeta_j-\chi_l)^{P_l}}.
\end{split}
\end{equation}
The differential operators in $\vec u_l$ can be pulled out of the contour integral, while the exponential and the ratio of the polynomials in $z_j$ and $\zeta_j$ can be put into the Pfaffian, again with the help of the generalised de Bruijn identity. This yields
\begin{equation}\label{main.result.a}
\begin{split}
&\lim_{\vec x\to \vec\chi}\left(\prod_{l=1}^L\prod_{j=1}^{P_l}\partial_{x_{l,j}}^{n_{l,j}}\right) \frac{1}{\Delta_P(\vec x)}{\rm Pf}\left[\begin{array}{cc} A(x_a,x_c) & B_d(x_a) \\ -B_b(x_c) & C_{bd} \end{array}\right]_{\substack{a,c=1,\ldots,P\\b,d=1,\ldots,Q}}\\
=&\lim_{\mathbf{u}\to 0}\left(\prod_{l=1}^L\prod_{k=1}^{d}D_{\vec u_l,k}^{m_{l,k}}\right){\rm Pf}\left[\begin{array}{cc} [\widehat{\mathcal{K}}_{\vec\chi,\alpha}\otimes\widehat{\mathcal{K}}_{\vec\chi,\gamma} A](\boldsymbol{u},\boldsymbol{u}) & [\widehat{\mathcal{K}}_{\vec\chi,\alpha}B_d](\boldsymbol{u}) \\ -[\widehat{\mathcal{K}}_{\vec\chi,\gamma}B_b](\boldsymbol{u}) & C_{bd} \end{array}\right]_{\substack{\alpha,\gamma=1,\ldots,P\\b,d=1,\ldots,Q}}
\end{split}
\end{equation}
with  the transformation
\begin{equation}\label{transform.a}
\widehat{\mathcal{K}}_{\vec\chi,\alpha} f(\boldsymbol u)=\oint_{|z|=R}\frac{\dif z}{2\pi i z^\alpha}\oint_{|\zeta|=\mathcal{C}}\frac{\dif \zeta\,f(\zeta)}{2\pi i (z-\zeta)}\exp\left[\sum_{l=1}^L\bigl(F_d(\vec u_l,\chi_l,\zeta)-F_d(\vec u_l,\chi_l,z)\bigl)\right]\prod_{l=1}^L\frac{(z-\chi_l)^{P_l}}{(\zeta-\chi_l)^{P_l}}.
\end{equation}
As $R$ has been chosen large enough, we know that the terms $\exp\left[-\sum_{l=1}^LF_d(\vec u_l,\chi_l,z)\right]/(z-\zeta)$
have an absolutely convergent Taylor series in $1/z$ on the contour $|z|=R$. Therefore, the integral over $z$ yields a polynomial in $\zeta$ of degree $P-\alpha$ due to the polynomial $\prod_{l=1}^L(z-\chi_l)^{P_l}$, i.e.,
\begin{equation}
\oint_{|z|=R}\frac{\dif z}{2\pi i z^\alpha(z-\zeta)}\exp\left[-\sum_{l=1}^LF_d(\vec u_l,\chi_l,z)\right]\prod_{l=1}^L(z-\chi_l)^{P_l}=\sum_{j=0}^{P-\alpha}c_j(\boldsymbol u,\chi)\zeta^j.
\end{equation}
The coefficients are explicitly given by 
\begin{equation}
c_j(\boldsymbol u,\chi)=\oint_{|z|=R}\frac{\dif z}{2\pi i z^{\alpha+j+1}}\exp\left[-\sum_{l=1}^LF_d(\vec u_l,\chi_l,z)\right]\prod_{l=1}^L(z-\chi_l)^{P_l}\ ,
\end{equation}
resulting from the geometric series of the term $1/(z-\zeta)$ in $\zeta/z$. These coefficients are evidently independent of the other contour integration variable $\zeta$, and the leading coefficient can be evaluated by taking the residue at infinity given by
\begin{equation}
\begin{split}
c_{P-\alpha}(\boldsymbol u,\chi)&=\oint_{|z|=R}\frac{\dif z}{2\pi i z^{P+1}}\exp\left[-\sum_{l=1}^LF_d(\vec u_l,\chi_l,z)\right]\prod_{l=1}^L(z-\chi_l)^{P_l}\\
&=\lim_{z\to\infty}\exp\left[-\sum_{l=1}^LF_d(\vec u_l,\chi_l,z)\right]\prod_{l=1}^L\left(1-\frac{\chi_l}{z}\right)^{P_l}=1.
\end{split}
\end{equation}
Thus, we can use the property of the Pfaffian by simultaneously adding rows and columns without changing the result, and replace these cumbersome polynomials by the leading order powers $\zeta^{P-\alpha}$. This leads to the claimed result with the transform~\eqref{transform} concluding the proof.

\subsection{Proof of Theorems~\ref{Thm-gen-det} and~\ref{Thm-gen-pfaffian}}
\label{proof:Thm-gen-unitary}

We start with proving Theorem~\ref{Thm-gen-det} and note that 
\begin{equation}
    \label{eq:FrakADefinition}
  D(\vec u,\vec v)= \frac{\det[B(u_a,{v}_b)]_{a,b=1,\ldots,k}}{\Delta_k(\vec u) \Delta_k(\vec { v})} 
\end{equation}
is a symmetric function in the entries of $\vec u= (u_1,\dots,u_k)$ and $\vec v=(v_1,\dots,v_k)$, separately.
Combining this with the holomorphicity of $B$ at $(\chi,\xi)$, $D(\vec u,\vec v)$ can be expanded in a basis of symmetric polynomials at $(\chi,\ldots,\chi)\in\mathbb{C}^{k}$ in the first set of variables, and $(\xi,\ldots,\xi)\in\mathbb{C}^{k}$ in the second set.
The most natural basis are the Schur polynomials,
\begin{equation}
    \label{eq:SchurByDet}
    S_{\vec\lambda}(\vec u) = \frac{\det\left[u_a^{\widehat{\lambda}_b}\right]_{a,b=1,\ldots, k}}{\det[u_a^{b-1}]_{a,b=1,\ldots,k}} = \frac{\det\left[u_a^{\widehat{\lambda}_b}\right]_{a,b=1,\ldots,k}}{\Delta_k(\vec u)},
\end{equation}
with $\widehat\lambda_a=\lambda_{k-a+1}+a-1$.

Schur polynomials form an orthogonal basis with respect to the Hall inner product, which is related to the Haar measure on the unitary group $\U(k)$. It can be written as the following integral over powers of the unitary group $\U^k(1)$ (the maximal Abelian subgroup of $\U(k)$), where the orthonormality relation reads
\begin{equation}
    \label{eq:SchurOrthogonality}
    \langle S_{\vec \lambda}, S_{\vec\lambda'} \rangle = \frac{1}{k!} \oint_{\U(1)^k} \prod_{a=1}^k \frac{\dif u_a}{2 \pi i u_a} \overline{S_{\vec\lambda}(\vec u)} S_{\vec\lambda'}(\vec u) |\Delta_k(\vec u)|^2 = \delta_{\vec\lambda,\vec\lambda'}.
\end{equation}
The Kronecker delta for partitions is $1$ when the two partitions are equal, and otherwise vanishes.

We will also use the fact that monomial expansion of Schur polynomials involves as its coefficients the Kostka numbers~\cite[\S 4.3]{FultonHarris2004}, i.e.,
\begin{equation}
    \label{eq:SchurByKostka}
    S_{\vec\lambda}(\vec u) = \sum_{\substack{|
    \vec\alpha| =|\vec\lambda|}} K_{\vec\lambda,\vec\alpha} \prod_{a=1}^k u_a^{\alpha_a}.
\end{equation}
In particular, this implies that
\begin{equation}
    \label{eq:SchurDerivsKostka}
    \lim_{u_1,\ldots,u_k \to 0} \frac{\partial_{u_1}^{\alpha_1}}{\alpha_1!} \ldots \frac{\partial_{u_k}^{\alpha_k}}{\alpha_k!} S_{\vec\lambda}(\vec u) = K_{\vec \lambda,\vec\alpha}.
\end{equation}

To use the above identity for limits $\vec u \to \chi$ and $\vec v\to \xi$, we apply
\begin{equation}
    \label{eq:translationidentity}
    \lim_{\vec u\to\chi,\vec v \to\xi} D(\vec u,\vec v)= 
   \lim_{\vec u,\vec v\to0}\Big( \prod_{a=1}^k \partial_{u_a}^{\alpha_a} \partial_{{v}_a}^{\beta_a} \Big) \frac{\det[B(\chi + u_a,\widetilde{\chi}+v_b)]_{a,b=1}^k}{\Delta_k(\vec u) \Delta_k(\vec{v})},
\end{equation}
where we have used the fact that the Vandermonde determinants are invariant under translation.
Expanding the ratio of determinants in~\eqref{eq:FrakADefinition} in terms of Schur polynomials with coefficients $\mathfrak{B}_{\vec\mu, \vec\lambda}$ gives
\begin{equation}
    \label{eq:FrakASchurExpansion}
    \frac{\det[B(\chi + u_a,\xi+v_b)]_{a,b=1}^k}{\Delta_k(\vec u) \Delta_k(\vec{v})}= 
    \sum_{n,n'=0}^{\infty} \sum_{\substack{\vec\mu \vdash n\\l(\vec\mu)\leq k}} \sum_{\substack{\vec\lambda \vdash n' \\ l(\vec\lambda) \leq k}} \mathfrak{B}_{\vec\mu,\vec\lambda} S_{\vec\mu}(\vec u) S_{\vec \lambda}(\vec{v}).
\end{equation}
In terms of the Hall inner product, the coefficients are given by
\begin{equation}
\begin{aligned}
    \label{eq:FrakAInnerProduct}
    \mathfrak{B}_{\vec\mu,\vec\lambda} 
    =& \frac{1}{k!^2} \left(\prod_{a=1}^k \oint_{|u_a|=\epsilon}\frac{\dif u_a}{2 \pi i u_a} \oint_{|v_a|=\epsilon}\frac{\dif {v}_a}{2 \pi i {v}_a} \right)\\
    &\times\det[{u}_a^{-\widehat{\mu}_b}]_{a,b=1,\ldots,k} \det[{v}_a^{-\widehat{\lambda}_b}]_{a,b=1,\ldots,k} \det[B(\chi+u_a,\xi+v_b)]_{a,b=1,\ldots,k}\ ,
\end{aligned}
\end{equation}
where $\epsilon>0$ is chosen small enough, so that the function $B$ is holomorphic on and inside the circular contours that are integrated counter-clockwise.
Applying Andr\'eief's integral formula~\cite[ch. 11]{LivanNovaesVivo} twice, once for each set of integrations, gives 
\begin{equation}
\begin{aligned}
    \mathfrak{B}_{\vec\mu,\vec\lambda} 
                               & = \det\left[\oint_{|u|=\epsilon}\frac{\dif u}{2 \pi i u^{\widehat{\mu}_a+1}} \oint_{|v|=\epsilon}\frac{\dif {v}}{2 \pi i {v}^{\widehat{\lambda}_b+1}} B(\chi+u,\xi+v)\right]_{a,b=1}^{k} = \det\left[\frac{\partial_\chi^{\widehat{\mu}_a}\partial_\xi^{\widehat{\lambda}_b}B(\chi,\xi)}{\widehat{\mu}_a!\widehat{\lambda}_b!}\right]_{a,b=1}^{k},
\end{aligned}
\label{proof: andreif step}
\end{equation}
where we used in the second step the residue theorem twice. Pulling the factorials out of the determinant and applying the derivatives on the Schur polynomials by using~\eqref{eq:SchurDerivsKostka} yields the claim~\eqref{del-gen-det}. We can pull the derivatives into the series as it is uniformly converging in a neighbourhood of the origin $\vec u=\vec v=0$.

The proof strategy of Theorem~\ref{Thm-gen-pfaffian}  is essentially the same only generalised to Pfaffians of the form \eqref{del-gen-pfaffian}.
Here, we study (in analogy with~\eqref{eq:FrakADefinition}) the translated functions and expand the Pfaffian in terms of Schur polynomials in $\vec x=(\chi,\ldots,\chi)+\vec u$,
\begin{equation}
    \label{eq:FrakBDefinition}
    \frac{\Pf[A(\chi+u_a,\chi+u_b)]_{a,b=1,\ldots,2k}}{\Delta_{2k}(\vec u)} = \sum_{n=0}^{\infty} \sum_{\substack{\vec\lambda \vdash n \\ l(\vec\lambda) \leq 2k}} \mathfrak{A}_{\lambda} S_{\vec\lambda}(\vec u).
\end{equation}
The coefficients $\mathfrak{A}_{\vec\lambda}$ can be anew calculated with the help of Hall's inner product and de Bruijn's integral formula~\cite[(4.7)]{DeBruijn1955},
\begin{equation}
\begin{aligned}
    \label{eq:FrakBInnerProduct}
    \mathfrak{A}_{\vec\lambda} & = \frac{1}{(2k)!} \left( \prod_{j=1}^{2k} \oint_{|u_j|=\epsilon}\frac{\dif u_j}{2 \pi i u_j}\right) \det[{u}_a^{-\widehat{\lambda}_b}]_{a,b=1,\ldots,2k}  \Pf[A(\chi+u_a,\chi+u_b)]_{a,b=1,\ldots,2k}\\
     &= \Pf\left[ \oint_{|u|=\epsilon}\frac{\dif u}{2 \pi i u^{\widehat{\lambda}_a+1}} \oint_{|v|=\epsilon}\frac{\dif {v}}{2 \pi i {v}^{\widehat{\lambda}_b+1}} A(\chi+u,\chi+v)\right]_{a,b=1,\ldots,2k}\\
     &= \Pf\left[ \frac{\partial_u^{\widehat{\lambda}_a}\partial_v^{\widehat{\lambda}_b}A(u,v)|_{u=v=\chi}}{\widehat{\lambda}_a! \widehat{\lambda}_b!}\right]_{a,b=1,\ldots,2k}.
 \end{aligned}
\end{equation}
Like before, we have employed the residue theorem in the last step.

Finally, we can pull the factorials out the Pfaffian and swap the derivatives in $\vec u$ with the series due to uniform convergence of the series in some neighbourhood of $\vec u=0$. After applying~\eqref{eq:SchurDerivsKostka}, we arrive at~\eqref{del-gen-pfaffian}, finishing the second proof.

\subsection{Proof of Theorem \ref{Thm-gen-symplectic-abs}}
\label{proof:Thm-gen-symplectic-abs}

We prove equation~\eqref{del-gen-symplectic-abs}. For this purpose, we define the abbreviations
\begin{eqnarray}
    A_{\chi\chi}(u,v) := A(\chi + u,\chi + v),\ A_{\xi \xi}(u,v) := A(\xi + u, \xi + v),\ A_{\chi \xi}(u,v) := A(\chi + u, \xi + v) \label{eq:def-f1},
\end{eqnarray}
of which $A_{\xi\xi}$ and $A_{\chi\chi}$ inherit the antisymmetry from $A$.
Next, we split the Vandermonde determinant into the terms
\begin{equation}
    \Delta_{2k}(\vec x) = \Delta_k(\vec u) \Delta_k(\vec v) \prod_{j,l=1}^k (\xi - \chi + v_l - u_j).
\end{equation}
As before, we have a function which is symmetric in $u_i$ and $v_i$, separately.
We expand it in Schur polynomials with coefficients $\mathfrak{D}_{\vec\lambda, \vec\mu}$
\begin{equation}
\begin{aligned}
    \label{eq:symplectic-abs-SchurExpan}
    \frac{1}{\Delta_k(\vec u)\Delta_k(\vec v) \prod_{j,l=1}^k (\xi-\chi + v_l - u_j)}
    &\Pf
    \begin{bmatrix}
        A_{\chi\chi}(u_a,{u}_b) & A_{\chi\xi}(u_a,v_b)\\
        -A_{\chi\xi}(u_b,v_a) & A_{\xi\xi}(v_a,v_b)\\
    \end{bmatrix}_{a,b=1,\ldots,k} \\
    &= \sum_{n,n'=0}^{\infty} \sum_{\substack{\vec\lambda \vdash n \\ l(\vec\lambda) \leq k}} \sum_{\substack{\vec\mu \vdash n' \\ l(\vec\mu) \leq k}} \mathfrak{D}_{\vec\lambda,\vec\mu} S_{\vec\lambda}(\vec u) S_{\vec\mu}(\vec v)\ ,
\end{aligned}
\end{equation}
and compute the coefficients using the Hall inner product
\begin{equation}
    \label{eq:symplectic-abs-SchurCoeff}
    \begin{aligned}
        \mathfrak{D}_{\vec\lambda,\vec\mu} =& \frac{1}{k!^2} \left( \prod_{j=1}^{k} \oint_{|u_j|=\epsilon}\frac{\dif u_j}{2 \pi i u_j }\oint_{|v_j|=\epsilon}\frac{ \dif v_j}{2 \pi i v_j} \right)
     \frac{\det[{u}_a^{-\widehat{\lambda}_b}]_{a,b=1,\ldots,k} \det[{v}_a^{-\widehat{\mu}_b}]_{a,b=1,\ldots,k}}{\prod_{j,l=1}^k (\xi - \chi + v_l - u_j)} \\
        &\times \Pf
    \begin{bmatrix}
        A_{\chi\chi}(u_a,{u}_b) & A_{\chi\xi}(u_a,v_b)\\
        -A_{\chi\xi}(u_b,v_a) & A_{\xi\xi}(v_a,v_b)\\
    \end{bmatrix}_{a,b=1,\ldots,k} .
    \end{aligned}
\end{equation}
We first concentrate on computing these coefficients, as the Kostka numbers and the restrictions of the double series in $n$ and $n'$ to the summand $n,n'=|\vec\alpha|$ follows from~\eqref{eq:SchurDerivsKostka}, i.e.,
\begin{equation}
    \label{eq:symplectic-abs-SchurCoeff.a.b}
\begin{aligned}
    \lim_{\vec x \to (\chi,\xi)} &\Big( \prod_{j=1}^k (\partial_{x_{1,j}} \partial_{x_{2,j}})^{\alpha_j} \Big) \frac{\Pf\begin{bmatrix}
            A(x_{1,a},x_{1,b}) & A(x_{1,a},x_{2,b}) \\
            A(x_{2,a},x_{1,b}) & A(x_{2,a},x_{2,a})
        \end{bmatrix}_{a,b=1,\ldots,k}}{\Delta_{2k}(\vec x)}
         =\sum_{\substack{\vec\lambda,\vec\mu \vdash |\vec\alpha| \\l(\vec\lambda), l(\vec\mu) \leq k}}(\vec\alpha!)^2 K_{\vec\lambda,\vec\alpha}K_{\vec\mu,\vec\alpha}\mathfrak{D}_{\vec\lambda,\vec\mu}\,.
\end{aligned}
\end{equation}
In the first step of evaluating the coefficients, we expand the determinants in ${u}_a^{-\widehat{\lambda}_b}$ and ${v}_a^{-\widehat{\mu}_b}$ via the Leibniz formula. For this purpose, we can exploit the skew-symmetry in the remaining terms of the integral so that
\begin{equation}
    \label{eq:symplectic-abs-SchurCoeff.b}
    \begin{aligned}
        \mathfrak{D}_{\vec\lambda,\vec\mu} =& \left( \prod_{j=1}^{k} \oint_{|u_j|=\epsilon}\frac{\dif u_j}{2 \pi i u_j^{\widehat\lambda_j+1} }\oint_{|v_j|=\epsilon}\frac{ \dif v_j}{2 \pi i v_j^{\widehat\mu_j+1}} \right)
     \frac{\Pf
    \begin{bmatrix}
        A_{\chi\chi}(u_a,{u}_b) & A_{\chi\xi}(u_a,v_b)\\
        -A_{\chi\xi}(u_b,v_a) & A_{\xi\xi}(v_a,v_b)\\
    \end{bmatrix}_{a,b=1,\ldots,k}}{\prod_{j,l=1}^k (\xi - \chi + v_l - u_j)}.
    \end{aligned}
\end{equation}
The contour integrals can be carried out with the help of the residue theorem at the poles which lie all at the origin. Yet, the tricky part is to carry out the derivative of the terms in the denominator in a systematic way. We first do it in the integration over $u_j$,  where we employ the identity
\begin{equation}
    \begin{aligned}
        &\oint_{|u_j|=\epsilon} \frac{\dif u_j}{2 \pi i u_j^{\widehat{\lambda}_{j}+1}} 
        \frac{f(u_j)}{\prod_{l=1}^k (\xi - \chi + v_l - u_j)}
      =\sum_{n_j=0}^{\widehat\lambda_j}\sum_{ r_1^{(j)} +\cdots+r_k^{(j)}=\widehat{\lambda}_{j}-n_j} 
        \frac{ \partial_{u_j}^{n_j}f(u_j)|_{u_j=0} }{n_j!\prod_{l=1}^k (\xi-\chi+v_l)^{r_l^{(j)}+1}} \ .
    \end{aligned}
\end{equation}
Here, $n_j$ is the number of times the Pfaffian is differentiated and $r_l^{(j)}$ for $1 \leq l \leq k$ is the number of times the factor $\frac{1}{\xi - \chi + v_l - u_j}$ is differentiated. We can do this computation for any $u_1,\ldots, u_k$ so that we arrive at
\begin{equation}
    \label{eq:symplectic-abs-SchurCoeff.c}
    \begin{aligned}
        \mathfrak{D}_{\vec\lambda,\vec\mu} &= \sum_{\substack{0\leq n_j\leq\widehat\lambda_j\\j=1,\ldots,k}}\sum_{\substack{r_1^{(j)} +\cdots+r_k^{(j)}=\widehat{\lambda}_{j}-n_j\\j=1,\ldots,k}} \left( \prod_{j=1}^{k} \frac{1}{n_j!}\oint_{|v_j|=\epsilon}\frac{ \dif v_j}{2 \pi i v_j^{\widehat\mu_j+1}} \right)\\
        &\times
     \frac{\Pf
    \begin{bmatrix}
        \partial_{u}^{n_a}\partial_{u'}^{n_b}A_{\chi\chi}(u,u')|_{u=u'=0} & \partial_{u}^{n_a}A_{\chi\xi}(u,v_b)|_{u=0}\\
        -\partial_{u}^{n_b}A_{\chi\xi}(u,v_a)|_{u=0} & A_{\xi\xi}(v_a,v_b)\\
    \end{bmatrix}_{a,b=1,\ldots,k}}{\prod_{j,l=1}^k (\xi - \chi + v_l )^{r_l^{(j)}+1}}.
    \end{aligned}
\end{equation}
When carrying out the integrals in $v_j$ we need to adjust the relevant Cauchy formula as follows
\begin{equation}
    \begin{aligned}
        &\oint_{|v_l|=\epsilon} \frac{\dif v_l}{2 \pi i v_l^{\widehat{\mu}_{l}+1}} 
        \frac{f(v_l)}{\prod_{j=1}^k (\xi - \chi + v_l)^{r_l^{(j)}+1}}\\
      =&\sum_{e_l=0}^{\widehat\mu_l}\sum_{s_1^{(l)} +\cdots+s_k^{(l)}=\widehat{\mu}_{l}-e_l}\left[\prod_{j=1}^k(-1)^{s_j^{(l)}}\left(\begin{array}{c} r_l^{(j)}+s_j^{(l)} \\ r_l^{(j)}\end{array}\right) \right]
       \frac{ \partial_{v_l}^{e_l}f(v_l)|_{v_l=0} }{e_l!\prod_{j=1}^k (\xi-\chi)^{r_l^{(j)}+s_j^{(l)}+1}}.
    \end{aligned}
\end{equation}
When using the fact that $\sum_{j,l=1}^ks_j^{(l)}=|\widehat{\mu}|-|\vec{e}|$ as well as $\sum_{j,l=1}^kr_l^{(j)}=|{\widehat\lambda}|-|{\vec{n}}|$, we find
\begin{equation}
    \label{eq:symplectic-abs-SchurCoeff.d}
    \begin{aligned}
        \mathfrak{D}_{\vec\lambda,\vec\mu} =& \sum_{\substack{0\leq n_j\leq\widehat\lambda_j\\j=1,\ldots,k}}\sum_{\substack{r_1^{(j)} +\cdots+r_k^{(j)}=\widehat{\lambda}_{j}-n_j\\j=1,\ldots,k}}\sum_{\substack{0\leq e_l\leq\widehat\mu_l\\l=1,\ldots,k}}\sum_{\substack{s_1^{(l)} +\cdots+s_k^{(l)}=\widehat{\mu}_{l}-e_l\\l=1,\ldots,k}}\frac{(-1)^{|{\vec\mu}|-|{\vec{e}}|}}{\vec{n}!\vec{e}!(\xi-\chi)^{|{\widehat\lambda}|-|{\vec{n}}|+|{\widehat\mu}|-|{\vec{e}}|+k^2}}\\
        &\times\left[\prod_{j,l=1}^k\binom{r_l^{(j)}+s_j^{(l)}}{r_l^{(j)}}\right]\Pf
    \begin{bmatrix}
        \partial_{u}^{n_a}\partial_{u'}^{n_b}A_{\chi\chi}(u,u')|_{u=u'=0} & \partial_{u}^{n_a}\partial_v^{e_b}A_{\chi\xi}(u,v)|_{u=v=0}\\
        -\partial_{u}^{n_b}\partial_v^{e_a}A_{\chi\xi}(u,v)|_{u=v=0} & \partial_v^{e_a}\partial_{v'}^{e_b}A_{\xi\xi}(v,v')|_{v=v'=0}
    \end{bmatrix}_{a,b=1,\ldots,k}\\
    =& \sum_{\substack{0\leq n_j\leq\widehat\lambda_j\\j=1,\ldots,k}}\sum_{\substack{0\leq e_l\leq\widehat\mu_l\\l=1,\ldots,k}}\frac{(-1)^{|{\vec\alpha}|-q'}}{\vec{n}!\vec{e}!(\xi-\chi)^{2|{\vec\alpha}|+k^2-q-q'}}\mathcal{A}\left(\begin{array}{c} \widehat{\lambda}-\vec{n}\\ \widehat{\mu}-\vec{e} \end{array}\right)\\
        &\times\Pf
    \begin{bmatrix}
        \partial_{u}^{n_a}\partial_{u'}^{n_b}A_{\chi\chi}(u,u')|_{u=u'=0} & \partial_{u}^{n_a}\partial_v^{e_b}A_{\chi\xi}(u,v)|_{u=v=0}\\
        -\partial_{u}^{n_b}\partial_v^{e_a}A_{\chi\xi}(u,v)|_{u=v=0} & \partial_v^{e_a}\partial_{v'}^{e_b}A_{\xi\xi}(v,v')|_{v=v'=0}
    \end{bmatrix}_{a,b=1,\ldots,k}.
    \end{aligned}
\end{equation}
We set here $q=|\vec{n}|-k(k-1)/2 \geq 0$ and $q'=|\vec{e}|-k(k-1)/2 \geq 0 $, and used the fact that $|\vec\lambda|=|\vec\mu|=|\vec\alpha|$. 
If $\vec{n}$ contains two equal entries, the Pfaffian vanishes, so we consider only $\vec{n}$ with distinct entries, and the same applies to $\vec{e}$.
This ensures $q,q' \geq 0$.
We can extend the sums over $\vec{n}$ and $\vec{e}$ to $\mathbb{N}_0^k$ by noticing that the coefficients $\mathcal{A}\left(\begin{array}{c} \widehat{\lambda}-\vec{n}\\ \widehat{\mu}-\vec{e} \end{array}\right)$ vanish
as the sums run over empty sets, when $n_j>\widehat\lambda_j$ or $e_l>\widehat\mu_l$ for some $j$ or $l$. In this way, the sums become independent of $\widehat{\lambda}$ and $\widehat{\mu}$.
Especially, the sets are invariant under permuting each of the entries of $\vec{n}$ and $\vec{e}$. While the Pfaffian is skew-symmetric under such permutations and the factors $\vec{n}!\vec{e}!$ are symmetric,
we can anti-symmetrise the remaining factor $\mathcal{A}\left(\begin{array}{c} \widehat{\lambda}-\vec{n}\\ \widehat{\mu}-\vec{e} \end{array}\right)$.
This leads to
\begin{equation}
    \label{eq:symplectic-abs-SchurCoeff.e}
    \begin{aligned}
        \mathfrak{D}_{\vec\lambda,\vec\mu} =& \sum_{\vec{n},\vec{e} \in\mathbb{N}_0^k}\frac{(-1)^{|{\vec\alpha}|-q'}}{(k!)^2\vec{n}!\vec{e}!(\xi-\chi)^{2|{\vec\alpha}|+k^2-q-q'}}\widetilde{\mathcal{A}}\left(\begin{array}{c} \widehat\lambda,\vec{n}\\ \widehat\mu,\vec{e} \end{array}\right)\\
        &\times\Pf
    \begin{bmatrix}
        \partial_{u}^{n_a}\partial_{u'}^{n_b}A_{\chi\chi}(u,u')|_{u=u'=0} & \partial_{u}^{n_a}\partial_v^{e_b}A_{\chi\xi}(u,v)|_{u=v=0}\\
        -\partial_{u}^{n_b}\partial_v^{e_a}A_{\chi\xi}(u,v)|_{u=v=0} & \partial_v^{e_a}\partial_{v'}^{e_b}A_{\xi\xi}(v,v')|_{v=v'=0}
    \end{bmatrix}_{a,b=1,\ldots,k},
    \end{aligned}
\end{equation}
with the notation~\eqref{eq:def-mathcal-A-tilde}.

In the last step, we interchange the sums over $\vec\lambda$ and $\vec\mu$ in~\eqref{eq:symplectic-abs-SchurCoeff.a.b} with all the other sums.
Also, as the expression is symmetric under permutations of the components of $\vec{n}$, and symmetric under permutations of the components of $\vec{e}$,
we replace $\vec{n}$ by the shifted partition $\widehat{\nu}$, $\vec{e}$ by the shifted partition $\widehat{\eta}$, and rewrite
\begin{equation}
 \frac{1}{(k!)^2}\sum_{\vec{n},\vec{e}\in\mathbb{N}_0^k}=
    \sum_{0 \leq q,q' \leq |\vec\alpha|} \sum_{\substack{\vec\nu \vdash q \\ l(\vec\nu) \leq k}} \sum_{\substack{\vec\eta \vdash q' \\ l(\vec\eta) \leq k}}
\end{equation}
by noting that the series stop beyond a point since it must be $q=|\vec\nu|\leq|\vec\lambda|=|\vec\alpha|$ and $q'=|\vec\eta|\leq|\vec\mu|=|\vec\alpha|$.
Anew, we used the fact that summands with $n_a=n_b$ or $e_a=e_b$ vanish due to the Pfaffian as well as the ordering~\eqref{ordering}. This completes the proof.


\section{Summary and open questions}\label{conclusio}

In the present work, we have provided explicit combinatorial and differential operator expressions for limits of ratios of a determinant
or Pfaffian determinant and a Vandermonde determinant. The study of such expressions was motivated by expectation values of characteristic
polynomials in random matrix theory, which are known to take such a form at finite and infinite matrix size, valid for general classes of
ensembles with unitary, orthogonal or symplectic symmetry. As examples, these expectation values in turn find applications in the low energy limit of the Dirac operator spectrum in quantum field theory, or the non-trivial zeros of the Riemann-$\zeta$ function. 

Perhaps not surprisingly, the expressions involving Pfaffians were more general, as they give the derivatives of ratios of determinants as special
case. We started in stating the most general expression for mixed derivatives of higher order at various points of a Pfaffian divided by a
Vandermonde determinant. We circumvented the difficulty to differentiate the inverse powers coming from the latter, that eventually have to cancel to
yield a multinomial, by replacing the inverse Vandermonde determinant by an integral representation. 

When specialising to first order derivatives at a single point, we uncovered closed form determinantal expressions. They involve derivatives of the Borel
transform of the original matrix inside the determinant that is differentiated, times some combinatorial factors. In an application to the CUE this
allowed us to recover and explain known results~\cite{CRS,SimmWei,Keating-etal,KW1,KW2} in terms of the modified Bessel function of first kind. For higher order derivatives in a single
point, we provided more involved expressions of representation theoretic origin, including so-called Kostka numbers.

The second example from random matrix theory we gave was the infinite dimensional complex Ginibre ensemble, where relatively simple and explicit
expressions were provided up to mixed moments with two derivatives. Interestingly, most of these results can be analytically continued in the total power of the moments, as it was already observed~\cite{SimmWei,Keating-etal,KW1,KW2} in the case of the CUE.

Several interesting open problems remain to be studied. For instance, our results have opened possibilities to study  the low energy Dirac operator spectrum from a new perspective and to provide analytical expression for susceptibilities and other quantities. Moreover, we know that expectation values of ratios rather than products of characteristic polynomials lead to ratios of determinantal or Pfaffian structures and mixtures of Vandermonde and Cauchy determinants, see~\cite{BS,KG09a,KG09,AP,Bergere}. Limits of derivatives of such combinations are also of interest, as a recent study in topological properties of random matrix fields~\cite{BHWGG,HKGG1,HKGG2,HKGG3} has shown. Similar techniques may apply, though other challenges may occur as well.
Also here, analytical expressions at finite matrix size exist for several classes of ensembles as a starting point, before taking derivatives and limits
thereof. These contain for instance the kernels of Cauchy transforms of orthogonal polynomials, in addition to their ordinary kernels. It remains to be seen if
closed form expressions of combinatorial nature also prevail in this setup.

\section*{Acknowledgments}
This work was partly funded by the
German Research Foundation DFG with  grant SFB 1283/2 2021 -- 317210226 (Akemann, Padellaro),  by the Australian Research Council via the
Discovery Project grant DP250102552 (Kieburg), by the Alexander-von-Humboldt Foundation (Kieburg), and by the Department of Atomic Energy, Government of India, under
project no. RTI4019 (Adrian). Kieburg also thanks the ZiF at Bielefeld University for their hospitality during his stay in October 2024 when this project has started. We thank Jon Keating, Nick Simm, and Fei Wei for illuminating discussions and comments on the manuscript. In particular Fei Wei as well as Alexander Grover, Francesco Mezzadri and Nick Simm are thanked for sharing unpublished work with us.
We also thank Jeanne Scott for pointing out the connection to the Bell polynomials in a preprint version of this work.

\begin{appendix}
\section{Expectation values of characteristic polynomials in unitary, symplectic and orthogonal ensembles}\label{App:results-OUSE}\label{sec:AppA}

In this appendix we provide a brief summary of existing results for expectation values of characteristic polynomials in ensembles of non-Hermitian random matrices with unitary, symplectic and orthogonal symmetry. These are based on~\cite{AV03,AB06,KG09,APS} to where we refer for details. 
Let us emphasise that these identities hold {\it without} specifying details of the ensemble, apart from the form of the jpdf that results from the underlying symmetry. 
In particular, they also hold for chiral and elliptic ensembles, and thus the Hermitian limit can be taken to obtain result for the corresponding Hermitian ensembles, too. While in the unitary symmetry class this is straightforward, in the symplectic class a Taylor expansion of the kernel inside the Pfaffian has to be made, prior to taking derivatives and letting the arguments of the characteristic polynomial become degenerate. Details of this limit in the symplectic class have been given in~\cite{AB06}. In the orthogonal class the Hermitian limit is more subtle~\cite{APS}.

We use the notation introduced in Subsection \ref{subsec:setup} and begin with the unitary symmetry class. 
Under the condition~\eqref{mom} on the weight function $w(z)$ we can construct monic planar orthogonal polynomials via the Heine formula~\cite{Mehta2004},
\begin{equation}
\label{OPN}
P_N(u)=u^N+\ldots=\Big\langle  D_N(u) \Big\rangle_{\rm U} 
\end{equation}
 that satisfy
\begin{equation}
\label{OPdef}
\int_\mathbb{C}d^2z_j\ w(z) P_j(z)\overline{P_k(z)}=h_j\delta_{j,k}\quad \mbox{for}\ j,k=0,1,\ldots
\end{equation}
Here, the $h_j$ denote the squared norms of the polynomials. 
The expression in terms of an expected characteristic polynomial follows from~\cite{AV03}.
We also give the polynomial part of the kernel of orthonormalised polynomials as
\begin{equation}
\label{KNdef}
\mathfrak{K}_{N+1}(u,\bar{v})=\sum_{j=0}^{N}h_j^{-1}P_j(u)\overline{P_j(v)}\ =\ h_N^{-1}\left\langle  D_N(u) \overline{D_N(v)}
\right\rangle_{\rm U}.
\end{equation} 
The second relation expressing the kernel as an expectation value of one characteristic polynomial and one conjugated one again follows from~\cite{AV03}.
The unitary ensembles with jpdf \eqref{PNU} are a determinantal point process and 
all complex eigenvalue correlation functions can be expressed in terms of this kernel.

Furthermore, the polynomials \eqref{OPN} and kernel \eqref{KNdef} form the building blocks to express the most general mixed expectation values of products of characteristic polynomials. 
The following theorem was proven in~\cite[Theorem 1]{AV03}. 
\begin{theorem}[Unitary class]
Let $\{u_i;\, i=1,\ldots,k\}$ and $\{v_i;\, i=1,\ldots,\ell\}$ be two
sets of complex numbers which are pairwise distinct within each set.
Without loss of generality we assume $k\geq \ell$, where the empty set
with $\ell=0$ is allowed.  Then, the following
statement holds
\begin{equation}
\left\langle  \prod_{i=1}^kD_N(u_i) \prod_{j=1}^\ell \overline{D_N(v_j)}
\right\rangle_{\rm U}
\ =\  \frac{\prod_{i=N}^{N+k-1}h_i^{\frac12}\ 
\prod_{j=N}^{N+\ell-1} h_j^{\frac12}}{\Delta_k(u)\ \Delta_\ell(\bar{v})}
\det\left[
\mathfrak{K}_{N+\ell}(u_a, \bar{v}_b),\  P_{N+d-1}(u_a)
\right]_{\substack{a = 1,\dots,k \\ b=1,\dots,l\\\ d=l+1,\dots,k}} ,
\label{ThmOP}
\end{equation}
\end{theorem}
For $k=\ell$ this leads to \eqref{productUk} with $C_{N,k}^{\rm (U)}=\prod_{i=N}^{N+k-1}h_i$.
Examples where the polynomials and their norms are explicitly known include planar Hermite, Laguerre and Gegenbauer polynomials, see ~\cite{PeterSungsoo,AEP} and references therein. Note that for the latter ensemble to date no matrix representation is known.

Let us turn to
random matrix ensembles with symplectic (S) and orthogonal (O) symmetry. 
We begin with the symplectic class, as the jpdf \eqref{PNS} is much simpler, containing only complex conjugate eigenvalues pairs. 
In~\cite{AEP}
a Gram--Schmidt construction of monic planar skew-orthogonal polynomials $q_j^{\rm (S)} (z)= z^j + \ldots$ was made explicit, see~\cite{Kanzieper01} for earlier work.
They satisfy the following skew-orthogonality conditions 
\begin{eqnarray}
	\langle q_{2 k}^{\rm (S)} ,q_{2 l}^{\rm (S)} \rangle_{\rm S} = \langle q_{2 k + 1}^{\rm (S)} ,q_{2 l + 1}^{\rm (S)} \rangle_{\rm S}\  =\ 0,\qquad\langle q_{2 k}^{\rm (S)} ,q_{2 l + 1}^{\rm (S)} \rangle_{\rm S} = -\langle q_{2 l + 1}^{\rm (S)} ,q_{2 k}^{\rm (S)} \rangle_{\rm S} \ \, = \ r_k^{\rm (S)} \delta_{k, l}, \label{eq:SOPdef2}
\end{eqnarray}
with $r_k^{\rm (S)}  > 0$ being their skew-norms. Here, the anti-symmetric skew-product with respect to the real positive weight $w(z)$ under the condition \eqref{mom} is defined as
\begin{equation}\label{eq:skewproddef}
	\langle f,g\rangle_{\rm S}:=
		\int_{\mathbb{C}}d^2z\ w(z)(z-\bar{z})
		(f(z) \overline{g(z)} - g(z) \overline{f(z)}).
\end{equation}
For such weights it is sufficient to consider $f,g$ to be polynomials with real coefficients. 
The monic skew-orthogonal polynomials can be made unique by fixing the second order coefficient of the odd polynomials to zero, and we refer to~\cite{AEP} for details.
Similarly to the unitary class, they can be represented in term of expected characteristic polynomials via Heine-like formulas, 
\begin{equation}
\label{SOP2N}
q_{2N}^{\rm (S)} (u)=\Big\langle  D_{N}(u) \Big\rangle_{\rm S}\ , 
\end{equation}
compare~\eqref{OPN},  and a similar relation holds for the odd polynomials including an extra trace, see~\cite{Kanzieper01,AB06,AKP}. Likewise, the corresponding anti-symmetric skew-kernel given in terms of the skew-orthogonal polynomials reads
\begin{equation}\label{prek-def}
	\kappa_N^{\rm (S)} (u, v)
=\sum_{k = 0}^{N - 1} (r_k^{\rm (S)})^{-1}(
			q_{2 k + 1}^{\rm (S)} (u) {q_{2 k}^{\rm (S)} ({v})} - q_{2 k}(u) {q_{2 k + 1}^{\rm (S)} ({v})})\ =\ 
						(r_N^{\rm (S)})^{-1}(u-v)\Big\langle  D_N(u) {D_N(v)}
\Big\rangle_{\rm S}.
\end{equation}
It also enjoys a representation as an expectation value of two characteristic polynomials~\cite{AB06}. 
The symplectic ensembles form a Pfaffian point process and all complex eigenvalue correlation functions can be expressed in terms of this skew-kernel. 
The known examples of Hermite, Laguerre and Gegenbauer polynomials from the unitary ensembles then translate into the corresponding examples for skew-orthogonal polynomials, see \cite{PeterSungsoo} for more details.

As in the unitary ensembles also all expectation values of characteristic polynomials in the symplectic class can be expressed in term of the corresponding skew-kernel and skew-orthogonal polynomials. Perhaps surprisingly, the very same expression holds for the orthogonal class. 
This is despite the fact that there, the jpdf is much more complicated, because a real non-symmetric random matrix can have real and complex conjugated eigenvalues pairs. 
For the corresponding jpdf distinguishing real and complex eigenvalues pairs, or the factorised form, we refer to~\cite{Peter-Taro} and~\cite{HJS07} respectively.

For the computation of products of expectation values of characteristic polynomials in orthogonal ensembles (O) we do not need to construct the skew-orthogonal polynomials with respect to a skew-product on the real line, but only those in the complex plane. It turns out, that these polynomials and their corresponding kernel enjoy the very same relations as \eqref{SOP2N} and \eqref{prek-def} (right equation), replacing (S) by (O), cf.~\cite{AKP,KG09,APS}.

The following theorem was derived in~\cite{AB06} for the symplectic class (S)\footnote{Note the missing prefactor in the summary of results section 2 in~\cite{AB06}. For the correct prefactors see the statement of Theorem 1 in section 3.} and in~\cite{AKP} (see section 5.3 treating the chiral class). For a discussion of $N$ even or odd in the orthogonal symmetry class (O) we refer to~\cite{KG09}.
\begin{theorem}[Orthogonal and Symplectic class]
Let $\{u_i;\, i=1,\ldots,k\}$ be pairwise distinct complex numbers.
Then, it holds for the orthogonal respectively symplectic class ${\rm C=O,S}$:
\begin{equation}
    \left\langle \prod_{j=1}^kD_N(u_j)\right\rangle_{\rm C} = \frac{(-1)^{[k/2]}\prod_{i=N}^{N+[k/2]-1}r_i^{\rm (C)} }{\Delta_k(u)}
\left\{    
\begin{array}{ll}
{\rm Pf}\left[\kappa_{N+[k/2]}^{\rm (C)}(u_i,u_j) \right]_{i,j = 1,\dots, k}  &{\rm  even}\ k\\
{\rm Pf}
\begin{bmatrix}
\kappa_{N+[k/2]}^{\rm (C)}(u_i,u_j) & q_{2N+k-1}^{\rm (C)} (u_i)\\
-q_{2N+k-1}^{\rm (C)} (u_j)& 0\\
\end{bmatrix}
_{i,j = 1,\dots, k}  &{\rm odd}\ k\\
\end{array}
\right. .
\end{equation}
\end{theorem}
From this the constant in \eqref{productSk1} follows, $C_{N,k}^{\rm (S)}=(-1)^{[k/2]}\prod_{i=N}^{N+[k/2]-1}r_i^{\rm (S)}$.
When choosing $k=2m$ to be even and $u_{j+m}=\bar{u}_j$ for $j=1,\ldots,m$, we obtain as a corollary an expression for the product of the modulus square of characteristic polynomials, valid for any $m$ even or odd. 
\begin{corollary} 
For $m\in\mathbb{N}$ and $u_1,\ldots,u_m\in \mathbb{C}$ pairwise distinct it holds for the orthogonal respectively symplectic class ${\rm C=O,S}$:
\begin{equation}
\label{productSk2Apx}
\left\langle \prod_{j=1}^{m}|D_N(u_j)|^2\right\rangle_{\rm C}=
\frac{(-1)^{m}\prod_{i=N}^{N-1+m}r_i^{\rm (C)} 
}{\Delta_{2m}(u,\bar{u})}
\Pf
\begin{bmatrix}
\kappa_{N+m}^{\rm (C)} (u_i,{u}_j)&\kappa_{N+m}^{\rm (C)} (u_i,\bar{u}_j)\\
\kappa_{N+m}^{\rm (C)} (\bar{u}_i,{u}_j) & \kappa_{N+m}^{\rm (C)} (\bar{u}_i,\bar{u}_j)\\
\end{bmatrix}
_{i,j=1}^{m}\ .
\end{equation}
\end{corollary}

There are generalisations of these results for mixtures of orthogonal and skew-orthogonal polynomials as discussed in~\cite{Mixing}. Such generalisations play an important role in random matrices modelling lattice Quantum Chromodynamics~\cite{AN,KVZ}.


\section{Bell Polynomial Derivative Operator}
\label{apx: bell poly}
In this appendix we prove that the derivative operator in \eqref{der.rel} is given by complete Bell polynomials \cite{Bell34} in derivatives
\begin{equation}
    D_{\vec u, k} = B_{k}(\partial_{u_1}, \dots,\partial_{u_k}).
\end{equation}
For this, define the generating function
\begin{equation}
    \begin{aligned}
        G(t) &\coloneqq \sum_{k=0}^\infty \frac{t^k}{k!}\partial^k_x \prod_{j=1}^N \frac{z_j-x}{\zeta_j -x} \\ &=\lim_{\vec u \to 0}\sum_{k=0}^\infty \frac{t^k}{k!}
    D_{\vec u,k} \exp\left[\sum_{j=1}^N \left(F_d(\vec u, x, \zeta_j) - F_d(\vec u, x, z_j)\right)\right] \prod_{j=1}^N \frac{z_j-x}{\zeta_j -x}.
    \end{aligned}
\end{equation}
The sum of derivatives on the l.h.s. is a translation operator and we have
\begin{equation}
    G(t) = \prod_{j=1}^N \frac{z_j-x - t }{\zeta_j -x - t}.
\end{equation}
We now re-write this in a form closer to the known generating function for the complete Bell polynomials.
For this, take the logarithm of the r.h.s.
\begin{equation}
    \log G(t)=
    \log \prod_{j=1}^N \frac{z_j - x}{\zeta_j -x} \prod_{j=1}^N \frac{1 - \frac{t}{z_j-x} }{1 - \frac{t}{\zeta_j-x}} = \sum_{j=1}^N \left[\log \frac{z_j-x}{\zeta_j-x} + \log \left(1 - \frac{t}{z_j-x}\right) - \log\left(1-\frac{t}{\zeta_j-x}\right)\right] \, ,
\end{equation}
and expand the last two terms to get
\begin{equation}
    \log G(t) = \sum_{j=1}^N \left(\log \frac{z_j-x}{\zeta_j-x} + \sum_{i=1}^\infty \frac{t^i}{i}\left[\frac{1}{(\zeta_j-x)^i}-\frac{1}{(z_j-x)^i}\right]\right)\, .
\end{equation}
We recover $G(t)$ by exponentiating
\begin{equation}
    G(t) = \prod_{j=1}^N \frac{z_j-x}{\zeta_j-x}\exp\left(\sum_{i=1}^\infty \frac{t^i}{i}\sum_{j=1}^N \left[\frac{1}{(\zeta_j-x)^i}-\frac{1}{(z_j-x)^i} \right]\right) \, .
\end{equation}

Secondly, we evaluate 
\begin{equation}
    \lim_{\vec u \to 0}\sum_{k=0}^\infty \frac{t^k}{k!}B_{k}(\partial_{u_1},\dots,\partial_{u_k}) 
\exp \left[\sum_{j=1}^N \left(F_d(\vec u, x, \zeta_j) - F_d(\vec u, x, z_j)\right)\right]
\end{equation}
by replacing the generating function for complete Bell polynomials with its known closed-form expression \cite[eq. (4.3)]{Bell34}, which follows from a special case of Fa\`a di Bruno's formula \cite{Faa}
\begin{equation}
    \lim_{\vec u \to 0}\exp\left(\sum_{i=1}^\infty \frac{t^i}{i!}\partial_{u_i}\right) \exp \left[\sum_{j=1}^N \left( F_d(\vec u, x, \zeta_j) - F_d(\vec u, x, z_j) \right)\right] \, .
\end{equation}
Since
\begin{equation}
    \partial_{u_i} \sum_{j=1}^N \left(F_d(\vec u, x, \zeta_j) - F_d(\vec u, x, z_j)\right) = \sum_{j=1}^N \left[\frac{(i-1)!}{(\zeta_j-x)^i} - \frac{(i-1)!}{(z_j-x)^i}\right]
\end{equation}
we have
\begin{equation}
    \lim_{\vec u \to 0}\exp\left(\sum_{i=1}^\infty \frac{t^i}{i!}\partial_{u_i} \right) \exp\left[\sum_{j=1}^N \left( F_d(\vec u, x, \zeta_j) - F_d(\vec u, x, z_j)\right)\right] = 
    \exp\left(\sum_{i=1}^\infty \frac{t^i}{i}\sum_{j=1}^N \left[\frac{1}{(\zeta_j-x)^i} - \frac{1}{(z_j-x)^i}\right]\right) .
\end{equation}

To summarise we have the following equality of generating functions
\begin{equation}
\begin{aligned}
    G(t) &= \sum_{k=0}^\infty \frac{t^k}{k!}\partial^k_x \prod_{j=1}^N \frac{z_j-x}{\zeta_j -x} \\
    &= \lim_{\vec u \to 0}\sum_{k=0}^\infty \frac{t^k}{k!}B_{k}(\partial_{u_1},\dots,\partial_{u_k})
     \exp\left[\sum_{j=1}^N \left(F_d(\vec u, x, \zeta_j) - F_d(\vec u, x, z_j)\right)\right]\prod_{j=1}^N \frac{z_j-x}{\zeta_j -x}\, .
\end{aligned}
\end{equation}
Comparing the coefficients of $t^{k}$ and the definition of $D_{\vec u, k}$ in \eqref{der.rel} gives
\begin{equation}
    D_{\vec u, k} = B_{k}(\partial_{u_1}, \dots, \partial_{u_k}).
\end{equation}

\section{The complex Ginibre ensemble for $N \to \infty$} \label{apx:GinUE}

In this appendix, we want to prove the limit~\eqref{eq:thmGinUE} of moments of derivatives of characteristic polynomials of the  complex Ginibre ensembe (Gin) which we denote by
\begin{equation}
    \label{eq:GinUEDefG}
   \widehat{\mathcal{Z}}_{\vec\alpha}^{\rm (Gin)}(\chi):=\lim_{N\to\infty}\frac{\mathcal{Z}_{d,2}^{\rm (Gin)}\left(\chi,\overline{\chi};\begin{array}{c}\vec{m}\\\vec m\end{array}\right)}{\prod_{j=N}^{N+k-1} \pi j!}= \lim_{N\to\infty}\frac{\left\langle \prod_{j=1}^k |\partial_{\chi}^{\alpha_j}D_N(\chi)|^2 \right\rangle_{\rm Gin}}{\prod_{j=N}^{N+k-1} \pi j!}\ ,
\end{equation}
with $d=\max_{j}\{\alpha_j\}$ and $m_{l}=\sharp\{\alpha_j=l|j=1,\ldots,k\}$ the number of $l$th derivatives.

We start from the mixed derivatives of the kernel $\mathfrak{K}_{\infty}^{\rm (Gin)}(u, \bar{v})$, see~\eqref{eq:KGinUELim}, which have the form, equivalent to~\eqref{eq:KGinUELimD},
\begin{equation}
    \label{eq:KGinUELimD.b}
    \partial_u^a\partial_{v}^b\mathfrak{K}_{\infty}^{\rm (Gin)}(u, v) = \frac{a! b! }{\pi}\ e^{u v}\sum_{r=0}^{b} \frac{u^{r} v^{r+a-b}}{(a-b+r)! (b-r)! r!}= \frac{a! b! }{\pi}\ e^{u v}\sum_{r=0}^{b} \frac{u^{b-r} v^{a-r}}{(a-r)! (b-r)! r!},
\end{equation}
which can be derived by the general Leibniz product rule for derivatives. In the second equality we reflected $r\to b-r$.
The final summation index can be sent to $\infty$ for convenience, though only terms with $0\leq r \leq \min\{a,b\}$ are in fact non-zero due to the factorials in the denominator.

The computation of~\eqref{eq:GinUEDefG} is similar to the one in Example~\ref{exmaple.1}. We substitute the derivatives into~\eqref{del-gen-det} to get
\begin{equation}
\begin{split}
        \widehat{\mathcal{Z}}_{\vec\alpha}^{\rm (Gin)}(\chi)=&\vec{\alpha}!^2 \sum_{\substack{\vec{\lambda},\vec{\mu} \vdash |\vec\alpha| \\ l(\vec{\lambda}),l(\vec{\mu}) \leq k}} 
        \frac{K_{\vec{\mu},\vec{\alpha}} K_{\vec{\lambda},\vec{\alpha}} }{\widehat{\mu}!\ \widehat{\lambda}!} \det\left[
            \frac{\widehat{\lambda}_a! \widehat{\mu}_b! }{\pi}\ e^{|\chi|^2}\sum_{r=0}^{\infty} \frac{\chi^{\widehat{\mu}_b-r} \bar{\chi}^{\widehat{\lambda}_a-r}}{(\widehat{\lambda}_a-r)! (\widehat{\mu}_b-r)! r!}
        \right]_{a,b=1,\ldots,k}\ .\\
\end{split}
\end{equation}
The common factors $e^{|\chi|^2}/\pi$ will be pulled out. Furthermore, we apply the reverse of the Andr\'eief identity~\cite{Andreief} for sums and arrive at
\begin{equation}
\begin{split}
    \label{eq:GinUEDetToSimplfy}
    \widehat{\mathcal{Z}}_{\vec\alpha}^{\rm (Gin)}(\chi)=&\frac{\vec{\alpha}!^2}{k!\pi^k}e^{k|\chi|^2} \sum_{\substack{\vec{\lambda},\vec{\mu} \vdash |\vec\alpha| \\ l(\vec{\lambda}),l(\vec{\mu}) \leq k}} 
        \frac{K_{\vec{\mu},\vec{\alpha}} K_{\vec{\lambda},\vec{\alpha}} }{\widehat{\mu}!\ \widehat{\lambda}!}\sum_{r_1,\ldots,r_k=0}^\infty\frac{1}{\prod_{j=1}^kr_j!} \\
        &\times\det\left[\frac{\widehat{\mu}_b!\chi^{\widehat{\mu}_b-r_a} }{(\widehat{\mu}_b-r_a)!}\right]_{a,b=1,\ldots,k}\det\left[\frac{\widehat{\lambda}_b!\bar\chi^{\widehat{\lambda}_b-r_a} }{(\widehat{\lambda}_b-r_a)!}\right]_{a,b=1,\ldots,k}.
\end{split}
\end{equation}
When pulling the factors of $\chi$ and $\bar\chi$ out of the determinant we can readily identify the terms $ \Delta_k(\widehat\mu)t_{\vec r}(\widehat\mu)$ and $\Delta_k(\widehat\lambda)t_{\vec r}(\widehat\lambda)$ with the corresponding Vandermonde determinants and $\vec r=(r_1,\ldots,r_k)$, see~\eqref{eq:DefineFactorialSchur}. Moreover, we can use $\sum_{a=1}^k(\widehat{\mu}_k-r_a)=|\vec\alpha|+\sum_{j=1}^k(j-1-r_j)$ and similar for $\widehat{\lambda}$ where we recall that $|\vec\mu|=|\vec\lambda|=|\vec\alpha|$, and the relation between $\vec\mu$ and $\widehat{\mu}$ as well as $\vec\lambda$ and $\widehat{\lambda}$ is given by~\eqref{lambda.hat.rel}. Thence, it is
\begin{equation}
\begin{split}
    \widehat{\mathcal{Z}}_{\vec\alpha}^{\rm (Gin)}(\chi)=&\frac{\vec{\alpha}!^2}{k!\pi^k}e^{k|\chi|^2} \sum_{\substack{\vec{\lambda},\vec{\mu} \vdash |\vec\alpha| \\ l(\vec{\lambda}),l(\vec{\mu}) \leq k}} 
        \frac{K_{\vec{\mu},\vec{\alpha}} K_{\vec{\lambda},\vec{\alpha}} }{\widehat{\mu}!\ \widehat{\lambda}!}\sum_{r_1,\ldots,r_k=0}^\infty\frac{|\chi|^{2|\vec\alpha|+2\sum_{j=1}^k(j-1-r_j)}}{\prod_{j=1}^kr_j!}\Delta_k(\widehat\mu)t_{\vec r}(\widehat\mu)\Delta_k(\widehat\lambda)t_{\vec r}(\widehat\lambda).
\end{split}
\end{equation}
The sum over $\vec\mu$ and $\vec\lambda$ can then be pulled into the sum over $\vec r$, yielding the same contribution as the two sums factorise. 

Furthermore, the summands are symmetric under permuting $r_1,\ldots, r_k$ and vanish whenever two summing indices are equal. Therefore, we may order $r_1<r_2<\ldots<r_k$ giving a factor $k!$, which cancels with the prefactor. The gives the partition $\vec{\nu}$ with the identification $\widehat\nu_j=r_j$, see~\eqref{lambda.hat.rel} for the relation between $\vec\nu$ and $\widehat\nu$. When denoting the weight of $\vec\nu$ by $m=|\vec\nu|$ we arrive at
\begin{equation}
\begin{split}
    \widehat{\mathcal{Z}}_{\vec\alpha}^{\rm (Gin)}(\chi)=&\frac{\vec{\alpha}!^2}{\pi^k}e^{k|\chi|^2} \sum_{m=0}^\infty\sum_{\vec\nu\vdash m,\ l(\vec\nu)\leq k}\frac{|\chi|^{2|\vec\alpha|-2m}}{\widehat\nu!} \left(\sum_{\vec{\lambda} \vdash |\vec\alpha| ,\ l(\vec{\lambda}) \leq k}
        \frac{K_{\vec{\lambda},\vec{\alpha}}  }{\widehat{\lambda}!}\Delta_k(\widehat\lambda)t_{\vec \nu}(\widehat\lambda)\right)^2.
\end{split}
\end{equation}
Actually, the sum in $m$ must truncate at $|\vec\alpha|$, because then $t_{\vec \nu}(\widehat\lambda)$ has at least one vanishing row. This is clear as the sum must  be a polynomial in $\chi$ as well as $\overline\chi$, which is only satisfied for $m\leq|\vec\alpha|$. Once we reflect $m\to|\alpha|-m$ we find the desired result~\eqref{eq:thmGinUE}.

\subsection{Restriction to first derivatives} \label{subsec:GinUEFirstDeriv}

For the case of $l$ first derivatives and $h$ non-differentiated factors, we have 
\begin{equation*}
    \vec{\alpha} = (1,\dots,1,0,\dots,0),
\end{equation*}
an integer vector with $l + h = k$ entries.
In this case, the Kostka numbers compute the dimensions of irreducible representations of the symmetric group $S_l$ which have a determinantal expression in terms of the Vandermonde determinant~\cite[\S 4.1]{FultonHarris2004},
\begin{equation}
    K_{\vec{\lambda},\vec{\alpha}} =  \frac{l!\Delta_k(\widehat{\lambda})}{\widehat{\lambda}!} \eqqcolon f_{\vec{\lambda}}\ .
\end{equation}
When we plug this into~\eqref{eq:thmGinUE}  we find
\begin{equation}
    \label{eq:corGinUEZerothFirstPrelim}
    \widehat{\mathcal{Z}}_{\vec\alpha}^{\rm (Gin)}(\chi)
     = \frac{e^{k|\chi|^2}}{\pi^k l!^2}  \sum_{m=0}^l |\chi|^{2m} \sum_{\vec{\nu} \vdash (l-m),\ l(\vec\nu)\leq k} \frac{1}{\widehat{\nu}!}
\left[ \sum_{\vec{\lambda} \vdash l,\ l(\vec\lambda)\leq k} f_\lambda^2 \, t_{\vec{\nu}}(\widehat{\lambda}) \right]^2\ .
\end{equation}
With this expression, it is not obvious how to calculate an anlytic continuation in the parameters $k$ and $l$ or $h$.
Our aim is to allow at least analytic continuation in $k$ at fixed $l$.
As $l(\vec\nu) \leq l$ in \eqref{eq:corGinUEZerothFirstPrelim}, the length limit on these partitions is not a real $k$ dependence,
but the lengths of shifted sequences like $\widehat{\nu}_j=\nu_{k-j+1} + j - 1$ depend on $k$.
The lower $h$ entries of the shifted parts have the simple form $\widehat{\nu}_j = j-1$ for $1\leq j \leq h$.
Writing $\widehat{\nu}!$ explicitly, we see that it is convenient to separate the lower and the upper entries,
\begin{equation}
    \label{eq:LowerFactorialsToG}
    \widehat{\nu}! = \prod_{j=1}^k (\nu_{k-j+1} + j - 1)! = \Big( \prod_{j=0}^{h-1} j! \Big) \prod_{j=1}^l (\nu_{l-j+1} + h + j - 1)!
        = G(h+1) \prod_{j=1}^l (\nu_j + k - j)!\ ,
\end{equation}
where $G$ is the Barnes $G$ function, and in the last equality we also replaced $j$ by $l-j+1$.
For the factorial Schur polynomial, we can use that also $l(\vec\lambda) \leq l$, so $\widehat{\nu}_j = \widehat{\lambda}_j = j-1$ for $1\leq j \leq h$.
Therefore in the definition of the factorial Schur polynomials~\eqref{eq:DefineFactorialSchur}, the numerator determinant expression is lower-triangular
in its upper $h$ rows. On the corresponding $h$ diagonal entries, we find $0!,1!,\ldots,(h-1)!$, so we simplify as follows,
\begin{equation}
\begin{aligned}
    \label{eq:GinUEFirstFacSchSimplPart}
    t_{\nu}(\widehat{\lambda}) 
        &= \frac{1}{\Delta_k(\widehat{\lambda})} \det\left[\frac{(\lambda_{k-a+1} + a - 1)!}{(\lambda_{k-a+1} + a - \nu_{k-b+1} - b)!}\right]_{a,b=1,\dots,k} \\
        &= \frac{1}{\Delta_k(\widehat{\lambda})} \det\left[\frac{(\lambda_{l-a+1} + h + a - 1)!}{(\lambda_{l-a+1} + a - \nu_{l-b+1} - b)!}\right]_{a,b=1,\dots,l} \prod_{j=0}^{h-1} j!\ .
\end{aligned}
\end{equation}
We split the occurring Vandermonde determinant into three parts, a Vandermonde determinant of the final $l$ arguments, a two-index product of differences between the final $l$
and the initial $h$ arguments, and a Vandermonde determinant of the initial $h$ arguments. The last of these parts cancels the product extracted in~\eqref{eq:GinUEFirstFacSchSimplPart},
leaving us with
\begin{equation}
\begin{aligned}
    \label{eq:GinUEFirstFacSchSimpl}
    t_{\nu}(\widehat{\lambda})
        &= \frac{\det\left[\frac{(\lambda_{l-a+1} + h + a - 1)!}{(\lambda_{l-a+1} + a - \nu_{l-b+1} - b)!}\right]_{a,b=1,\dots,l}}{\Delta_l(\lambda_{k-h}+h,\ldots,\lambda_{1}+k-1) \prod_{j=1}^l \prod_{i=0}^{h-1} (\lambda_{l-j+1} + h+j-1 - i)}\\
        &= \frac{\det\left[\frac{(\lambda_{l-a+1} + h + a - 1)!}{(\lambda_{l-a+1} + a - \nu_{l-b+1} - b)!}\right]_{a,b=1,\dots,l}}{\Delta_l(\lambda_{k-h}+h,\ldots,\lambda_{1}+k-1)}
            \prod_{j=1}^l \frac{(\lambda_{l-j+1} + j-1)!}{(\lambda_{l-j+1} + h+j-1)!}\\
        &= \frac{\det\left[\frac{(\lambda_{l-a+1} + a - 1)!}{(\lambda_{l-a+1} + a - \nu_{l-b+1} - b)!}\right]_{a,b=1,\dots,l}}{\Delta_l(\lambda_l,\ldots,\lambda_{1}+l-1)}\ .
\end{aligned}
\end{equation}
In the last equality we inserted the factors of the product into the rows of the numerator determinant, and subtracted $h$ from each argument of the Vandermonde determinant.
We notice that the last line of~\eqref{eq:GinUEFirstFacSchSimpl} is the same as the definition of a factorial Schur polynomial of the $l$-shifted sequence (as opposed to $k$-shifted) of the partition $\vec\lambda$.
Therefore, for the remainder of this subsection, we change the semantics of $\widehat{\lambda}$, using $l$ in place of $k$ in~\eqref{lambda.hat.rel}, i.e.\ $\widehat{\lambda}_j = \lambda_{l-j+1}+j-1$, and likewise $\widehat{\nu}_j = \nu_{l-j+1}+j-1$.
Alongside comes the explicit change in the length restrictions on the partitions to $l$,
as we substitute~\eqref{eq:GinUEFirstFacSchSimpl} and~\eqref{eq:LowerFactorialsToG} in~\eqref{eq:corGinUEZerothFirstPrelim},

\begin{equation}
    \label{eq:corGinUEZerothFirst}
    \widehat{\mathcal{Z}}_{\vec\alpha}^{\rm (Gin)}(\chi)
     = \frac{e^{k|\chi|^2}}{\pi^k l!^2 G(h+1)} \sum_{m=0}^l |\chi|^{2m} \sum_{\vec{\nu} \vdash (l-m),\ l(\vec\nu)\leq l} \frac{1}{\prod_{j=1}^l (\nu_j+k-j)!}
\left[ \sum_{\vec{\lambda} \vdash l,\ l(\vec\lambda)\leq l} f_\lambda^2 \, t_{\vec{\nu}}(\widehat{\lambda}) \right]^2\ ,
\end{equation}
leading to~\eqref{eq:corGinUEZerothFirstExplicit} except for the simplification of the highest and next-to-highest coefficients.

For the highest coefficient, it is $\vec{\nu} = \vec{0}$ and $t_{(0)}(\widehat{\lambda}) = 1$.
This follows from the fact that the numerator in definition~\eqref{eq:DefineFactorialSchur} of factorial Schur polynomials becomes a Vandermonde determinant for $\widehat{\nu}= (0,1,\dots,l-1)$.
So the sum over $\vec{\lambda} \vdash l$ in~\eqref{eq:corGinUEZerothFirst} becomes a sum over the squares of the dimensions of irreducible representations of $S_l$ which is equal to $l!$ by standard representation theory.
Therefore, the $m=l$ term is equal to $e^{k|\chi|^2}/(\pi^k\prod_{j=1}^{k-1}j!)$, where the denominator product results from multiplying $G(h+1)$ by the
$k$-dependent product in \eqref{eq:corGinUEZerothFirst}.

For the next-to-highest coefficient, we will show that $t_{(1)}(\widehat{\lambda}) = l$.
From the definition~\eqref{eq:DefineFactorialSchur} we have
\begin{equation}
    t_{(1)}(\widehat{\lambda})  = \frac{1}{\Delta_k(\widehat{\lambda})}
        \det\begin{bmatrix}
            1 & \widehat{\lambda}_1 & \widehat{\lambda}_1 (\widehat{\lambda}_1 - 1) & \dots & \frac{\widehat{\lambda}_1!}{(\widehat{\lambda}_1 - (l-2))!} & \frac{\widehat{\lambda}_1!}{(\widehat{\lambda}_1 - l)!} \\
            1 & \vdots & \vdots & \ddots & \vdots & \vdots \\
            1 & \widehat{\lambda}_l & \widehat{\lambda}_l (\widehat{\lambda}_l - 1) & \dots & \frac{\widehat{\lambda}_l!}{(\widehat{\lambda}_l - (l-2))!} & \frac{\widehat{\lambda}_l!}{(\widehat{\lambda}_l - l)!}
        \end{bmatrix} 
\end{equation}
By linear combination of the first $l-1$ columns we can replace the matrix entries to the monomials $\widehat{\lambda}_a^{b-1}$ for $b=1,\dots,l-1$ by column operations.
However, the last column can not be reduced to a simple power but contains $\widehat{\lambda}_i^l$ and $\widehat{\lambda}_i^{l-1}$ so we have
\begin{equation}
\begin{aligned}
    t_{(1)}(\widehat{\lambda}) & = \frac{1}{\Delta_l(\widehat{\lambda})}
        \det\begin{bmatrix}
            1 & \widehat{\lambda}_1 & \widehat{\lambda}_1^2 & \dots & \widehat{\lambda}_1^{l-2} & \widehat{\lambda}_1^l - \frac{1}{2}l(l-1) \widehat{\lambda}_1^{l-1} \\
            1 & \vdots & \vdots & \ddots & \vdots & \vdots \\
            1 & \widehat{\lambda}_l & \widehat{\lambda}_l^2 & \dots & \widehat{\lambda}_l^{l-2} & \widehat{\lambda}_l^l - \frac{1}{2}l(l-1) \widehat{\lambda}_l^{l-1}
        \end{bmatrix} \\
        & = \frac{1}{\Delta_l(\widehat{\lambda})} \left(\det
        \begin{bmatrix}
            1 & \widehat{\lambda}_1 & \widehat{\lambda}_1^2 & \dots & \widehat{\lambda}_1^{l-2} & \widehat{\lambda}_1^l \\
            1 & \vdots & \vdots & \ddots & \vdots & \vdots \\
            1 & \widehat{\lambda}_l & \widehat{\lambda}_l^2 & \dots & \widehat{\lambda}_l^{l-2} & \widehat{\lambda}_l^l
        \end{bmatrix}
        - \frac{1}{2}l(l-1) \Delta_l(\widehat{\lambda}) \right) \\
        & = \frac{1}{\Delta_l(\widehat{\lambda})} \left(\Delta_l(\widehat{\lambda}) \sum_{j=1}^l \widehat{\lambda}_j
        - \frac{1}{2}l(l-1) \Delta_l(\widehat{\lambda}) \right)  = l.
\end{aligned}
\end{equation}
The second equality follows from the multilinearity of the determinant.
The first term in the next-to-last equality can be understood from the fact that the determinant is an antisymmetric polynomial of $\widehat{\lambda}_i$ and therefore factors into $\Delta_l(\widehat{\lambda})$ times a symmetric polynomial.
The symmetric polynomial must be a monic linear function which is unique and given by the trace $\sum_{j=1}^l \widehat{\lambda}_j=|\vec\lambda|+l(l-1)/2$ with $|\vec\lambda|=l$ for $\vec\lambda\vdash l$, and $l(l-1)/2$ the sum of the shifts which distinguish $\widehat{\lambda}$ and $\vec{\lambda}$.
In total, the $m=l-1$ term is $ l^2|\chi|^{2(l-1)}e^{k|\chi|^2} / [\pi^k\,k!\prod_{j=0}^{k-2} j! ]$ which is shown in terms of the Barnes $G$ function in~\eqref{eq:corGinUEZerothFirstExplicit}.

\subsection{One higher-order derivative}
\label{apx: ex one deriv GinUE}

In the special case of a single $n$th derivative, we have $\vec{\alpha} = (n,0,0,\dots)$ and the Kostka number $K_{\vec{\lambda},\vec{\alpha}}$ vanishes unless $\vec{\lambda}=(n)$, in which case $\widehat{\lambda}=(0,1,2,3,\dots,k-2,n+k-1)$.
Therefore, equation~\eqref{eq:thmGinUE} collapses to
\begin{equation}
    \begin{aligned}
    \widehat{\mathcal{Z}}_{\vec\alpha}^{\rm (Gin)}(\chi)
     = \frac{n!^2 e^{k|\chi|^2}}{\pi^k} \sum_{m=0}^{n} |\chi|^{2m} \sum_{\substack{\vec{\nu} \vdash n-m,\ l(\vec{\nu}) \leq k}} 
    \frac{1}{\widehat{\nu}!}\left(\det\left[\frac{1}{(\widehat{\lambda}_a - \widehat{\nu}_b)!}\right]_{a,b=1,\ldots,k} \right)^2\ .
    \end{aligned}
\end{equation}
The determinant vanishes unless $\vec{\nu}$ is the one-part partition $\vec{\nu} = (n-m)$ with $\widehat{\nu}=(0,1,2,3,\dots,k-2,n-m+k-1)$ because of $\widehat{\lambda}_j=j-1$ for all $j=1,\ldots,k-1$. Therefore,
\begin{equation}
    \begin{aligned}
    \widehat{\mathcal{Z}}_{\vec\alpha}^{\rm (Gin)}(\chi)
     = \frac{n!^2 e^{k|\chi|^2}}{\pi^k} \sum_{m=0}^{n} |\chi|^{2m}  
    \frac{1}{(n-m+k-1)!\prod_{j=0}^{k-2}j!} \left(\det\left[\begin{array}{c|c}\frac{1}{(a-b)!} & \frac{1}{(a-n+m-k)!}\\\hline  \frac{1}{(n+k-b)!} & \frac{1}{m!} \end{array}\right]_{a,b=1}^{k-1} \right)^2 .
    \end{aligned}
\end{equation}
The matrix inside the determinant is a lower-triangular matrix with diagonal entries $(0!,\dots,0!,\frac{1}{m!})$, and therefore it is
\begin{equation}
    \label{eq:GinUEOneHigherSum}
    \widehat{\mathcal{Z}}_{\vec\alpha}^{\rm (Gin)}(\chi)
     = \frac{n!^2 e^{k|\chi|^2}}{\pi^k\prod_{j=0}^{k-2}j!} \sum_{m=0}^{n}\frac{|\chi|^{2m}  
}{(n-m+k-1)!m!^2}\ .
\end{equation}
For $k=1$, this result is proportional to a Laguerre polynomial. More generally,
we employ the truncated Laguerre polynomial~\eqref{Lag.def}. Then, \eqref{eq:GinUEOneHigherSum} can be written as in~\eqref{high.der.Gin}.

\subsection{Two higher-order derivatives}
\label{apx: ex two derivs GinUE}

In the case of a single $n_1$st derivative and $n_2$nd derivative, we have $\vec{\alpha}=(n_1,n_2,0,\dots)$. Then, $K_{\vec{\lambda},\vec{\alpha}}$ vanishes unless $l(\vec{\lambda}) \leq 2$ and $\lambda_2 \leq n_2$.
In other words, the partition lies in a one parameter family of partitions $\vec{\lambda}(s) = (n_1+n_2-s, s)$ with $s\in\{0,\dots,n_2\}$ and the shifted partition becomes $\widehat{\lambda}(s) = (0,1,2,\dots,s+k-2, n_1+n_2-s+k-1)$, meaning~\eqref{eq:thmGinUE} becomes
\begin{equation}
\begin{aligned}
    \widehat{\mathcal{Z}}_{\vec\alpha}^{\rm (Gin)}(\chi)= \frac{n_1!^2 n_2!^2 e^{k|\chi|^2}}{\pi^k} \sum_{m=0}^{n_1+n_2} |\chi|^{2m}
    \sum_{\substack{\vec{\nu} \vdash n_1+n_2-m \\l(\vec{\nu}) \leq k}}\!\! \widehat{\nu}!\left( \sum_{s=0}^{n_2}\frac{
    \det\left[\binom{\widehat{\lambda}_a}{\widehat{\nu}_b}\right]_{a,b=1,\ldots,k} }{(n_1+n_2-s+k-1)!(s+k-2)!\prod_{i=0}^{k-3}i!}\right)^2\!\! .
\end{aligned}
\end{equation}
As before, the determinant vanishes unless $l(\vec{\nu})\leq 2$. Therefore, also the partition $\vec{\nu}$ can be parametrised as $\vec{\nu}(r) = (n_1+n_2-m-r,r)$ for $r\in \{0,\dots,\floor{\frac{1}{2}(n_1+n_2-m)}\}$, where we use the floor function $\lfloor.\rfloor$.
This results in the double sum
\begin{equation}
\begin{aligned}
    \widehat{\mathcal{Z}}_{\vec\alpha}^{\rm (Gin)}(\chi)=& n_1!^2 n_2!^2 \Big(\frac{e^{|\chi|^2}}{\pi}\Big)^k \sum_{m=0}^{n_1+n_2} |\chi|^{2m}
    \sum_{r=0}^{\floor{\frac{1}{2}(n_1+n_2-m)}} \frac{(n_1+n_2-m-r+k-1)!(r+k-2)!} {\prod_{j=0}^{k-3}j!} \\
    &\times \left(
    \sum_{s=0}^{n_2} \frac{\mathfrak{X}_m^{(k,n_1,n_2)}(r,s)}{(n_1+n_2-s+k-1)! (s+k-2)!}
        \right)^2\ ,
\end{aligned}
\end{equation}
with the $k\times k$ determinant
\begin{equation}
    \mathfrak{X}_m^{(k,n_1,n_2)}(r,s) = \det
    \left[\begin{array}{c|c|c}
        \binom{a-1}{b-1} & 0 & 0 \\ \hline
        \binom{s+k-2}{b-1}  & \binom{s+k-2}{r+k-2} & \binom{s+k-2}{n_1+n_2-m-r+k-1} \\\hline
        \binom{n_1+n_2-s+k-1}{b-1}  & \binom{n_1+n_2-s+k-1}{r+k-2} & \binom{n_1+n_2-s+k-1}{n_1+n_2-m-r+k-1}
    \end{array}\right]_{a,b=1,\ldots k-2} .
\end{equation}
The upper left $(k-2) \times (k-2)$ subdeterminant is $1$.
The lower right $2 \times 2$ subdeterminant can be computed explicitly to get
\begin{equation}
\begin{aligned}
     \widehat{\mathcal{Z}}_{\vec\alpha}^{\rm (Gin)}(\chi)=&  \frac{n_1!^2 n_2!^2 e^{k|\chi|^2}}{\pi^k} \sum_{m=0}^{n_1+n_2} |\chi|^{2m}
    \sum_{r=0}^{\floor{\frac{1}{2}(n_1+n_2-m)}} \frac{(n_1+n_2-m-r+k-1)!(r+k-2)!} {\prod_{i=0}^{k-3}i!} \\
    &\times \Big(
    \sum_{s=0}^{n_2} \frac{
\binom{s+k-2}{r+k-2}\binom{n_1+n_2-s+k-1}{n_1+n_2-m-r+k-1} - \binom{s+k-2}{n_1+n_2-m-r+k-1} \binom{n_1+n_2-s+k-1}{r+k-2} 
}{(n_1+n_2-s+k-1)! (s+k-2)!}
        \Big)^2 ,
\end{aligned}
\end{equation}
which is equal to~\eqref{2nd.high.der.Gin}.

\end{appendix}

\end{document}